\documentclass{lmcs}
\pdfoutput=1
\usepackage[utf8]{inputenc}

\usepackage{lastpage}
\lmcsdoi{18}{3}{37}
\lmcsheading{}{\pageref{LastPage}}{}{}%
{Apr.~01,~2021}{Sep.~21,~2022}{}

\usepackage[toc,page]{appendix}
\usepackage{stackrel}
\usepackage{xspace,mathpartir,xparse}
\usepackage{cmll,thmtools,thm-restate}
\usepackage{ebmath,ebutf8}

\usepackage{hyperref}
\usepackage{fullshort}
\setboolean{fullpaper}{true}
\newlist{enumerati}{enumerate}{10}
\setlist[enumerati]{label=\emph{(\roman*)}, ref=\emph{(\roman*)}}
\newlist{enumeratiwide}{enumerate}{10}
\setlist[enumeratiwide]{label=\emph{(\roman*)}, ref=\emph{(\roman*)}, wide, labelindent=0pt}
\newlist{enumeratawide}{enumerate}{10}
\setlist[enumeratawide]{label=\emph{(\alph*)}, ref=\emph{(\alph*)}, wide, labelindent=0pt}
\newlist{enumeratewide}{enumerate}{10}
\setlist[enumeratewide]{label=\emph{\arabic*.}, ref=\emph{(\arabic*)}, wide, labelindent=0pt, before = \leavevmode}

\usepackage{environ}

\newcommand{\transitionmonoidof}[1]{transition {$#1$}-monoid\xspace}

\newcommand{\transitionmonoidsof}[1]{transition {$#1$}-monoids\xspace}

\newcommand{\transitionmonoid}{transition monoid\xspace}

\newcommand{\transitionmonoids}{transition monoids\xspace}

\newcommand{\transitionmonoidalgebras}{transition monoid algebras\xspace}

\newtheorem{terminology}[thm]{Terminology}

\usepackage{tikz-cats}
\usetikzlibrary{matrix,cd}
\tikzcdset{arrow style=tikz}
\tikzcdset{every arrow/.append style = -{Computer Modern Rightarrow[]}}
\tikzset{every arrow/.append style = -{Computer Modern Rightarrow[]}}

\renewcommand{\diaggrandhauteur}{.6}
\renewcommand{\diaggrandlargeur}{1}

\newcommand{\psh}[1][ℂ]{\widehat{#1}}
\newcommand{\xinto}[1]{\xarrow[into]{#1}}
\newcommand{\xto}[1]{\xrightarrow{#1}}
\newcommand{\xot}[1]{\xleftarrow{#1}}

\newcommand{\restr}[2]{{#1}_{{|}#2}}
\newcommand{\ens}[1]{\{ #1 \}}
\newcommand{\wbotright}[1]{{#1}^⋔}
\newcommand{\wbotleft}[1]{{{}^⋔{#1}}}
\newcommand{\wbotrightleft}[1]{\wbotleft{(\wbotright{#1})}}
\newcommand{\source}{𝐬}
\newcommand{\but}{𝐭}
\newcommand{\op}[1]{{#1}^{\mathit{op}}}

\DeclareMathOperator{\alg}{-\mathbf{alg}}
\DeclareMathOperator{\algv}{-\mathbf{alg}_{\mathrm{v}}}
\DeclareMathOperator{\dom}{dom}
\DeclareMathOperator{\Mod}{-\mathbf{Mod}}

\DeclareMathOperator{\Trans}{-\mathbf{Trans}}

\DeclareMathOperator{\Mon}{-\mathbf{Mon}}
\DeclareMathOperator{\mon}{-\mathbf{Mon}}
\DeclareMathOperator{\id}{id}
\DeclareMathOperator{\el}{el}
\DeclareMathOperator{\colim}{colim}

\DeclareMathOperator{\Lan}{Lan}
\newcommand{\Span}{\mathit{Span}}
\newcommand{\Bispan}{\mathit{2Sp}}

\newcommand{\std}{\mathit{std}}
\newcommand{\lax}{\mathit{lax}}
\newcommand{\pack}{\mathit{pack}}

\begin{document}

\author[T.\ Hirschowitz]{Tom Hirschowitz\lmcsorcid{0000-0002-7220-4067}}[a] %
\address{Univ. Grenoble Alpes, Univ. Savoie Mont Blanc, CNRS, LAMA, 73000, Chambéry, France}

\author[A.\ Lafont]{Ambroise Lafont\lmcsorcid{0000-0002-9299-641X}}[b]
\address{UNSW, Sydney, Australia}

\title[A categorical framework for congruence of applicative
  bisimilarity]{A categorical framework for congruence of applicative
  bisimilarity in higher-order languages}

\keywords{Programming languages; categorical semantics; operational semantics;
  Howe's method}
\begin{abstract}
  Applicative bisimilarity is a coinductive characterisation of
  observational equivalence in call-by-name lambda-calculus,
  introduced by Abramsky (1990). Howe (1996) gave a direct proof that
  it is a congruence, and generalised the result to all languages
  complying with a suitable format.  We propose a categorical
  framework for specifying operational semantics, in which we prove
  that (an abstract analogue of) applicative bisimilarity is
  automatically a congruence. Example instances include standard
  applicative bisimilarity in call-by-name, call-by-value, and
  call-by-name non-deterministic $λ$-calculus, and more generally all
  languages complying with a variant of Howe's format.
\keywords{operational semantics \and category theory \and bisimilarity
  \and congruence \and Howe's method.}
\end{abstract}

\maketitle              

\section{Introduction}\label{s:intro}
\subsection{Motivation}
This paper is a contribution to the search for efficient and
high-level mathematical tools to specify and reason about programming
languages.  This search arguably goes back at least to Turi and
Plotkin~\cite{plotkin:turi:bialgebraic}, who coined the name
“Mathematical Operational Semantics”, and proved a general congruence
theorem for bisimilarity.  This approach has been deeply investigated,
notably for quantitative languages~\cite{Bartels}.  However, as of
today, attempts to apply it to higher-order (e.g., functional)
languages have failed.

In previous
work~\cite{hirschowitz:hal-01815328,hirschowitz:hal-02273790}, the
first author has proposed an alternative approach to the problem,
dropping the coalgebraic notion of bisimulation used by Turi and
Plotkin in favour of a notion based on factorisation systems, similar
to Joyal et al.'s~\cite{DBLP:conf/lics/JoyalNW93}.  Furthermore,
congruence of bisimilarity is notably obtained by assuming that syntax
induces a \emph{familial} monad~\cite{Diers1978Spectres,DBLP:journals/mscs/CarboniJ95,Weber:famfun}.

However, the new approach has only been applied to simple, first-order
languages like the
$π$-calculus~\cite{Milner:pi,DBLP:books/daglib/0004377}, and Positive
GSOS specifications~\cite{GSOS}.  In this paper, we extend it to
functional languages, notably covering the paradigmatic case of
\emph{applicative} bisimilarity~\cite{Applicative} in call-by-name and
call-by-value $λ$-calculus, as well as in a simple, non-deterministic
$λ$-calculus~\cite[§7]{DavideLazy,DBLP:journals/iandc/Howe96}.  We
even show that our framework subsumes the general, syntactic format
proposed by Howe~\cite[Lemma 6.1]{DBLP:journals/iandc/Howe96}.  We thus obtain
for the first time a generic, categorical congruence result for
applicative bisimilarity in functional languages.

\subsection{Overview}
A bit more precisely:
\begin{itemize}
\item We propose a simple notion of signature for programming
  languages.
\item Each signature has a category of models, including an initial
  one, intuitively its operational semantics.
\item An abstract analogue of applicative bisimilarity, called
  \emph{substitution-closed} bisimilarity, may be defined in any
  model, and in particular in the initial one.
\item Under suitable hypotheses, we show that substitution-closed
  bisimilarity is a congruence in the initial model.
\end{itemize}

Categorically, this unfolds as follows.
\begin{enumeratiwide}
\item \label{item:lts} We define an abstract notion of (labelled)
  transition systems, as objects of a category $𝐂$, in such a way that
  \begin{itemize}
  \item there is a forgetful functor $𝐂 → 𝐂₀$, intuitively returning
    the (potentially structured) set of states of a transition system;
  \item bisimulation and bisimilarity may be defined for any
    transition system.
\end{itemize}
\item Adopting Fiore, Plotkin, and Turi's seminal
  framework~\cite{fiore:presheaf,DBLP:conf/lics/Fiore08}, we then
  assume that $𝐂₀$ is monoidal, and define models of the syntax to be
  \emph{monoid algebras} for a given \emph{pointed strong} endofunctor
  $Σ₀$ on $𝐂₀$.  Monoid algebras, a.k.a.\ $Σ₀$-monoids, are
  $Σ₀$-algebras equipped with compatible monoid structure, which
  models capture-avoiding substitution.  The category $Σ₀\mon$ of
  $Σ₀$-monoids has an initial object $𝐙₀$, whose carrier is the free
  $Σ₀$-algebra on the monoidal unit $I$, as we prove in
  Coq~\cite{BHLcode}. In all examples, $𝐙₀$ is precisely the syntax.
\item This category $Σ₀\mon$ induces by pullback a category $Σ₀\Trans$
  of transition systems whose states are equipped with $Σ₀$-monoid
  structure. We call these \emph{\transitionmonoidalgebras}, or
  \emph{\transitionmonoidsof{Σ₀}}, or even simply
  \emph{\transitionmonoids} when $Σ₀$ is the identity. The relevant
  notions of bisimulation and bisimilarity for such objects are
  defined as in~\ref{item:lts}, but for \emph{substitution-closed}
  relations.
\item We then define models of the dynamics to be certain algebras,
  called \emph{vertical}, for an endofunctor on $Σ₀\Trans$. There is an
  initial vertical algebra $𝐙$, which in examples is the syntactic
  transition system. (Furthermore, standard applicative bisimilarity coincides
  with substitution-closed bisimilarity.)
\item Finally, following an abstract analogue of Howe's method, we
  show that, under suitable hypotheses, subs\-ti\-tu\-tion-closed
  bisimilarity on $𝐙$ is a congruence. One crucial hypothesis is
  \emph{cellularity}, in a sense closely related
  to~\cite{garner:hal-01246365}.
\end{enumeratiwide}

\subsection{Related work}
Plotkin and Turi's \emph{bialgebraic
  semantics}~\cite{plotkin:turi:bialgebraic} and its few
variants~\cite{DBLP:journals/tcs/CorradiniHM02,DBLP:conf/lics/Staton08}
prove abstract congruence theorems for bisimilarity.  However, they do
not cover higher-order languages like the $λ$-calculus, let alone
applicative bisimilarity. This was one of the main motivations for our
work.  Among more recent work, quite some inspiration was drawn from
Ahrens et al.~\cite{AHLM19,HHL}, notably in the use of vertical
algebras. However, a difference is that we do not insist that
transitions be stable under substitution.  In a different direction,
Dal Lago et al.~\cite{DBLP:conf/lics/LagoGL17} prove a general
congruence theorem for applicative bisimilarity, for a $λ$-calculus
with algebraic effects. As briefly discussed in the conclusion, our
framework does not yet account for such results. However, it places
the generality in a different direction: namely, it is not tied to any
particular language (like the $λ$-calculus
in~\cite{DBLP:conf/lics/LagoGL17}).  It would of course be useful to
find a common generalisation.

Links with other relevant work by, e.g., Bodin et
al.~\cite{DBLP:journals/pacmpl/BodinGJS19}, though desirable, remain
unclear, perhaps because of the very different methods used.

Furthermore, the cellularity used here is close to but different from
the $𝐓ₛ^∨$-familiality of~\cite{hirschowitz:hal-01815328}. It would be
instructive to better understand potential links between the two.
Finally, let us mention recent work which, just like ours, strives to
establish abstract versions of standard constructions and theorems in
programming language theory like type
soundness~\cite{DBLP:conf/lics/ArkorF20} or
gluing~\cite{Fiore02,DBLP:conf/lics/FioreS20,DBLP:conf/fossacs/FioreS20}.

\subsection{Relation to conference version}
This paper is a bit more than a journal version of our previous
work~\cite{BHL}. Here is a brief summary of changes.

\begin{enumeratawide}
\item In~\cite{BHL}, we work with a non-trivial generalisation of
  monoid algebras to \emph{skew monoidal} categories~\cite{Szlachanyi}
  and \emph{structurally strong} functors.  Here, by giving a better
  type to $Σ₁$, the endofunctor for specifying the dynamics, we manage
  to work with standard monoid algebras. This has the additional
  advantages of
  \begin{itemize}
  \item avoiding a slightly \emph{ad hoc} compositionality assumption
    of~\cite{BHL}, and
  \item relaxing the requirement that the tensor product should be
    familial.
\end{itemize}
\item In~\cite{BHL}, because Howe's closure operates only at
  the level of states, we work mostly with
  \emph{prebisimulations}, in the sense
  of~\cite[§5.1]{hirschowitz:hal-01815328}.  This notion is designed
  to detect when the state part of a relation underlies a
  bisimulation, regardless of what it does on transitions.  However,
  it feels more \emph{ad hoc} than the standard definition of
  bisimulation by lifting~\cite{DBLP:conf/lics/JoyalNW93}.  In this
  paper, we extend Howe's closure to transitions, thus avoiding
  prebisimulations entirely.
\item In~\cite{BHL}, we rely on directed unions of relations, which
  leads to quite a few, rather painful proofs by induction.  Here, we
  use higher-level methods to construct Howe's closure, essentially
  through categorification and algebraicisation.  Namely:
  \begin{enumerate}
  \item We define bisimilarity as the final object not in some
    partially-ordered set of relations as usual, but in some category
    of spans (see also~\cite{DBLP:conf/calco/BasoldPR17}).
  \item Furthermore, we define Howe's closure directly as a free
    monoid algebra for a suitable pointed strong endofunctor on spans.  
  \item More generally, we systematically rely on universal
    properties, which simplifies a significant number of proofs.
  \end{enumerate}
\item We put less emphasis on cellularity, viewing it only as a
  sufficient condition for a perhaps more natural hypothesis, which
  already appeared in a slightly different form
  in~\cite{DBLP:conf/lics/Staton08}, namely the fact that $Σ₁$
  preserves functional bisimulations.
\item We obtain a congruence theorem of similar scope
  (Theorem~\ref{thm:main}), and cover three new, detailed applications
  (§\ref{s:applications}): call-by-value, big-step $λ$-calculus (which
  was covered but too naively in~\cite{BHL}, as we explain), a
  call-by-name $λ$-calculus with unary, erratic choice
  from~\cite[§7]{DavideLazy}, and a general format proposed by
  Howe~\cite[Lemma 6.1]{DBLP:journals/iandc/Howe96}.
\item We fill a gap in the proof of~\cite[Lemma~5.13]{BHL}, by
  requiring the endofunctor $Σ₀$ for specifying the dynamics to
  preserve sifted colimits (see Remark~\ref{rk:sifted}).
\end{enumeratawide}

\subsection{Plan}
In~§\ref{s:prereq}, we start by briefly recalling call-by-name
$λ$-calculus and applicative bisimilarity. We then explain how to view
the latter as substitution-closed bisimilarity, and sketch Howe's
method.  In~§\ref{s:overview}, we then give a brief overview of the
new framework by example, including a recap on monoid algebras and a
statement of the main theorem (in the considered case).  We then dive
into the technical core of the paper by presenting our framework for
transition systems and bisimilarity (§\ref{s:trans}), operational
semantics (§\ref{s:howecontexts}), and then substitution-closed
bisimilarity and the main result (Theorem~\ref{thm:main}), together
with a high-level proof sketch (§\ref{s:substclosedbisim}).
In~§\ref{s:cellularity}, we reformulate the main hypothesis of
Theorem~\ref{thm:main} using cellularity, which allows us to use
well-known results from weak factorisation systems as sufficient
conditions.  We then apply our results to examples
in~§\ref{s:applications}.  \iffull{ The full proof of
  Theorem~\ref{thm:main} is given in~§\ref{s:congruence}.} Finally, we
conclude and give some perspectives on future work in~§\ref{s:conclu}.

\subsection{Notation and preliminaries}\label{ss:notation}
In this subsection, we fix some basic notation, and review some
preliminaries.
\subsubsection{Basic notation}
We often conflate natural numbers $n ∈ ℕ$ with the corresponding sets
$\ens{1,…,n}$.  For all sets $X$ and objects $C$ of a given category,
we denote by $X · C$ the $X$-fold coproduct of $C$ with itself, i.e.,
$∑_{x ∈ X} C$.  Let $𝐆𝐩𝐡$ denote the category of (directed, multi)
graphs, $𝐂𝐚𝐭$ the category of small categories, and $𝐂𝐀𝐓$ the category
of locally small categories.

\subsubsection{Comma categories and lax limits}
Given functors $F∶ 𝐀 → 𝐂$ and $G∶ 𝐁 → 𝐂$,
the \emph{comma} category $F/G$ has
\begin{itemize}
\item as objects all triples $(A,B,ϕ)$, where $A ∈ 𝐀$, $B ∈ 𝐁$,
  and $ϕ∶ F(A) → G(B)$, and 
\item as morphisms $(A,B,ϕ) → (A',B',ϕ')$ all pairs of morphisms
  $u∶ A → A'$ and $v∶ B → B'$ making the following square commute.
  \begin{center}
    \diag{%
      F(A) \& F(A') \\
      G(B) \& G(B') %
    }{%
      (m-1-1) edge[labela={F(u)}] (m-1-2) %
      edge[labell={ϕ}] (m-2-1) %
      (m-2-1) edge[labelb={G(v)}] (m-2-2) %
      (m-1-2) edge[labelr={ϕ'}] (m-2-2) %
    }
  \end{center}
\end{itemize}
We have the following well-known fact:
\begin{prop}
  \label{prop:projcreates}
  If $𝐀$ and $𝐁$ have, and $F$ preserves colimits of any given shape,
  then the projection functor $F/G → 𝐀 × 𝐁$ creates them.

  Symmetrically, if $𝐀$ and $𝐁$ have, and $G$ preserves limits of any
  given shape, then the projection functor $F/G → 𝐀 × 𝐁$ creates them.
\end{prop}

The comma category $F/G$ is well-known~\cite{Kelly89,TwoToposes}
to be the universal category equipped with projections to
$𝐀$ and $𝐁$ and a natural transformation as in the following
diagram.
\begin{center}
  \Diag{%
    \twocell{m-1-1}{m-2-1}{m-1-1}{m-1-2}{}{cell=.1,bend right} %
  }{%
    F/G \& 𝐁 \\
    𝐀 \& 𝐂
  }{%
    (m-1-1) edge[labela={}] (m-1-2) %
    edge[labell={}] (m-2-1) %
    (m-2-1) edge[labelb={F}] (m-2-2) %
    (m-1-2) edge[labelr={G}] (m-2-2) %
  }%
  \end{center}

  Kelly~\cite{Kelly89} explains that the comma category is a kind of
  lax limit of $F$ and $G$.  When $F$ is an identity, we call the
  comma category a \emph{lax limit} of $G$.

\subsubsection{Presheaves}
Let $\psh$ denote the category of (contravariant) presheaves on $ℂ$,
and $𝐲∶ ℂ → \psh$ the Yoneda embedding, mapping $c$ to $ℂ(-, c)$.
Given a presheaf $F∈ \psh$, an element $x∈ F(c)$, and a morphism
$c \xto{f} c'$, we sometimes denote $F(f)(x)$ by $x·f$.  Given two
categories $ℂ₁$ and $ℂ₂$, we denote by $[ℂ₁, ℂ₂]$ the functor category
between them.

\subsubsection{Spans and relations}
In a category $𝐂$ with binary products, we interchangeably use spans
$X ← R → Y$ and their pairings $R → X×Y$, sometimes also calling the
latter spans.  Spans from $X$ to $Y$ are the objects of a category
$\Span(𝐂)(X,Y)$, which is isomorphic to the slice category $𝐂/X×Y$ in
the presence of binary products. When $𝐂$ has pullbacks, these
categories are the hom-categories of a bicategory~\cite{Benabou}
$\Span(𝐂)$, in which composition of morphisms is given by pullback. A
\emph{relation} from $X$ to $Y$ is merely a span whose pairing
$R → X×Y$ is monic.

\subsubsection{Images}
Let us now recall a few elements about images.
\begin{defi}
  An \emph{image} of a morphism $f∶ A→B$ is a factorisation
  $A \xto{e} M \xinto{m} B$ with $m$ a monomorphism, which is initial
  in the sense that for any factorisation $A \xto{e'} M' \xinto{m'} B$
  there is a (unique by monicness) morphism $k∶ M → M'$ making both
  triangles commute in the following diagram.
  \begin{center}
    \diag{%
      A \& M' \\
      M \& B
    }{%
      (m-1-1) edge[labela={e'}] (m-1-2) %
      edge[labell={e}] (m-2-1) %
      (m-2-1) edge[labelb={m}] (m-2-2) %
      edge[dashed,labelal={k}] (m-1-2) %
      (m-1-2) edge[labelr={m'}] (m-2-2) %
    }
  \end{center}
\end{defi}
\begin{defi}
  A \emph{strong epimorphism} is a morphism with the strong left
  lifting property w.r.t.\ all monomorphisms, i.e., a morphism
  $e∶ A → B$ such that for all (solid) commuting squares
  \begin{center}
    \diag{%
      A \& X \\
      B \& Y %
    }{%
      (m-1-1) edge[labela={}] (m-1-2) %
      edge[labell={e}] (m-2-1) %
      (m-2-1) edge[labelb={}] (m-2-2) %
      edge[dashed,labelal={k}] (m-1-2) %
      (m-1-2) edge[labelr={m}] (m-2-2) %
    }
  \end{center}
  with $m$ monic there exists a unique lifting $k$ making both
  triangles commute.
\end{defi}
The terminology is justified by the following result.
\begin{lem}
  In a category with equalisers, any strong epimorphism is an
  epimorphism.
\end{lem}
\begin{proof}
  Let us assume that $e∶ A → B$ is a strong epi and $fe = ge$, with
  $f,g∶ B → C$. Then, let $k∶ A' → B$ denote the equaliser of $f$ and
  $g$. Because $e$ equalises $f$ and $g$, it factors as $kh$, for some
  unique $h∶ A → A'$. But now $k$ is monic, so by lifting there is a
  unique $l$ making both triangles commute in the following diagram.
  \begin{center}
    \diag{%
      A \& A' \\
      B \& B
    }{%
      (m-1-1) edge[labela={h}] (m-1-2) %
      edge[labell={e}] (m-2-1) %
      (m-2-1) edge[identity,labelb={}] (m-2-2) %
      edge[dashed,labelal={l}] (m-1-2) %
      (m-1-2) edge[labelr={k}] (m-2-2) %
    }
  \end{center}
  We thus have \[f = f ∘ \id_B = f ∘ k ∘ l = g ∘ k ∘ l = g ∘ \id_B = g.\]
  The morphism $e$ is thus epi, as claimed.
\end{proof}

\begin{cor}
Factoring a morphism as a strong epi followed by a mono
yields an image.
\end{cor}
\begin{proof}
  Initiality is directly given by the lifting property.
\end{proof}
\begin{prop}
  In any locally finitely presentable category, images always exist,
  and may be computed as (strong epi, mono)-factorisations.
\end{prop}
\begin{proof}
  This is (part of)~\cite[Proposition~1.61]{Adamek}.
\end{proof}
Let us finally observe:
\begin{prop}\label{prop:union}
  In locally presentable category, unions of subobjects
  exist, and may be computed by taking the cotupling of all considered
  subobjects, and then their (strong epi, mono)-factorisation.
\end{prop}
\begin{proof}
  Straightforward.
\end{proof}

\subsubsection{Initial algebras}
Any finitary endofunctor $F$ on any cocomplete category admits
by~\cite[Theorem~2.1]{Reiterman} an initial algebra, which we denote
by $𝐙_F$.  Although this is detailed below, we prefer to avoid
confusion and warn the reader that we also use $𝐙_F$ for the initial
$F$-monoid, for any \emph{pointed strong} endofunctor $F$ on any nice
monoidal category (which is incidentally the initial $(I+F)$-algebra).
Throughout the paper, when not explicitly attached to any $F$, $𝐙$ is
shorthand for $𝐙_{\check{Σ}₁}$ (see, e.g., Proposition~\ref{prop:Z} or
Theorem~\ref{thm:Z}).

\subsubsection{Weak factorisation systems}
Finally, let us fix some notation about weak factorisation systems.
In any category $𝒞$, we say that a morphism $f∶ A → B$ has the (weak)
left lifting property w.r.t.\ $g∶ C → D$ when for all commuting
squares 
\begin{center}
  \diaginline(.6,.6){%
    A \& C \\
    B \& D\rlap{,} %
  }{%
    (m-1-1) edge[labela={u}] (m-1-2) %
    edge[labell={f}] (m-2-1) %
    (m-2-1) edge[labelb={v}] (m-2-2) %
    edge[dashed,labelal={k}] (m-1-2) %
    (m-1-2) edge[labelr={g}] (m-2-2) %
  }
\end{center}
there is a lifting $k$ as shown that makes both triangles commute.
Equivalently, we say that $g$ has the right lifting property w.r.t.\
$f$, and write $f \mathrel{⋔} g$.  Given a fixed set $𝕁$ of morphisms,
the set of morphisms $g$ such that $j \mathrel{⋔} g$ for all $j ∈ 𝕁$
is denoted by $𝕁^⋔$. Similarly, the set of morphisms $f$ such that
$f \mathrel{⋔} j$ for all $j ∈ 𝕁$ is denoted by $\wbotleft{𝕁}$.  In
particular, if $f ∈ \wbotrightleft{𝕁}$ and $g ∈ \wbotright{𝕁}$, then
$f \mathrel{⋔} g$.  If $𝒞$ is locally presentable~\cite{Adamek}, then
$(\wbotrightleft{𝕁},𝕁^⋔)$ forms a \emph{weak factorisation system}, in
the sense that additionally any morphism $f∶ X → Y$ factors as
$X \xto{l} Z \xto{r} Y$ with $l ∈ \wbotrightleft{𝕁}$ and $r ∈ 𝕁^⋔$
(see~\cite[Theorem~2.1.14]{Hovey}).  Morphisms in $𝕁^⋔$ are
generically called \emph{fibrations}, while morphisms in
$\wbotrightleft{𝕁}$ are called \emph{cofibrations}.

Let us conclude with the following easy, yet helpful result.
\begin{lem}\label{lem:colimitsim}
  For any locally finitely presentable category and set $𝕁$ of maps
  therein, if the domains and codomains of maps in $𝕁$ are finitely
  presentable, then fibrations are closed under filtered colimits in
  the arrow category.
\end{lem}
\begin{proof}
  Let us consider any given filtered diagram $(fᵢ∶ Aᵢ → Bᵢ)_{i ∈ 𝔻}$
  of fibrations, and a colimit $f∶ A → B$ in the arrow category, say
  $𝐂^{→}$.  We must show that $j \mathrel{⋔} f$ for all $j ∈ 𝕁$. Let
  us thus  consider any given commuting square
  \begin{center}
    \diag{%
      X \& A \\
      Y \& B %
    }{%
      (m-1-1) edge[labela={u}] (m-1-2) %
      edge[labell={j}] (m-2-1) %
      (m-2-1) edge[labelb={v}] (m-2-2) %
      (m-1-2) edge[labelr={f}] (m-2-2) %
    }
  \end{center}
  with $j ∈ 𝕁$.  Colimits in the arrow category are pointwise, so
  $A = \colim ᵢ Aᵢ$ and $B = \colim ᵢ Bᵢ$. Thus, by finite
  presentability of $X$ and $Y$, and by filteredness of the diagram
  $𝔻 → 𝐂$, $u$ and $v$ factor through some $A_{i₀}$ and $B_{i₁}$,
  respectively. By filteredness of the diagram again, w.l.o.g., we
  may take $i₀ = i₁$, such that $(u,v)∶ j → f$ factors through
  $f_{i₀}$.
  But because $f_{i₀}$ is a fibration, we find a lifting as
  in
  \begin{center}
    \diag{%
      X \& A_{i₀} \& A \\
      Y \& B_{i₀} \& B\rlap{,} %
    }{%
      (m-1-1) edge[labela={u'}] (m-1-2) %
      edge[labell={j}] (m-2-1) %
      edge[bend left,labela={u}] (m-1-3) %
      (m-2-1) edge[labela={v'}] (m-2-2) %
      edge[dashed,labelal={}] (m-1-2) %
      edge[bend right,labelb={v}] (m-2-3) %
      (m-1-2) edge[labelr={f_{i₀}}] (m-2-2) %
      (m-1-2) edge[labela={}] (m-1-3) %
      (m-2-2) edge[labelb={}] (m-2-3) %
      (m-1-3) edge[labelr={f}] (m-2-3) %
    }
  \end{center}
  which provides the desired lifting for the original square.
\end{proof}

\section{A brief review of Howe's method}\label{s:prereq}
\subsection{Applicative bisimilarity}
Let us consider the standard, big-step presentation of call-by-name
$λ$-calculus:
\begin{mathpar}
  \inferrule{ }{λx.e ⇓ λx.e}
  \and
  \inferrule{e₁ ⇓ λx.e'₁ \\ e'₁[x↦e₂] ⇓ e₃}{e₁\ e₂ ⇓ e₃}
\end{mathpar}
Standardly, the evaluation relation $⇓$ is considered between closed
terms only.

Applicative bisimilarity is an important notion of program equivalence
in this language. Indeed, it is coinductive, so one may prove that any
two given programs are applicative bisimilar merely by exhibiting an
applicative bisimulation. Furthermore, it is sound and complete
w.r.t.\ (i.e., it coincides with) standard contextual equivalence.

Let us briefly recall the definition.  Applicative bisimilarity is
standardly introduced in two stages, which we now recall.
\begin{defiC}[{\cite[Definition~2.3]{Applicative}}]
  A relation $R$ over closed $λ$-terms is an \emph{applicative
    simulation} iff $e₁ \mathrel{R} e₂$ and $e₁ ⇓ λx.e'₁$ entail
  the existence of $e'₂$ such that $e₂ ⇓ λx.e'₂$ and, for all terms
  $e$, $e'₁[x↦e] \mathrel{R} e'₂[x↦e]$.

  An \emph{applicative bisimulation} is an applicative simulation $R$
  whose converse, say $R^†$, is also an applicative simulation.
\end{defiC}
Applicative bisimulations are closed under unions, and so there is a
largest applicative bisimulation, called \emph{applicative
  bisimilarity} and denoted by $∼$.

Then comes the second stage:
\begin{defi}
  The \emph{open extension} of a relation $R$ on closed terms is the
  relation $R^⊗$ on potentially open terms such that
  $e \mathrel{R^⊗} e'$ iff for all closed substitutions $σ$ covering
  all involved free variables we have $e[σ] \mathrel{R} e'[σ]$.
\end{defi}

Let us readily notice the following alternative characterisation of
open extension.
\begin{defi}
  A relation $S$ on open terms is \emph{substitution-closed} iff for
  all $e \mathrel{S} e'$ and (potentially open) substitutions $σ$, we
  have $e[σ] \mathrel{S} e'[σ]$.
\end{defi}
\begin{lem}\label{lem:bisimequiv}
  The open extension of any relation $R$ is the greatest
  \emph{substitution-closed} relation contained in $R$ on closed
  terms.
\end{lem}
\begin{proof}
  Let us first show that $R^⊗$ is substitution-closed.  For any
  $e₁ \mathrel{R^⊗} e₂$ and $σ$, we want to show
  $e₁[σ] \mathrel{R^⊗} e₂[σ]$.  For this, we in turn need to show that
  for all closing substitutions $γ$, we have
  $e₁[σ][γ] \mathrel{R} e₂[σ][γ]$.  But $eᵢ[σ][γ] = eᵢ[σ[γ]]$, where
  by definition and $σ[γ](x) = σ(x)[γ]$. Furthermore, $σ[γ]$ is closing. So,
  because we have $e₁ \mathrel{R^⊗} e₂$, by definition of open extension, we
  get $e₁[σ][γ] \mathrel{R} e₂[σ][γ]$ as desired.

  Let us now show that $R^⊗$ is the greatest substitution-closed
  relation contained in $R$ on closed terms. For this, consider any
  substitution-closed $R'$ contained in $R$ on closed terms: for all
  $e \mathrel{R'} e'$, by substitution-closedness, we have
  $e[σ] \mathrel{R'} e'[σ]$ for all closing $σ$.  So because $R'$ is
  contained in $R$ on closed terms, we further have
  $e[σ] \mathrel{R} e'[σ]$. This proves $e \mathrel{R^⊗} e'$, and 
  thus $R' ⊆ R^⊗$ as desired.
\end{proof}

The result we wish to abstract over is the following
(see~\cite{Pitts:howe} for a historical account).
\begin{thm}\label{thm:wow}
  The open extension $∼^⊗$ of applicative bisimilarity is a congruence:
  it is an equivalence relation, and furthermore it is context-closed, i.e.,
  \begin{itemize}
  \item   $e₁ \mathrel{∼^⊗} e₂$ entails $λx.e₁ \mathrel{∼^⊗} λx.e₂$ for all $x$;
  \item $e₁ \mathrel{∼^⊗} e₂$ and $e'₁ \mathrel{∼^⊗} e'₂$ entail
    $e₁\ e'₁ \mathrel{∼^⊗} e₂\ e'₂$.
  \end{itemize}
\end{thm}
Proving that $∼^⊗$ is an equivalence relation is in fact straightforward. 
In the following, we focus on the context-closedness property. 
\subsection{Howe's method}
Howe's method for proving Theorem~\ref{thm:wow} consists in
considering a suitable relation $∼^•$, closed under substitution and
context, and containing $∼^⊗$ by construction. He then shows that this
relation $∼^•$ is an applicative bisimulation. By maximality of $∼^⊗$,
we thus also have ${∼^•} ⊆ {∼^⊗}$ hence both relations coincide and
$∼^⊗$ is context-closed as desired.  However, as explained
in~\cite[§5.1]{BHL}, the presence of a substitution in the premises of
a transition rule seems to require ${∼^•}$ to be closed under
\emph{heterogeneous} substitution, in the sense that, e.g., if
$e₁ ∼^• e'₁$ and $e₂ ∼^• e'₂$ (for open terms), then
$e₁[x↦e₂] ∼^• e'₁[x↦e'₂]$.  The problem is that building this into the
definition of $∼^•$ leads to difficulties in the proof that it is an
applicative bisimulation. Howe's workaround consists in requiring $∼^•$ to be
closed under sequential composition with $∼^⊗$ from the outset.
Coupling this right action with context closedness, he thus defines
$∼^•$ as the smallest context-closed relation satisfying the rules
  \begin{mathpar}
    \inferrule{ }{x ∼^• x} \and
    \inferrule{e ∼^• e' \\ e' ∼^⊗ e''}{e ∼^• e''}~·
  \end{mathpar}
  By construction, $∼^•$ is reflexive and context-closed.  By 
  induction, it also substitution-closed. Furthermore, by reflexivity
  and the second rule, it also contains $∼^⊗$, and finally the second
  rule clearly entails ${∼^•};{∼^⊗} ⊆ {∼^•}$, where ``$;$'' denotes
  relational (or sequential) composition.  It takes an induction to
  prove stability under heterogeneous substitution, but to give a feel
  for it, in the basic case where $e₁ ∼^⊗ e'₁$, we have
\[e₁[x↦e₂] ∼^• e₁[x↦e'₂] ∼^⊗ e'₁[x↦e'₂]\] by context closedness of $∼^•$
and substitution closedness of $∼^⊗$, so we conclude by
${∼^•};{∼^⊗} ⊆ {∼^•}$.

The initial plan was to show that $∼^•$ is an applicative bisimulation
and deduce that it coincides with $∼^⊗$. It can in fact be slightly
optimised by first showing that $∼^•$ is an applicative simulation,
and then that its transitive closure $(∼^•)⁺$ is
symmetric.
The relation $(∼^•)⁺$ is also an applicative
simulation, hence by symmetry an applicative bisimulation.  This entails
the last inclusion in the chain ${∼^⊗} ⊆ {∼^•} ⊆ {(∼^•)⁺} ⊆ {∼^⊗},$
showing that all relations coincide. Finally, because $∼^•$ is
context-closed, so is $∼^⊗$, as desired.

\subsection{Non-standard presentation}\label{ss:nonstandard}
The above, standard evaluation rules for call-by-name $λ$-calculus are
not directly compatible with our framework.  We thus adopt a slightly
different presentation, where the evaluation relation relates closed
terms to terms with just one potential free variable.  The problem and
its solution should become clear in~§\ref{ss:wow:format}, where we
investigate Howe's general format~\cite[Lemma 6.1]{DBLP:journals/iandc/Howe96}.
There, we show that any language complying with Howe's format may be
covered by our framework, up to suitable encoding.  The present,
non-standard presentation is a slight variant of this encoding,
optimised for $λ$-calculus.
The new transition rules are as follows.
\begin{mathpar}
  \inferrule{ }{λx.e ⇓ e}
  \and
  \inferrule{e₁ ⇓ e'₁ \\ e'₁[e₂] ⇓ e₃}{e₁\ e₂ ⇓ e₃}
\end{mathpar}
Here $e'₁[e₂]$ denotes substitution of the unique potential free
variable in $e'₁$ by $e₂$.  We will see below that, with this
transition system, the essentially standard notion of bisimulation
coupled with the substitution-closedness requirement yields
applicative bisimilarity.

\section{Overview by example}\label{s:overview}
In this section, we describe one particular instance of our framework,
which models call-by-name $λ$-calculus.

\subsection{Syntax}\label{s:exsyntax}
Let us first define the syntax of $λ$-calculus,
following~\cite{fiore:presheaf}, as an initial\footnote{This pattern
  is advocated by the approach of \emph{initial algebra
    semantics}~\cite{InitialSemantics}, where initiality provides a
  recursion principle.} object in a suitable category of models.  Very
roughly, a model of $λ$-calculus syntax should be something equipped
with operations modelling abstraction and application, but also with
substitution.  Furthermore, certain natural compatibility axioms
should be satisfied, e.g.,
\begin{equation}
  (e₁\ e₂)[σ] = e₁[σ]\ e₂[σ]\rlap{.}
  \label{eq:substapp}
\end{equation}

A natural setting for specifying such operations is the functor
category $𝐂₀ ≔ [𝔽,𝐒𝐞𝐭]$, where $𝔽 ↪ 𝐒𝐞𝐭$ denotes the full subcategory
spanning all sets of the form $n$ (i.e., $\ens{1,…,n}$, recalling
notation from §\ref{s:intro}).  For any $X ∈ [𝔽,𝐒𝐞𝐭]$ and $n ∈ 𝔽$, we
think of $X(n)$ as a set of `terms' with $n$ potential free variables,
e.g., in $\ens{1,…,n}$ or, if the reader prefers, $\ens{x₁,…,xₙ}$. The
action of $X$ on morphisms $n → n'$ is thought of as variable
renaming.  Returning to operations, being equipped with abstraction is
the same as being a $Σ₀$-algebra, where $Σ₀∶ 𝐂₀ → 𝐂₀$ is defined by
$Σ₀(X)(n) = X(n+1)$.  An algebra structure on any $X$ thus consists
of a family of maps $X(n+1) → X(n)$, natural in $n$.  Similarly, for
specifying both application and abstraction, we consider
\begin{equation}
  Σ_Λ(X)(n) = X(n+1) + X(n)².
  \label{eq:lambdasyntax}
\end{equation}

Let us now consider substitution. The idea here is to equip $𝐂₀$ with
monoidal structure $(⊗,I)$, such that
\begin{itemize}
\item elements of $(X⊗Y)(n)$ are like
\emph{explicit substitutions} $x⦇σ⦈$, where $x ∈ X(p)$ and
$σ∶ p → Y(n)$ for some $p$, considered equivalent up to some standard
equations\footnote{In~\cite{BHL}, we instead considered a skew-monoidal
  variant where the tensor product does not enforce any standard equation.};
\item elements of $I(n) ≔ \ens{1,…,n}$ are merely variables.
\end{itemize}
Being equipped with substitution (and variables) is thus
the same as being a \emph{monoid} for this tensor product:
\begin{itemize}
\item the multiplication $m_X∶ X⊗X → X$ maps any formal, explicit
  substitution $x⦇σ⦈$ to an actual substitution $x[σ]$, and
\item the unit $e_X∶ I → X$ injects variables into terms.
\end{itemize}
Finally, how do we enforce equations such as~\eqref{eq:substapp}?
This goes in two stages:
\begin{itemize}
\item we first collect the way substitution is supposed to commute
  with each operation, by providing a \emph{pointed strength}, i.e., a
  natural transformation with components
  $st_{X,Y}∶ Σ_Λ(X)⊗Y → Σ_Λ(X⊗Y)$, where $X ∈ 𝐂₀$ and $Y ∈ I/𝐂₀$,
  satisfying some equations~\cite[§I.1.2]{DBLP:conf/lics/Fiore08};
\item we then use the pointed strength to enforce all equations in one
  go, by requiring models to have compatible $Σ_Λ$-algebra and
  substitution structure, in a suitable sense.
\end{itemize}
Let us first explain the notion of pointed strength.
\begin{description}
\item[Application] For modelling Equation~\eqref{eq:substapp} for
  application, we would in particular define $st_{X,Y}$ to map any
  $(in₂(x₁,x₂))⦇σ⦈$ to $in₂(x₁⦇σ⦈,x₂⦇σ⦈)$, for all $x₁,x₂ ∈ X(p)$ and
  $σ∶ p → Y(n)$. (The coproduct injection $in₂$ here acts as a formal
  application, recalling $Σ_Λ(X)(n) = X(n+1) + X(n)²$.)
\item[Abstraction] For abstraction, let us start by first stating the
  corresponding equation.  We will then define the pointed strength
  accordingly.  Supposing that $Y$ is equipped with a point
  $e_Y∶ I → Y$, we define $σ^↑∶ p+1 → Y(n+1)$ by copairing
  \begin{center}
    $p \xto{σ} Y(n) \xto{Y(in₁)} Y(n+1)$ \hfil and \hfil
    $1 = I(1) \xto{(e_Y)₁} Y(1) \xto{Y(in₂)} Y(n+1)$.
  \end{center}
  The equation is then
\begin{equation}
  \label{eq:substlambda}
  λ(e)[σ] = λ(e[σ^↑]).
\end{equation}
Accordingly, we define the pointed strength to map any $in₁(x)⦇σ⦈$,
where $x ∈ X(p+1)$ and $σ∶ p → Y(n)$, to $in₁(x⦇σ^↑⦈)$.
\end{description}
Let us now go through the second stage of how we impose the desired
equations: a model of syntax will be a monoid $X$ equipped with
$Σ_Λ$-algebra structure $ν_X∶ Σ_Λ(X) → X$, such that the following
diagram commutes.
  \begin{equation}
    \diag{%
      Σ_Λ(X)⊗X \& Σ_Λ (X⊗X) \& Σ_Λ(X) \\
      X⊗X \& \& X %
    }{%
      (m-1-1) edge[labela={st_{X,X}}] (m-1-2) %
      edge[labell={ν_X⊗X}] (m-2-1) %
      (m-1-2) edge[labela={Σ_Λ(m_X)}] (m-1-3) %
      (m-2-1) edge[labelb={m_X}] (m-2-3) %
      (m-1-3) edge[labelr={ν_X}] (m-2-3) %
    }
    \label{eq:monoidalg}
  \end{equation}
  Indeed, suppose given, e.g., $in₁(e)⦇σ⦈ ∈ Σ_Λ(X)⊗X$, by applying the
  left then bottom morphisms we obtain $λ(e)[σ]$, while applying the
  top then right morphisms we obtain $λ(e[σ^↑])$, as desired.

All in all, we have:  
\begin{defi}
  For any finitary, pointed strong endofunctor $Σ₀$, a \emph{monoid
    algebra} for $Σ₀$, or a \emph{$Σ₀$-monoid}, is a $Σ₀$-algebra
  $(X,ν_X∶ Σ₀(X) → X)$, equipped with monoid structure
  $(m_X∶ X⊗X → X, e_X∶ I → X)$, such that~\eqref{eq:monoidalg}
  commutes.  A $Σ₀$-monoid \emph{morphism} is a morphism in $𝐂₀$ which
  is both a monoid morphism and a $Σ₀$-algebra morphism.

  Let $Σ₀\mon$ denote the category of $Σ₀$-monoids and morphisms between
  them.
\end{defi}

Let us conclude by (slightly informally) stating the result exhibiting
standard syntax as the initial
model~\cite{fiore:presheaf,DBLP:conf/rta/FioreS17,BHL}. See
Proposition~\ref{prop:ptstrgeneral} below for a general and rigorous
statement.
\begin{prop}\label{prop:ptstr}
  For any finitary, pointed strong endofunctor $Σ₀$, under mild
  hypotheses, the forgetful functor $Σ₀\mon → 𝐂₀$ is monadic, and the
  free $Σ₀$-algebra over $I$ (equivalently the initial
  $(I+Σ₀)$-algebra) is an initial $Σ₀$-monoid.
\end{prop}

\begin{exa}\label{ex:syntax}
  In the case of $λ$-calculus, the initial $Σ_Λ$-monoid is thus the least
  fixed point $𝐙₀ ≔ μA.(I + Σ_Λ(A))$, which is isomorphic to the standard,
  low-level construction of syntax.
\end{exa}

From this, one may deduce a characterisation of not only the initial
$Σ_Λ$-monoid, but all \emph{free} $Σ_Λ$-monoids, or in other words an
explicit formula for the left adjoint to the forgetful functor.
Namely:
\begin{prop}\label{prop:freelambda}
  For any finitary, pointed strong endofunctor $Σ₀$, under the same
  hypotheses as in Proposition~\ref{prop:ptstr}, the free
  $Σ₀$-monoid, say $ℒ₀(K)$, over any $K ∈ 𝐂₀$ is
\[μA.(I + Σ₀(A) + K⊗A).\]
\end{prop}
Syntactically, when $Σ₀=Σ_Λ$, letting $n ⊢_K e$ mean that
$e ∈ ℒ₀(K)(n)$, $ℒ₀(K)$ is inductively generated by the following
rules~\cite[§3.1]{DBLP:conf/aplas/Hamana04},
\begin{mathpar}
  \inferrule{ }{n ⊢_K x}~(x ∈ n) \and
  \inferrule{n ⊢_K e₁ \\ … \\ n ⊢_K eₚ}{n ⊢_K k⦇e₁,…,eₚ⦈}~(k ∈ K(p)) \\
  \inferrule{n ⊢_K e₁ \\ n ⊢_K e₂}{n ⊢_K e₁\ e₂} \and
  \inferrule{n+1 ⊢_K e}{n ⊢_K λ(e)}   
\end{mathpar}
modulo the equivalence
\[(f · k)⦇e₁,…,e_q⦈ ∼ k⦇e_{f(1)},…,e_{f(p)}⦈\rlap{,}\]
for all $f∶ p → q$, $k ∈ K(p)$, and $n ⊢_K e₁,…,e_q$,
or perhaps more synthetically
\[(f · k)⦇σ⦈ ∼ k⦇σ∘f⦈\rlap{,}\]
where $σ∶ q → ℒ₀(K)(n)$ denotes the
cotupling of $e₁,…,e_q$ viewed as
maps $1 → ℒ₀(K)(n)$.

The first rule is the standard rule for variables, while the second
one is for ``constants'', i.e., elements of $K$. It corresponds to the
term $K⊗A$ in the above fixed-point formula. When $p = 0$, we
sometimes shorten the notation from $k⦇⦈$ to $k$.  The last two rules
are the standard rules for application and abstraction, and they
correspond to the term $Σ₀(A)$ in the formula. The $Σ_Λ$-monoid
structure is syntactically straightforward; notably substitution
satisfies $k⦇e₁,…,eₚ⦈[σ] = k⦇e₁[σ],…,eₚ[σ]⦈$.

\subsection{Transition systems and bisimilarity}\label{ss:lts:ex}
The appropriate notion of transition system for $λ$-calculus
is as follows.
\begin{defi}
  A \emph{transition system} $X$ consists of
  \begin{itemize}
  \item a \emph{state object}
    $X₀ ∈ 𝐂₀ = [𝔽,𝐒𝐞𝐭]$, 
  \item a set $X₁$ of \emph{transitions}, and 
  \item maps $X₀(0) \xot{s_X} X₁ \xto{t_X} X₀(1)$ giving the source and target of
    transitions (cf.~§\ref{ss:nonstandard}).
  \end{itemize}
  Transition systems form a category $𝐂$, whose morphisms $X → Y$
  consist of compatible morphisms $f₀∶ X₀ → Y₀$ and $f₁∶ X₁ → Y₁$, in the sense
  that both of the following squares commute.
  \begin{center}
    \diag{%
      X₁ \& Y₁ \\
      X₀(0) \& Y₀(0) %
    }{%
      (m-1-1) edge[labela={f₁}] (m-1-2) %
      edge[labell={s_X}] (m-2-1) %
      (m-2-1) edge[labelb={f₀}] (m-2-2) %
      (m-1-2) edge[labelr={s_Y}] (m-2-2) %
    }
    \hfil
    \diag{%
      X₁ \& Y₁ \\
      X₀(1) \& Y₀(1) %
    }{%
      (m-1-1) edge[labela={f₁}] (m-1-2) %
      edge[labell={t_X}] (m-2-1) %
      (m-2-1) edge[labelb={f₀}] (m-2-2) %
      (m-1-2) edge[labelr={t_Y}] (m-2-2) %
    }    
  \end{center}
\end{defi}
\begin{nota}
  We write $r∶ e ⇓ f$ for $r ∈ X₁$ such that $s_X(r) = e$ and
  $t_X(r) = f$.
\end{nota}
\begin{exa} \label{ex:syntactic} The \emph{syntactic} transition
  system has $𝐙₀ ∈ 𝐂₀$ from Example~\ref{ex:syntax} as state object,
  and as transitions all derivations following the transition
  rules. We will come back to this case in Proposition~\ref{prop:Z}.
\end{exa}

Our next goal is to introduce bisimulation, for which it is convenient
to characterise $𝐂$ as a presheaf category. This characterisation may
be established by abstract means, but let us describe it concretely
first. It is clear from the definition that transitions systems are
glorified graphs. And they form a presheaf category for essentially
the same reason as graphs do. Here is the base category:
\begin{defi}
  Let $𝔽[⇓]$ denote the category obtained by augmenting $𝔽$ with an
  object $⇓$, together with morphisms $0 \xot{s_⇓} {⇓} \xto{t_⇓} 1$, and
  their formal composites with non-identity morphisms from $𝔽$.
\end{defi}
More concretely:
\begin{itemize}
\item There is exactly one morphism $0 → n$ in $𝔽$ for all $n$, which
  is an identity when $n=0$, so for all $n ≠ 0$ we have a morphism
  $s_{⇓,n}∶ {⇓} → n$ making the following triangle commute.
  \begin{center}
    \diag{%
      \& {⇓} \\
      0 \& \& n %
    }{%
      (m-1-2) edge[labelal={s_⇓}] (m-2-1) %
      (m-1-2) edge[labelar={s_{⇓,n}}] (m-2-3) %
      (m-2-1) edge[labelb={!}] (m-2-3) %
    }
  \end{center}
\item There are exactly $n$ morphisms $1 → n$ in $𝔽$ for all
  $n ∉ \ens{0,1}$ (and no morphisms $1→0$), so for all such $n$ and
  $i ∈ n$ we have a morphism $t_{⇓,n,i}∶ {⇓} → n$ making the following
  triangle commute.
  \begin{center}
    \diag{%
      \& {⇓} \\
      1 \& \& n %
    }{%
      (m-1-2) edge[labelal={t_⇓}] (m-2-1) %
      (m-1-2) edge[labelar={t_{⇓,n,i}}] (m-2-3) %
      (m-2-1) edge[labelb={i}] (m-2-3) %
    }
  \end{center}
\end{itemize}

\begin{prop}
  Transition systems are isomorphic to covariant presheaves on $𝔽[⇓]$.
\end{prop}
\begin{nota}
  We often implicitly convert from transition systems to covariant
  presheaves, and conversely.
\end{nota}
\begin{proof}[Proof sketch]
  This will be proved below by abstract means, but for intuition let us
  sketch the correspondence.  Given a transition system
  $⟨s_X,t_X⟩∶ X₁ → X₀(0)×X₀(1)$, we construct a presheaf $X'$ by
  setting
  \begin{itemize}
  \item $X'(n) = X₀(n)$ for all $n ∈ 𝔽$,
  \item $X'(⇓) = X₁$,
  \item with the action of $s_⇓,t_⇓ ∈ 𝔽$ given by $s_X$ and $t_X$,
  \item inducing the action of all $s_{⇓,n}$ and $t_{⇓,n,i}$ by composition.
  \end{itemize}
  Conversely, for any presheaf $Y$, we construct a transition system $Y'$ defined as follows:
  \begin{itemize}
  \item the state object $Y'₀ ∈ [𝔽,𝐒𝐞𝐭]$ is given by restriction of
    $Y$;
  \item the set $Y'₁$ of transitions is $Y(⇓)$;
  \item and $s_{Y'}$ and $t_{Y'}$ are $Y(s_⇓)$ and $Y(t_⇓)$,
    respectively. \qedhere
\end{itemize}
\end{proof}

The correspondence yields basic, graph-like examples of transition
systems.
\begin{exa}
  \label{ex:trans-systems}
  \hfill
  \begin{enumeratawide}
  \item The representable presheaf $𝐲₀$ associated to $0 ∈ 𝔽$ has a
    single closed state $k₀$ and its renamings (i.e., $(𝐲₀)₀(n) = 1$
    for all $n$ and for transitions $(𝐲₀)₁ = ∅$).
  \item The representable presheaf $𝐲_⇓$ consists of a closed state
    $k₀$, a state $k₁$ with one free variable, their renamings, and a
    transition $e∶ k₀ ⇓ k₁$.
  \item Let $𝐲_{s_⇓}∶ 𝐲₀ → 𝐲_⇓$ denote the morphism mapping $k₀$ to
    $k₀$.
  \end{enumeratawide}
\end{exa}
  
Using these basic examples, we may define bisimulation and
bisimilarity by lifting following~\cite{DBLP:conf/lics/JoyalNW93}:
\begin{defi}\label{def:bisimlifting}
  A morphism $X → Y$ in $𝐂$ is a \emph{functional bisimulation} when
  it has the right lifting property w.r.t.\ the source map
  $𝐲_{s_⇓}∶ 𝐲₀ → 𝐲₁$.  A span $X ← R → Y$ is a \emph{simulation} when its
  left leg $R → X$ is a functional bisimulation, and a
  \emph{bisimulation} when both legs are.  
\end{defi}
\begin{rem}
  In this case, the Yoneda lemma says that $𝐂(𝐲₀,X) ≅ X₀(0)$ and
  $𝐂(𝐲_⇓,Y) ≅ Y₁$.  The right lifting property for a morphism
  $f∶ X → Y$ thus says that given any $e ∈ Y₁$ whose source
  $e · s_⇓$ is $f(x)$ for some $x ∈ X₀(0)$, there exists $e' ∈ X₁$
  such that $f(e') = e$ and $e' · s_⇓ = x$, which matches the
  usual definition of functional bisimulation. The following diagram
  might help readability.
  \begin{center}
    \diag{%
      𝐲₀ \& X \\
      𝐲_⇓ \& Y
    }{%
      (m-1-1) edge[labela={x}] (m-1-2) %
      edge[labell={s_⇓}] (m-2-1) %
      (m-2-1) edge[labelb={e}] (m-2-2) %
      edge[dashed,labelal={e'}] (m-1-2) %
      (m-1-2) edge[labelr={f}] (m-2-2) %
    }
  \end{center}
\end{rem}

\begin{defi}
  Let $𝐁𝐢𝐬𝐢𝐦(X,Y)$ denote the category of bisimulations, with span
  morphisms between them.
\end{defi}
\begin{prop}
  $𝐁𝐢𝐬𝐢𝐦(X,Y)$ has a terminal object, called \emph{bisimilarity} and
  denoted by $∼_{X,Y}$.
\end{prop}

\begin{exa}\label{ex:bisimZ}
  Bisimilarity on the syntactic transition system merely amounts
  to simultaneous convergence, because evaluation returns an open
  term, which does not have any further transition. In this case, a
  more relevant behavioural equivalence is substitution-closed
  bisimilarity, which we will define below.
\end{exa}

\subsection{Operational semantics}\label{ss:opsem}
Just as we have defined the syntax as an initial $Σ₀$-monoid
(Example~\ref{ex:syntax}), let us now define the dynamics by
initiality, again starting by finding the right notion of model.
First of all, models will be found among transition systems $X$ whose
underlying presheaf $X₀ ∈ [𝔽,𝐒𝐞𝐭]$ is a $Σ₀$-monoid. Let us give these
a name.
\begin{defi}
  A \emph{\transitionmonoidof{Σ₀}} is a transition system $X$,
  together with $Σ₀$-monoid structure on its vertices (a.k.a.\ states)
  object $X₀$.  When $Σ₀ = \id$, we call \transitionmonoidsof{Σ₀}
  simply \emph{\transitionmonoids}. (Any \transitionmonoidof{Σ₀} is
  thus in particular a \transitionmonoid.)

  Let $Σ₀\Trans$ denote the category of \transitionmonoidsof{Σ₀}, with as
  morphisms all transition system morphisms which induce $Σ₀$-monoid
  morphisms on vertices objects.
\end{defi}

The idea is to model the transition rules as an endofunctor on
\transitionmonoidsof{Σ₀}, leaving the underlying $Σ₀$-monoid untouched,
i.e., a functor making the triangle
\begin{center}
  \diag|baseline=(m-1-1.base)|(.6,.6){%
    Σ₀\Trans \& \&     Σ₀\Trans \\
    \& Σ₀\mon %
  }{%
    (m-1-1) edge[labela={Σ₁}] (m-1-3) %
    edge[labelbl={𝒟}] (m-2-2) %
    (m-1-3) edge[labelbr={𝒟}] (m-2-2) %
  }
\end{center}
commute, where $𝒟$ denotes the forgetful functor (i.e., $𝒟(X) = X₀$).

For call-by-name $λ$-calculus, the functor $Σ₁∶ Σ₀\Trans → Σ₀\Trans$
modelling the non-standard rules at the end of~§\ref{s:prereq} is
defined as follows.
\begin{itemize}
\item On states, commutation of the above triangle imposes
  $Σ₁(X)₀ = X₀$.
\item On transitions, let  \[Σ₁(X)₁ = X₀(1) + Aᵦ(X)\rlap{,}\]
 where $Aᵦ(X)$ denotes the set of valid premises for the
  second rule in~§\ref{ss:nonstandard}, i.e., triples $(r₁,e₂,r₂)$
  such that
  \begin{itemize}
  \item $r₁,r₂ ∈ X₁$ are transitions,
  \item $e₂ ∈ X₀(0)$ is a state, and
  \item $r₂·s_⇓ = (r₁·t_⇓)[e₂]$, i.e., the source $r₂ · s_⇓$ of $r₂$ is
    obtained by substituting $e₂$ for the unique free variable in the target of
    $r₁$.
  \end{itemize}
  In other words:
  \begin{mathpar}
    \inferrule{
      \inferrule{r₁}{e₁ ⇓ e'₁} \\
      \inferrule{r₂}{e'₁[e₂] ⇓ e₃} \\
    }{
      e₁\ e₂ ⇓ e₃
    }~·
  \end{mathpar}
  Let us notice that substitution here follows from the monoid
  structure of $X$.
\item We then define the source and target maps:
  \begin{itemize}
  \item for the first term $X₀(1)$,
    \begin{itemize}
    \item the source of any $in₁(e)$ is $λ₁(e)$, where
      $λₙ∶ X₀(n+1) → X₀(n)$ follows from the $Σ₀$-algebra
      structure of $X₀$;
    \item the target is $e$ itself;
    \end{itemize}
  \item for the second term $Aᵦ(X)$,
    \begin{itemize}
    \item the source of any $in₂(r₁,e₂,r₂)$ is $(r₁·s_⇓)\ e₂$, i.e., the
      application of the source of $r₁$ to $e₂$ (again using the
      $Σ₀$-algebra structure of $X₀$);
    \item the target is $r₂ · t_⇓$.
    \end{itemize}
  \end{itemize}
\end{itemize}

Accordingly, our notion of model is the following.
\begin{defi}\label{def:verticalg:ex}
  A \emph{vertical $Σ₁$-algebra} is a \transitionmonoidof{Σ₀} $X$
  equipped with a morphism $ν_X∶ Σ₁(X) → X$ such that
  $𝒟(ν_X) = \id_{X₀}$, or equivalently a map $(ν_X)₁$ making the following
  triangle commute.
  \begin{center}
    \diag|baseline=(m-1-1.base)|(.6,.6){%
      Σ₁(X)₁ \& \& X₁ \\
      \& X₀(0)×X₀(1) %
    }{%
      (m-1-1) edge[labela={(ν_X)₁}] (m-1-3) %
      edge[labelbl={⟨s_{Σ₁(X)},t_{Σ₁(X)}⟩}] (m-2-2) %
      (m-1-3) edge[labelbr={⟨s_X,t_X⟩}] (m-2-2) %
    }
  \end{center}
\end{defi}    
In the case of call-by-name $λ$-calculus, it should be clear that such
a vertical algebra is indeed a model of the rules.

However, in order to ensure that the rules are syntax-directed, we
want to distinguish, for each rule, the head operator of the source of
the conclusion (abstraction for the first rule; application for the
second one).  Instead of demanding that $(Σ₁^∂)_X$ have the form
$Σ₁(X)₁ → X₀(0)×X₀(1)$, we thus rather require something of the form
$Σ₁(X)₁ → Σ₀(X₀)(0)×X₀(1)$:
\begin{defi}[Dynamic signatures and vertical
  algebras]\label{def:dynsig}\
  \begin{itemize}
  \item A \emph{dynamic signature} $Σ₁$ consists of
    \begin{itemize}
    \item a finitary functor $Σ₁^F∶ Σ₀\Trans → 𝐒𝐞𝐭$, and
    \item a natural transformation $(Σ₁^∂)_X∶ Σ₁^F(X) → Σ₀(X₀)(0) × X₀(1)$.
    \end{itemize}
  \item The endofunctor $\check{Σ}₁$ \emph{induced} by a dynamic
    signature maps any $X$ to the composite
    $Σ₁^F(X) \xto{(Σ₁^∂)_X} Σ₀(X₀)(0) × X₀(1) \xto{ν_{X₀,0} × X₀(1)}
    X₀(0)×X₀(1)$, where $ν_{X₀}∶ Σ₀(X₀) → X₀$ denotes the $Σ₀$-algebra
    structure of $X₀$.
\item A \emph{vertical algebra} of a dynamic signature is a vertical
  algebra of the induced endofunctor, in the sense of
  Definition~\ref{def:verticalg:ex}.
\end{itemize}
\end{defi}
Concretely, a vertical algebra is a dashed map making the following
diagram commute.
\begin{center}
  \diag|baseline=(m-1-1.base)|(.6,1.8){%
    Σ₁^F(X) \& X₁ \\
    Σ₀(X₀)(0)×X₀(1)       \& X₀(0)×X₀(1) %
  }{%
    (m-1-1) edge[dashed,labela={}] (m-1-2) %
    edge[labell={(Σ₁^∂)_X}] (m-2-1) %
    (m-2-1) edge[labelb={ν_{X₀,0} × X₀(1)}] (m-2-2) %
    (m-1-2) edge[labelr={⟨s_X,t_X⟩}] (m-2-2) %
  }
\end{center}
\begin{exa}\label{ex:Sigmaun}
  For call-by-name $λ$-calculus, we only need to modify the source
  components of the above definition of $Σ₁$, replacing actual
  operations by formal ones, like so:
\begin{itemize}
\item the source of any $in₁(e) ∈ X₀(1) + Aᵦ(X)$ is
  $in₁(e) ∈ Σ₀(X₀)(0) = X₀(1) + X₀(0)²$;
\item the source of any $in₂(r₁,e₂,r₂)$ is $in₂((r₁·s_⇓), e₂) ∈ Σ₀(X₀)(0)$.
\end{itemize}
\end{exa}

This successfully captures the syntactic transition system:
\begin{prop}\label{prop:Z}
  The initial $\check{Σ}₁$-algebra $𝐙_{\check{Σ}₁}$, or $𝐙$ for short,
  is an initial vertical algebra, and is isomorphic to the transition
  system of Example~\ref{ex:syntactic}.
\end{prop}

\subsection{Substitution-closed bisimilarity}\label{ss:substclosed:ex}
There is an obvious notion of bisimulation for \transitionmonoidsof{Σ₀}:
\begin{defi}
  \label{def:Sigma0-Mon-bisim}
  A morphism is $Σ₀\Trans$ is a \emph{functional bisimulation} iff its
  underlying morphism in $𝐂$ is.
\end{defi}
However, as foreshadowed by Example~\ref{ex:bisimZ}, the relevant
notion in this case combines bisimulation with
substitution-closedness, in the following sense.
\begin{defi}
  For any monoid $M$ in $𝐂₀$, an \emph{$M$-module} is an object $X$
  equipped with algebra structure $X⊗M → X$ for the monad ${-}⊗M$. A
  module morphism is an algebra morphism.
\end{defi}

\begin{exa}
  The monoid $M$ is itself an $M$-module by multiplication, and
  $M$-modules are closed under limits in $𝐂₀$, so in particular $M²$
  is an $M$-module, with action given by the composite
  $M²⊗M \xto{⟨π₁⊗M,π₂⊗M⟩} (M⊗M)² \xto{m_M²} M².$
\end{exa}

\begin{defi}
  For any \transitionmonoid $M$, a span of the form $R → M²$ in $𝐂$ is
  \emph{substitution-closed} iff $R₀$ may be equipped with $M₀$-module
  structure making the morphism $R₀ → M₀²$ into an $M₀$-module
  morphism.
\end{defi}

\begin{exa}
  To see what this definition has to do with substitution-closedness,
  let us observe that if $R$ is a relation in $[𝔽,𝐒𝐞𝐭]$, an element of
  $(R⊗M)(n)$ is an explicit substitution $r⦇σ⦈$ with $r ∈ R(p)$ for
  some $p$, and $σ∶ p → M(n)$. Now, substitution-closedness amounts to
  a morphism $m∶ R⊗M → R$ commuting with projections, so if $R$ is a
  relation, then $r$ is merely a pair $e \mathrel{R} e'$, and the
  morphism $m$ ensures that $e[σ] \mathrel{R} e'[σ]$.
\end{exa}

\begin{prop}
  For any \transitionmonoidof{Σ₀} $M$, there is a terminal
  subs\-ti\-tu\-tion-closed bisimulation $∼_M^⊗$, called \emph{substitution-closed
    bisimilarity}.
\end{prop}
\begin{proof}
  See Proposition~\ref{prop:scbisim} for a proof in the general case.
\end{proof}
\begin{rem}
  One may prove that substitution-closed bisimilarity is a relation.
\end{rem}
\begin{prop}
  Substitution-closed bisimilarity $∼_𝐙^⊗$ on the syntactic transition
  system $𝐙$ coincides with applicative bisimilarity.
\end{prop}
\begin{proof}
  Let us denote the open extension of applicative bisimilarity by
  $∼^⊗_{\std}$, and recall that applicative bisimilarity is denoted by
  $∼$.  Using Lemma~\ref{lem:bisimequiv}, $∼^⊗_{\std}$ is
  straightforwardly a substitution-closed bisimulation, so we have
  ${∼^⊗_{\std}} ⊆ {∼^⊗_𝐙}$ .  But conversely any substitution-closed
  bisimulation relation $R$ (hence $∼^⊗_𝐙$) is in particular a
  substitution-closed relation contained in $∼$ on closed terms. It is
  thus globally contained in $∼^⊗_{\std}$ by
  Lemma~\ref{lem:bisimequiv}.
\end{proof}

Finally, our main result instantiates to the following.
\begin{thm}\label{thm:main:ex}
  Substitution-closed bisimilarity is context-closed. More precisely, it
  is a \transitionmonoidof{Σ₀}, and ${∼^⊗_𝐙} → 𝐙²$ is a 
  \transitionmonoidof{Σ₀} morphism.
\end{thm}
In particular, there exists a span morphism $Σ₀((∼^⊗_{𝐙})₀) → (∼^⊗_{𝐙})₀$.

\section{Transition systems and bisimilarity}\label{s:trans}
In this section, we start to abstract over the situation
of~§\ref{s:overview}, by introducing a general framework for
transition systems and bisimilarity.  In~§\ref{ss:prehowe}, we first
introduce the ambient setting for this, \emph{pre-Howe contexts}, and
construct a category of transition systems, for any pre-Howe context.
Then, in~§\ref{ss:transpresh}, we show that transition systems form a
presheaf category.  We then exploit this in~§\ref{ss:bisimex} to
define bisimulation and bisimilarity.

\subsection{Pre-Howe contexts and transition systems}\label{ss:prehowe}
\begin{defi}
  A \emph{pre-Howe context}\footnote{The Howe contexts of~\cite{BHL}
    may be defined similarly. The difference is that for them,
    $\source$ and $\but$ are not necessarily functorial, but
    $c₁ ↦ (\source(c₁),\but(c₁))$ defines a functor
    $ ℂ₁ → ℂ₀\underline{×} ℂ₀$, where $ℂ₀\underline{×} ℂ₀$ denotes the
    category whose objects are pairs of elements of $ℂ₀$, and where a
    morphism $(a₁,a₂) → (b₁,b₂)$ consists of some indices
    $i,j∈\{1,2\}$, together with a pair of morphisms $a₁ → bᵢ$
    and $a₂ → bⱼ$.} consists of
  \begin{itemize}
  \item a small category $ℂ₀$ of \emph{state types},
  \item a small category $ℂ₁$ of \emph{transition types}, and
  \item two \emph{source} and \emph{target} functors $\source,\but∶ ℂ₁ → ℂ₀$.
  \end{itemize}
  Precomposition by $\source$ and $\but$ yields functors
  $Δ_\source,Δ_{\but}∶ \psh[ℂ₀] → \psh[ℂ₁]$ mapping any $X ∈ \psh[ℂ₀]$
  to $X ∘ \source$ and $X ∘ \but$, respectively.  Let $Δ$ denote the
  pointwise product $Δ_\source × Δ_{\but}$.
\end{defi}
Intuitively, objects of $ℂ₀$ may be thought of as typings (typically
sequents $A₁,…,Aₙ ⊢ A$, or merely natural numbers, as
in~§\ref{s:overview}), and objects of $ℂ₁$ as transition types. The
source and target functors associate to every transition type the
corresponding typings.  We now use these functors to define transition
systems.

\begin{defi}
  Given any pre-Howe context, a \emph{transition system} $X$
  consists of
  \begin{itemize}
  \item a \emph{state} presheaf $X₀ ∈ \psh[ℂ₀]$,
  \item a \emph{transition} presheaf $X₁ ∈ \psh[ℂ₁]$, and
  \item two \emph{source} and \emph{target} natural
    transformations $X₀ ∘ \source ← X₁ → X₀ ∘ \but$, or equivalently a
    natural transformation $X₁ → Δ(X₀)$.
  \end{itemize}
\end{defi}

\begin{prop}
  In any pre-Howe context, transition systems are precisely the objects of
  the lax limit category $\psh[ℂ₁]/Δ$ of the functor
  $\psh[ℂ₀] \xto{Δ} \psh[ℂ₁]$ in $𝐂𝐀𝐓$, or equivalently the comma
  category $\id_{\psh[ℂ₁]} / Δ$.
\end{prop}
\begin{proof}
  An object of the lax limit is by construction a triple $(X₁,X₀,∂)$,
  where $∂∶ X₁ → Δ(X₀) = X₀\source×X₀\but$.
\end{proof}

\begin{nota}
  In any pre-Howe context, we let $𝐂 ≔ \psh[ℂ₁]/Δ$, and
  denote the projections by
  $pr₁∶ 𝐂 → \psh[ℂ₁]$ and $pr₀∶ 𝐂 → \psh[ℂ₀]$, respectively.
\end{nota}

\begin{prop}
  \label{prop:proj-adjoint}
  The projection functor $-₀∶ 𝐂 → \psh[ℂ₀]$ has a left adjoint mapping any
  object $X₀$ to $∅ → X₀\source×X₀\but$, where $∅$ denotes the initial
  presheaf on $ℂ₁$.  For any $X₀$, we call this object the
  \emph{discrete} transition system on $X₀$.
\end{prop}
\begin{proof}
  Straightforward.
\end{proof}

\begin{exa}
  We can get $𝐂 → \psh[ℂ₀]$ to be the forgetful functor $𝐆𝐩𝐡 → 𝐒𝐞𝐭$ by
  taking
  \begin{itemize}
  \item $ℂ₀ = 1$, so that $\psh[ℂ₀] = \psh[1] ≅ 𝐒𝐞𝐭$,
  \item $ℂ₁ = 1$, so that $𝐂 = 𝐒𝐞𝐭$, and
  \item $\source,\but∶ 1 → 1$ to be the unique such functor, i.e., the
    identity.
\end{itemize}
A transition system thus consists of sets $V$ and $E$ together with
a map $E → V²$, i.e., a graph.
\end{exa}

\begin{exa}
  A proof-relevant variant of standard labelled transition
  systems (over any set $𝔸$ of labels) may be obtained as follows. We take
  \begin{itemize}
  \item $ℂ₀ = 1$ again,
  \item $ℂ₁ = 𝔸$ viewed as a discrete category, and
  \item $\source,\but∶ ℂ₁ → ℂ₀$ the unique such functor.
  \end{itemize}
  Thus, a transition system $X$ consists of a set $X₀$ and sets $Xₐ$ for
  all $a ∈ 𝔸$, together with maps $Xₐ → X₀²$ returning the source
  and target of each $a$-labelled edge.

  More generally, given any graph $𝕃$, taking $\source,\but∶ ℂ₁ → ℂ₀$ to be the
  source and target maps $𝕃₁ → 𝕃₀$ viewed as functors between discrete
  categories, we obtain for $𝐂 → \psh[ℂ₀]$ a functor equivalent to
  $𝐆𝐩𝐡/𝕃 → 𝐒𝐞𝐭/𝕃₀$.
\end{exa}

\begin{exa}
  Let $ℂ₀ = \op{𝔽}$ and $ℂ₁ = 1$, with $\source$ and $\but$ picking
  respectively $0$ and $1$. In particular, $\psh[ℂ₁] ≅ 𝐒𝐞𝐭$.  Then,
  $Δ(X₀) = X₀(0)×X₀(1)$ and we recover the category
  $𝐂$ of~§\ref{ss:lts:ex}, and its forgetful functor to
  $\psh[ℂ₀] = [𝔽,𝐒𝐞𝐭]$.
\end{exa}

\subsection{Transition systems as presheaves}\label{ss:transpresh}
Before introducing bisimulation, let us establish an alternative
characterisation of the category $𝐂$ of transition systems.

\begin{prop}\label{prop:isom}
  The lax limit category $\psh[ℂ₁]/Δ$ of transition systems is
  isomorphic to a presheaf category $\psh[ℂ_{\source,\but}]$.
\end{prop}
\begin{proof}
  Let $ℂ_{\source,\but}$ denote the lax colimit in $𝐂𝐚𝐭$ of the parallel pair
  $\source,\but$. By definition, it is the universal category
  with functors and natural transformations as in
  \begin{center}
    \Diag(1.2,1){%
    \twocell[.55]{m-2-2}{m-1-3}{m-2-2}{m-1-1}{}{cell=0.2,bend right=10,labelon={s}} %
    \twocell[.45]{m-2-2}{m-1-3}{m-2-2}{m-1-1}{}{cell=0.1,bend left=10,labelon={t}} %
  }{%
      ℂ₁ \& \& ℂ₀ \\
      \& ℂ_{\source,\but}\rlap{.} %
    }{%
      (m-1-1) edge[bend left=10,labela={\source}] (m-1-3) %
       edge[bend right=10,labelb={\but}] (m-1-3) %
      edge[labelbl={in₁}] (m-2-2) %
      (m-1-3) edge[labelbr={in₀}] (m-2-2) %
    }
  \end{center}
  It thus consists of the coproduct $ℂ₁ + ℂ₀$, augmented with arrows
  $s_L∶ \source(L) → L$ and $t_L∶ \but(L) → L$ for all $L ∈ ℂ₁$,
  naturally in $L$. Presheaves on $ℂ_{\source,\but}$ coincide with
  $\psh[ℂ₁]/Δ$ because the presheaf construction turns lax colimits
  into lax limits.
\end{proof}

\begin{nota}
  We often omit the isomorphism $\psh[ℂ₁]/Δ ≅ \psh[ℂ_{\source,\but}]$,
  considering it as an implicit coercion. E.g., for any $P ∈ ℂ₀$,
  $𝐲_P$ may be used to denote the transition system $\underline{P}$
  with $\underline{P}₁ = ∅$ and $\underline{P}₀ = 𝐲_P$.

  Similarly, $𝐲_L$ may be used to denote the `minimal' transition
  system with one transition over $L$, say $\underline{L}$, i.e.,
  $\underline{L}₁ = 𝐲_L$,
  $\underline{L}₀ = 𝐲_{\source(L)} + 𝐲_{\but(L)}$, and the map
  $\underline{L}₁ → \underline{L}₀\source×\underline{L}₀\but$ uniquely
  determined by the element
  $(in₁(\id_{\source(L)}),in₂(\id_{\but(L)})) ∈ \underline{L}₀(\source(L)) ×
  \underline{L}₀(\but(L))$.

  Finally, $𝐲_{s_L}∶ 𝐲_{\source(L)} → 𝐲_L$ and $𝐲_{t_L}∶ 𝐲_{\but(L)} → 𝐲_L$
  denote the Yoneda embedding of the canonical morphisms $s_L$ and
  $t_L$ from the proof of Proposition~\ref{prop:isom}.
\end{nota}

By Yoneda, we thus have:
\begin{cor}
  For all $X$, we have $𝐂(𝐲_L,X) ≅ X₁(L)$ and $𝐂(𝐲_P,X) ≅ X₀(P)$.
\end{cor}

\begin{nota}\label{not:Downarrow}
  In the case of call-by-name $λ$-calculus, we call $⇓$ the unique
  object coming from $ℂ₁ = 1$.
\end{nota}

\begin{rem}
  Presheaves on $ℂ_{\source,\but}$ intuitively have two dimensions,
  $0$ and $1$; the projection functor forgets dimension 1, while the
  left adjoint (Proposition~\ref{prop:proj-adjoint}) adds an empty
  dimension 1, thus lifting its 0-dimensional argument to a
  1-dimensional object.
\end{rem}

This dimensional intuition leads to the following useful observation
on the forgetful functor.
\begin{prop}\label{prop:forgetfulpreserves}
  The forgetful functor $𝐂 → \psh[ℂ₀]$ preserves all limits and colimits, as
  well as image (in the sense of strong epi, mono) factorisations.
\end{prop}
\begin{proof}
  The forgetful functor $𝐂 → \psh[ℂ₀]$ is equivalent to the
  restriction functor $\psh[ℂ_{\source,\but}] → \psh[ℂ₀]$, which is
  both a left and right adjoint, hence preserves all limits and
  colimits.  Finally, image factorisations are computed pointwise in
  presheaf categories (see, e.g., \cite[§0]{Adamek}), hence are
  preserved by restriction functors.
\end{proof}

\subsection{Bisimulation and bisimilarity}\label{ss:bisimex}
Morphisms in $𝐂$ are a generalisation of graph
morphisms, which are a proof-relevant version of functional
simulations.  The analogue of functional bisimulations is as follows.
\begin{defi}\label{def:fib}
  A morphism $f∶ X → Y$ in $\psh[ℂ_{\source,\but}]$ is a \emph{functional
    bisimulation}, or a \emph{fibration}, iff it enjoys the (weak)
  right lifting property w.r.t.\ $𝐲_{s_L}∶ 𝐲_{\source(L)} → 𝐲_L$, for all
  $L ∈ ℂ₁$.
\end{defi}
\begin{rem}
  This definition is strongly inspired by Joyal et al.'s~\cite{DBLP:conf/lics/JoyalNW93}.
\end{rem}

Here is a characterisation of fibrations which will be important.  Let
us recall that a weak pullback satisfies the same universal property
as a pullback, albeit without uniqueness.
\begin{prop}\label{prop:sim:wpbk}
  A morphism $f∶ X → Y$ is a functional bisimulation iff the following
  diagram is a pointwise weak pullback.
  \begin{center}
  \begin{tikzcd}
    X₁ \arrow[r,"f₁"] \arrow["s_X"',d] & Y₁ \arrow[d,"s_Y"] \\
    X₀\source \arrow[r,"f₀\source"] & Y₀\source 
  \end{tikzcd}
\end{center}
\end{prop}
\begin{rem}
  Being a \emph{pointwise} weak pullback means that all squares
  \begin{center}
  \begin{tikzcd}
    X₁(L) \arrow[r,"(f₁)_L"] \arrow["(s_X)_L"',d] & Y₁(L) \arrow[d,"(s_Y)_L"] \\
    X₀(\source(L)) \arrow[r,"(f₀)_{\source(L)}"] & Y₀(\source(L))
  \end{tikzcd}
\end{center}
should be weak pullbacks, for $L ∈ 𝐨𝐛(ℂ₁)$. This is weaker than being
a weak pullback.
\end{rem}
\begin{proof}[Proof of Proposition~\ref{prop:sim:wpbk}]
  By Yoneda, a lifting problem in $𝐂$ as below left is the same as a
  cone in $\psh[ℂ₁]$ as below right, and a lifting is the same as a
  mediating morphism to $X₁$.
  \begin{center}
    \diag|baseline=(m-2-2.base)|(.6,.6){%
      𝐲_{\source(L)} \& X \\
      𝐲_L \& Y %
    }{%
      (m-1-1) edge[labela={x}] (m-1-2) %
      edge[labell={𝐲_{s_L}}] (m-2-1) %
      (m-2-1) edge[labelb={e}] (m-2-2) %
      edge[dashed] (m-1-2) %
      (m-1-2) edge[labelr={f}] (m-2-2) %
    }
    \hfil
    \diag|baseline=(m-4-2.base)|(.1,.6){
      𝐲_L \\
      \& X₁ \& \& Y₁ \\
      \& {} \\
      \& X₀ \source \& \& Y₀ \source %
    }{%
      (m-1-1) edge[bend right=10,labelbl={x}] (m-4-2) %
      (m-1-1) edge[bend left=10,labelar={e}] (m-2-4) %
      (m-1-1) edge[dashed] (m-2-2) %
      (m-2-2) edge[labelb={f₁}] (m-2-4) %
      edge[labelr={s_X}] (m-4-2) %
      (m-4-2) edge[labelb={f₀ \source}] (m-4-4) %
      (m-2-4) edge[labelr={s_Y}] (m-4-4) %
      }
  \end{center}
  (On the left $L$ and $\source(L)$ are viewed as objects of
  $ℂ_{\source,\but}$, hence should technically by written $in₁(L)$ and
  $in₀ (\source(L))$, respectively.)
\end{proof}

We now define general bisimulations, based on functional
bisimulations. Usually, one considers bisimulation relations. Here, we
generalise this a bit and consider arbitrary spans:

\begin{defi}
  A \emph{simulation} is a span $X ← R → Y$ whose left leg is a
  fibration.  A \emph{bisimulation} is a span of fibrations
  (equivalently, a simulation whose converse span is also a
  simulation).

  A \emph{simulation (resp.\ bisimulation) relation} is a relation
  which is a simulation (resp.\ bisimulation).
\end{defi}

\begin{rem}
  Of course, the relevant notion in our applications is
  substitution-closed bisimulation, to which we will come below.
\end{rem}

\begin{lem}\label{lem:bisimunion}
  Simulation relations and bisimulation relations are stable under unions.
\end{lem}
\begin{proof}
  By symmetry, it is enough to deal with the case of simulation
  relations.  Consider any family $(Rᵢ ↪ X×Y)_{i ∈ I}$ of simulation
  relations.  By Proposition~\ref{prop:union}, their union is the
  image of their cotupling. But because the domain $𝐲_{\source(L)}$ of
  $𝐲_{s_L}$ is representable for all $L ∈ ℂ₁$, any lifting problem
  $𝐲_{s_L} → ⋃ᵢ Rᵢ$ lifts to a lifting problem $𝐲_{s_L} → ∑ᵢ Rᵢ$,
  which in turn lifts to a lifting problem $𝐲_{s_L} → R_{i₀}$, for
  some $i₀ ∈ I$. We then find a lifting for the latter because
  $R_{i₀}$ is a simulation by hypothesis, which yields a lifting for
  the original.
\end{proof}

\begin{prop}\label{prop:bisim}
  For all $X,Y ∈ 𝐂$, the full subcategory $𝐁𝐢𝐬𝐢𝐦(X,Y)$ of spans
  between $X$ and $Y$ 
  which are bisimulations admits a terminal object
  $∼_{X,Y}$, called \emph{bisimilarity}.
\end{prop}
\begin{proof}
  As a presheaf category by Proposition~\ref{prop:isom}, $𝐂$ is
  well-powered, so we may consider the union $∼_{X,Y}$ of all
  bisimulation relations, which is again a bisimulation by
  Lemma~\ref{lem:bisimunion}.  Finally, $∼_{X,Y}$ is terminal, because
  any bisimulation $R$ factors through its image $im(R)$, which is
  again a bisimulation; as a bisimulation relation, $im(R)$ thus
  embeds into $∼_{X,Y}$, hence we obtain a morphism
  $R ↠ im(R) ↪ {∼_{X,Y}}$, which is unique by monicity of
  ${∼_{X,Y}} ↪ X×Y$.
\end{proof}

\section{Howe contexts for operational semantics}\label{s:howecontexts}
Operational semantics is a combination of syntax and transition
systems, in the sense that it is about transition systems whose states
form a model of a certain syntax.  Our framework for operational
semantics thus combines the frameworks of Fiore et
al.~\cite{fiore:presheaf} for syntax with variable binding, and
of~§\ref{s:trans} for transition systems.

In~§\ref{ss:transitionmonoidalgebras}, we introduce the ambient
setting for our framework, \emph{Howe contexts}, which are pre-Howe
contexts equipped with structure modelling substitution.  Furthermore,
for any Howe context and pointed strong endofunctor $Σ₀$ on
$\psh[ℂ₀]$, we introduce the category $Σ₀\Trans$ of \emph{transition
  $Σ₀$-monoids}, which are transition systems whose states form a
$Σ₀$-monoid.  We prove that the forgetful functor $Σ₀\Trans → 𝐂$ is
monadic.

In~§\ref{ss:sigs}, we then introduce \emph{dynamic signatures} over
$Σ₀$, which specify the dynamics of a transition system.  The
(dependent) pair $(Σ₀,Σ₁)$ then forms what we call an
\emph{operational semantics signature}.  We then define the models of
any such signature, called \emph{vertical} $Σ₁$-algebras.  They form a
category $Σ₁\algv$, and we prove that the forgetful functor
$Σ₁\algv → Σ₀\Trans$ is monadic.  We also prove that a suitably
constrained construction of the initial $Σ₁$-algebra is in fact
vertical, yielding an initial vertical $Σ₁$-algebra.  Finally, we
prove that, although both components are monadic, the composite
functor $Σ₁\algv → Σ₀\Trans → 𝐂$ is not.

\subsection{Transition monoid
  algebras} \label{ss:transitionmonoidalgebras} In this section, we
introduce Howe contexts, and introduce transition $Σ₀$-monoids, for
any suitable endofunctor $Σ₀$.

\begin{defi}
  A \emph{Howe context} consists of a pre-Howe context
  $\source,\but ∶ ℂ₁ → ℂ₀$, together with a monoidal structure on $\psh[ℂ₀]$,
  such that the tensor preserves all colimits on the left and filtered
  colimits on the right.
\end{defi}
\begin{nota}
  As for pre-Howe contexts, we let $𝐂 ≔ \psh[ℂ₁]/Δ$.
\end{nota}

Let us assume that some syntax has been specified by a finitary,
pointed strong endofunctor $Σ₀$ on $\psh[ℂ₀]$.  We then define 
\transitionmonoidsof{Σ₀} just as in~§\ref{s:overview}.
\begin{defi}
  The category $Σ₀\Trans$ of \emph{\transitionmonoidsof{Σ₀}} is the following pullback in $𝐂𝐀𝐓$. \hfil
    \Diag|baseline=(m-1-1.base)|{%
      \stdpbk %
    }{%
      Σ₀\Trans \&       Σ₀\mon \\
      𝐂 \& \psh[ℂ₀] %
    }{%
      (m-1-1) edge[labela={𝒟}] (m-1-2) %
      edge[labell={𝒰}] (m-2-1) %
      (m-2-1) edge[labelb={(-)₀}] (m-2-2) %
      (m-1-2) edge[shorten >=.5ex,labelr={𝒰₀}] (m-2-2) %
    }%

    When $Σ₀$ is the constantly empty endofunctor, we speak of
    \emph{\transitionmonoids}: they consist of objects $X$ equipped
    with monoid structure on $X₀$.
\end{defi}

We are now interested in computing initial $Σ₀$-monoids. For this, we
abstract over the concrete Proposition~\ref{prop:ptstr}, as follows,
replacing $[𝔽,𝐒𝐞𝐭]$ with any suitable category $𝒞₀$.
\begin{propC}[{\cite{fiore:presheaf,DBLP:conf/rta/FioreS17,BHL}}]\label{prop:ptstrgeneral}
  For any finitary, pointed strong endofunctor $Σ₀$ on a monoidal,
  cocomplete category $𝒞₀$ such that the tensor preserves all colimits
  on the left and filtered colimits on the right, the forgetful
  functor $𝒰₀∶ Σ₀\mon → 𝒞₀$ is monadic, and the free $Σ₀$-algebra over
  $I$ (equivalently the initial $(I+Σ₀)$-algebra) is an initial
  $Σ₀$-monoid.
\end{propC}
\begin{proof}
  This has been proved in Coq~\cite{BHL}.
\end{proof}
\begin{nota}\label{not:Z0}
  We denote the initial $Σ₀$-monoid by $𝐙_{Σ₀}$, or $𝐙₀$ for short.
\end{nota}

Using this, we obtain the following.
\begin{prop}
  \label{prop:adj-DM}
  The adjunction between $𝐂$ and $\psh[ℂ₀]$ (Proposition~\ref{prop:proj-adjoint}) lifts to an
  adjunction
  \begin{center}
    \adj{Σ₀\mon}{Σ₀\Trans}{ℳ}{𝒟}
  \end{center}
  with
  \begin{itemize}
  \item $𝒟(ℳ(X₀)) = X₀$ and 
  \item the left adjoint $ℳ$ maps any $Σ₀$-monoid $M$ to the discrete
    transition system on $M$, equipped with the original $Σ₀$-monoid
    structure on $M$. (In particular, we have $ℳ(X₀)₁ = ∅$.)
  \end{itemize}
\end{prop}
\begin{proof}
  This directly follows from the next lemma.
\end{proof}
\begin{lem}
  Let us consider any pullback \begin{center} \Diag{%
      \stdpbk }{%
      𝐀 \& 𝐂 \\
      𝐁 \& 𝐃 %
    }{%
      (m-1-1) edge[labela={S}] (m-1-2) %
      edge[labell={V}] (m-2-1) %
      (m-2-1) edge[labelb={R}] (m-2-2) %
      (m-1-2) edge[labelr={U}] (m-2-2) %
    }%
    \end{center}
    in $𝐂𝐀𝐓$ such that $U$ is monadic, say with left adjoint
    $J∶ 𝐃 → 𝐂$, and $R$ has a left adjoint $L$ with identity unit
    $ηᴿ_D = \id_D∶ D → RLD$.

    Then, $S$ admits a left adjoint $K$ with identity unit, such that
    the canonical natural transformation $LU → VK$ is an identity.
\end{lem}
\begin{proof}
  First of all by the triangle identities
  \begin{center}
    \diag{%
      RB \& \& RLRB \\
      \& RB\rlap{,}
    }{%
      (m-1-1) edge[identity,labela={ηᴿ_{RB}}] (m-1-3) %
      edge[identity,labelbl={}] (m-2-2) %
      (m-1-3) edge[labelbr={R(εᴿ_B)}] (m-2-2) %
    }
    \hfil
      \diag{%
        LD \& \& LRLD \\
        \& LD
      }{%

        (m-1-1) edge[identity,labela={Lηᴿ_D}] (m-1-3) %
        edge[identity,labelbl={}] (m-2-2) %
        (m-1-3) edge[labelbr={εᴿ_{LD}}] (m-2-2) %
      }
  \end{center}
  we have
  \begin{equation}
    R(εᴿ_B) = \id_{RB} \qquad \mbox{and} \qquad
    εᴿ_{LD} = \id_{LD}
    \label{eq:triangles}
  \end{equation}
  for all $C$ and $D$.

  Similarly, we have
  \begin{equation}
    RL = \id_{𝐃}\label{eq:RLid}
  \end{equation}
  not only on objects but also on
  morphisms, since by naturality of $η$ we have for any
  $f∶ D → D'$:

  \[RLf = RLf ∘ η_D = η_{D'} ∘ f = f.\]

  Let us furthermore assume w.l.o.g.\ that the pullback is constructed
  in the standard way, using compatible pairs.

  We then define $K(C) = (LU(C),C)$, which is legitimate since
  $RLU(C) = U(C)$ by hypothesis.  To prove the universal property,
  assume given $(B',C')$ such that $R(B') = U(C')$, and a morphism
  $f∶ C → S(B',C') = C'$ in $𝐂$.  Then, letting $\widetilde{U(f)}∶ LU(C) → B'$
  denote the transpose of $U(C) \xto{U(f)} U(C') = R(B')$, we have
  by~\eqref{eq:triangles} and~\eqref{eq:RLid}
  \[R(\widetilde{U(f)}) = R(εᴿ_{B'} ∘ LUf) = RLUf = Uf\rlap{,}\]
    so $(\widetilde{U(f)},f)∶ (LUC,C) → (B',C')$ in the pullback $𝐀$.
    Furthermore, the desired triangle
    \begin{center}
      \diag{%
        C \& S(LUC,C) \\
        \& S(B',C')\rlap{ = C'}
      }{%
        (m-1-1) edge[identity,labela={}] (m-1-2) %
        edge[labelbl={f}] (m-2-2) %
        (m-1-2) edge[labelr={S(\widetilde{U(f)},f) = f}] (m-2-2) %
      }
    \end{center}
    commutes as desired, trivially. Finally, any
    $(g,h)∶ (LUC,C) → (B',C')$ making it commute must satisfy $h = f$
    and $Rg = Uh = Uf$.  But $Rg∶ UC = RLUC → RB'$ is the transpose of
    $g$, so $g$ must conversely be the transpose of $Rg = Uf$, and
    hence $(g,h) = (\widetilde{U(f)},f)$, proving the desired
    uniqueness property.

    It remains to prove that the canonical natural transformation
    \[LUC = LUSKC = LRVKC \xto{εᴿ_{VKC}} VKC\]
    is an identity. But by construction $VKC = V(LUC,C) = LUC$, and
    $εᴿ_{LUC} = \id_{LUC}$ by~\eqref{eq:triangles}, hence the result.
\end{proof}

\begin{rem}
  The names, $ℳ$ and $𝒟$, stand for ``monter'' and ``descendre'', ``go
  up'' and ``go down'' in French.
\end{rem}

\begin{prop}
  \label{prop:SigmaMon-monadic}
  The forgetful functor $𝒰∶ Σ₀\Trans → 𝐂$ is finitary and monadic.
\end{prop}
\begin{proof}
  This follows from the fact that \transitionmonoidsof{Σ₀} are the
  algebras of an equational system over $𝐂$ in the sense of Fiore and
  Hur, to which~\cite[Theorem~6.1]{FioreHurEquational} applies.
\end{proof}
\begin{nota}\label{not:L}
  We denote by $ℒ$ the left adjoint to $𝒰$.
\end{nota}

\subsection{Operational semantics signatures}\label{ss:sigs}

Similarly, we define abstract dynamic signatures, which abstract over
those of Definition~\ref{def:dynsig}:
\begin{defi}
  Given a Howe context $\source,\but ∶ ℂ₁ → ℂ₀$ and a pointed strong
  $Σ₀∶ \psh[ℂ₀] → \psh[ℂ₀]$, a \emph{dynamic signature}
  $Σ₁ = (Σ₁^F,Σ₁^∂)$ over $Σ₀$ consists of a finitary functor
  $Σ₁^F∶ Σ₀\Trans → \psh[ℂ₁]$, together with a natural transformation
  $Σ₁^∂$ with components $Σ₁^F(X) → Σ₀(X₀)\source × X₀\but$.
\end{defi}

Let us pack up the static and dynamic notions of signature.
\begin{defi}
  An \emph{operational semantics signature} $(Σ₀,Σ₁)$ on a given Howe context
  $\source,\but ∶ ℂ₁ → ℂ₀$ consists of a pointed strong endofunctor $Σ₀$
  preserving sifted colimits, together with a dynamic signature $Σ₁$ over
  it.
\end{defi}
\begin{rem}\label{rk:sifted}
  Preservation of sifted colimits~\cite{Sifted} is stronger than
  finitarity for $Σ₀$.  We need it \shortfull{in the proof for $Σ₀$ to
    behave nicely w.r.t.\ relational transitive closure (when proving
    Lemma~\ref{lem:howtranssym} below)}{for
    Lemma~\ref{lem:wowomonoid} below}.  In a presheaf category like
  $\psh[ℂ₀]$, if $Σ₀$ preserves pullbacks (for example, by familiality), it
  is equivalent to being finitary and preserving all epis, as seen
  from the proof of \cite[Theorem 18.1]{algebraictheories}.
  In~\cite{BHL}, we mistakenly only require $Σ₀$ to be finitary, which
  yields a gap in the proof of~\cite[Lemma~5.13]{BHL}.
\end{rem}
\begin{exa}
  The endofunctor $Σ₀(X)(n) = X(n+1) + X(n)²$ on $[𝔽,𝐒𝐞𝐭]$ preserves
  sifted colimits.  This easily follows from the fact that sifted colimits
  are the ones commuting with products in sets.
\end{exa}

Let us now introduce the models of a dynamic signature.  We start by
fixing, for the rest of this section, an operational semantics
signature $(Σ₀,Σ₁)$ on a Howe context $\source,\but∶ ℂ₁ → ℂ₀$.
\begin{defi}
  Let
  $\check{Σ}₁∶ Σ₀\Trans → Σ₀\Trans$ map any \transitionmonoidof{Σ₀} $X$ to
  the composite 
  \[Σ₁^F(X) \xto{(Σ₁^∂)_X} Σ₀(X₀)\source × X₀\but \xto{ν_{X₀}\source × X₀\but} X₀\source × X₀\but\rlap{,}\]
  where $ν_{X₀}$ denotes the $Σ₀$-algebra structure of $X₀$.
\end{defi}
\begin{prop}
The endofunctor  $\check{Σ}₁$ is finitary and makes
  the following triangle commute.
  \hfil\begin{tikzcd}[baseline=(\tikzcdmatrixname-1-1.base)]
      Σ₀\Trans \arrow[rr,"\check{Σ}₁"] \arrow[dr,"𝒟"'] & &  Σ₀\Trans \arrow[dl,"𝒟"] \\
      & Σ₀\mon 
    \end{tikzcd}
\end{prop}
\begin{proof}
  Commutativity of the triangle holds by construction, and finitarity
  follows from finitarity of $Σ₁^F$.
\end{proof}
\begin{defi}
  A $\check{Σ}₁$-algebra structure $\check{Σ}₁(X) → X$ on an object
  $X ∈ Σ₀\Trans$ is \emph{vertical} when its image under the forgetful
  functor $Σ₀\Trans → Σ₀\mon$ is the identity. Let $Σ₁\algv$ denote
  the full subcategory of $\check{Σ}₁\alg$ spanned by all vertical
  algebras.
\end{defi}

\begin{thm}\label{thm:Z}
  The forgetful functor $Σ₁\algv → Σ₀\Trans$ is monadic, and
  furthermore the initial $\check{Σ}₁$-algebra $𝐙_{\check{Σ}₁}$, or
  $𝐙$ for short, may be chosen to be vertical, hence is also initial
  in $Σ₁\algv$.
\end{thm}
\begin{proof}
  For the first statement, vertical algebras may be specified as an
  equational system, in the sense of~\cite{FioreHurEquational}, so the
  result follows by~\cite[Theorem~6.1]{FioreHurEquational}.  For the
  second statement, $𝐙$ is the colimit of the initial chain
  \[𝐙₀ \xto{!} \check{Σ}₁(𝐙₀) \xto{\check{Σ}₁(!)} …
  \xto{\check{Σ}₁^{n-1}(!)} \check{Σ}₁ⁿ(𝐙₀) → … \] (where $𝐙₀$ is
  shorthand for $ℳ(𝐙₀)$, for readability, which is initial in
  $Σ₀\Trans$, by Proposition~\ref{prop:adj-DM}).  The image of this
  chain in $Σ₀\mon$ is the everywhere-identity chain on $𝐙₀$.
\end{proof}
By construction, we have:
\begin{prop}
  We have $𝒟(𝐙) = 𝐙₀$.
\end{prop}

Let us readily annihilate any hope that vertical $Σ₁$-algebras are
monadic over $𝐂$.
\begin{lem}
  Consider the operational semantics signature of $λ$-calculus with
  $Σ₀$ from~\eqref{eq:lambdasyntax} and $Σ₁$ from
  Example~\ref{ex:Sigmaun}.  Then, the composite forgetful functor
  $Σ₁\algv → Σ₀\Trans → 𝐂$ is not monadic.
    \end{lem}
    \begin{proof}
      Let us draw inspiration from the classical example of a
      non-monadic composite of monadic
      functors~\cite[p.~107]{BarrWells:ttt}, namely the composite
      $𝐂𝐚𝐭 → 𝐑𝐆𝐩𝐡 → 𝐒𝐞𝐭$, where $𝐑𝐆𝐩𝐡$ denotes the category of
      reflexive graphs, and the forgetful functor to sets returns the
      set of arrows. One may show that both functors are monadic, but
      that their composite is not.  By Beck's monadicity
      theorem~\cite[Theorem~VI.7.1]{MacLane:cwm}, it suffices to find
      a parallel pair $X ⇉ Y$ in $𝐂𝐚𝐭$ whose image in sets admits an
      absolute coequaliser, and show that this coequaliser is not
      created by the composite forgetful functor.  The idea is to take
      $Y$ to consist of two arrows
      \begin{mathpar}
        1 \xto{a} 2 \and 3 \xto{b} 4\rlap{,}
      \end{mathpar}
      and $X = 1 + Y$ to have an additional object, say $0$. We then
      define $u,v∶ X → Y$ to be the identity on $Y$, and respectively
      map $0$ to $2$ and $3$.  The coequalisers in $𝐂𝐚𝐭$ (top) and
      $𝐒𝐞𝐭$ (bottom) look as follows, abbreviating each $\id_C$ to
      just $C$ for readability.
      \begin{center}
        \diag(1,0.01){%
          0 \&\& |(one)| 1 \&\& |(two)| 2 \&\& |(three)| 3 \&\& |(four)| 4
          \& \qquad\qquad \&
          |(onei)| 1 \&\& |(twoi)| 2 \&\& |(threei)| 3 \&\& |(fouri)| 4
          \& \qquad\qquad \&
          |(oneii)| 1 \&\& |(twothree)| {\ens{2,3}} \&\& |(fourii)| 4 \\
          0 \&\& 1 \& a \& 2 \&\& 3 \& b \& |(quatre)| 4 
          \& \qquad\qquad \&
          |(uni)| 1 \& a \& 2 \&\& 3 \& b \&  |(quatrei)| 4
          \& \qquad\qquad \&
          |(unii)| 1 \& a \& {\ens{2,3}} \& b \& |(quatreii)| 4
    }{%
      (one) edge[labela={a}] (two) %
      (three) edge[labela={b}] (four) %
      (onei) edge[labela={a}] (twoi) %
      (threei) edge[labela={b}] (fouri) %
      (oneii) edge[labela={a}] (twothree) %
      edge[bend right,labelb={b∘a}] (fourii) %
      (twothree) edge[labela={b}] (fourii) %
      (four) edge[bend left,shorten <=1em,shorten >=1em,labela={u}] (onei) %
      (four) edge[bend right,shorten <=1em,shorten >=1em,labelb={v}] (onei) %
      (fouri) edge[onto,shorten <=1em,shorten >=1em] (oneii) %
      (quatre) edge[bend left,shorten <=1em,shorten >=1em,labela={}] (uni) %
      (quatre) edge[bend right,shorten <=1em,shorten >=1em,labelb={}] (uni) %
      (quatrei) edge[onto,shorten <=1em,shorten >=1em] (unii) %
      (fourii) edge[bend left,negate,mapsto,shorten <=.5em,shorten >=.5em] (quatreii) %
        }
      \end{center}
      One easily proves that the one in sets is split, hence absolute.
      Because of the composite arrow $b∘a$ that the
      coequaliser in $𝐂𝐚𝐭$ must have but the one in sets does not, the
      coequaliser is not created by the forgetful functor, thus
      contradicting monadicity.      
      
      For proving the lemma, we rely on the $β$-rule to mimick
      composition in constructing the following parallel pair $X ⇉ Y$
      in $Σ₁\algv$.
      \begin{itemize}
      \item $Y$ is the vertical $Σ₁$-algebra defined as the
        $λ$-calculus extended with two constants $a$ and $b$, unary
        operations $k$ and $l$, and an axiom $b ⇓ l⦇x⦈$;
      \item 
        $X$ is the vertical $Σ₁$-algebra extending $Y$ with a constant $c$;
      \item the $Σ₁$-algebra morphisms $u,v∶ X → Y$ respectively map
        $c$ to $b$ and $k⦇a⦈$.
      \end{itemize}
      The coequaliser $E$ of these two morphisms computed in the
      presheaf category $𝐂$ is a quotient of $Y$ by the equation
      $b = k⦇a⦈$. It thus has as reductions $[e] ⇓ [f]$ between
      equivalence classes, for all reductions $e' ⇓ f'$ between
      representatives $e' ∈ [e]$ and $f' ∈ [f]$.  E.g., it has a
      reduction $[k⦇a⦈]⇓ [l⦇x⦈]$, since $b ∈ [k⦇a⦈]$ and $b ⇓ l⦇x⦈$.
      However, $E$ lacks a reduction $[(λx.k⦇x⦈)\ a] ⇓ [l⦇x⦈]$.
      Indeed, $[(λx.k⦇x⦈)\ a]$ has a unique representative, namely
      $(λx.k⦇x⦈)\ a$, whose evaluation in $Y$ gets stuck at $k⦇a⦈$.
      However, the coequaliser in $Σ₁\algv$ does have such a
      reduction by applying the $β$-rule:
      \begin{mathpar}
        \inferrule{[λx.k⦇x⦈] ⇓ [k⦇x⦈] \\ [k⦇a⦈] = [b] ⇓ [l⦇x⦈]}{
          [(λx.k⦇x⦈)\ a] ⇓ [l⦇x⦈]}~·
      \end{mathpar}
      Thus, the coequaliser is not created by the forgetful functor.

      Finally, this coequaliser is split (hence absolute): the
      morphism $v∶ X → Y$ mapping $c$ to $k⦇a⦈$ has a section
      $f∶ Y → X$, which maps $k⦇a⦈$ to $c$, and the coequaliser arrow
      $e∶Y → E$ has a section $g∶ E → Y$ mapping $k⦇a⦈$ to $b$.  It
      is straightforward to check that this indeed defines a split
      coequaliser, i.e., that $g∘ e = u ∘ f$.
    \end{proof}
    \begin{rem}
      A crucial point is that $g$ may (and does) map $a$ and $k$ to
      themselves (the latter being necessary to preserve the target of
      $(λx.k⦇x⦈) ⇓ k⦇x⦈$), and $k⦇a⦈$ to $b ≠ k⦇a⦈$. Indeed, $g$ lives
      in $𝐂$, as opposed to $Σ₀\Trans$, hence need not preserve
      substitution. All that is required is naturality. But $k⦇a⦈$,
      being closed, cannot be the target of any transition, and by
      induction cannot be the source of any transition either, hence
      the result.

      Also, let us stress that in $Y$, $(λx.k⦇x⦈)\ a$ does not
      evaluate, since $k⦇a⦈$ itself does not. Similarly, in $E$, which
      is not saturated by the evaluation rules, we do have the
      premises $(λx.k⦇x⦈) ⇓ k⦇x⦈$ and $k⦇a⦈ = b⇓l⦇x⦈$, but not the
      conclusion $(λx.k⦇x⦈)\ a ⇓ l⦇x⦈$.
    \end{rem}

\section{Substitution-closed bisimilarity}\label{s:substclosedbisim}
We have now introduced our notion of syntactic transition system,
given by $Σ₀$-transition monoids in a Howe context, and explained how
to generate such systems from operational semantics signatures.  In
this section, we incorporate substitution into the notions of
bisimulation and bisimilarity introduced in~§\ref{ss:bisimex}, which
yields substitution-closed bisimilarity.  We then state our main
theorem, for which we include a high-level proof sketch.

In order to introduce substitution-closed bisimilarity, we first lift
the notion of bisimulation to $Σ₀$-monoids,
generalising~§\ref{ss:substclosed:ex}:
\begin{defi}\label{def:bisimSigmaMon}
  A morphism in $Σ₀\Trans$ is a \emph{functional bisimulation} iff its
  underlying morphism in $𝐂$ is. A span is a \emph{simulation} iff its
  left leg is a functional bisimulation, and a \emph{bisimulation} iff
  both of its legs are functional bisimulations.
\end{defi}
Let us readily prove the following characterisation by lifting,
recalling from Notation~\ref{not:L} that $ℒ∶ 𝐂 → Σ₀\Trans$ is left
adjoint to the forgetful functor.  
\begin{prop}\label{prop:bisimadj}
   A morphism in $Σ₀\Trans$ is a functional bisimulation
   iff it has the right lifting property w.r.t.\ 
   $ℒ(𝐲_{s_L})∶ ℒ(𝐲_{\source(L)}) → ℒ(𝐲_L)$, for all $L ∈ ℂ₁$.
 \end{prop}
 \begin{proof}
   By adjunction.
 \end{proof}

 As seen in Example~\ref{ex:bisimZ}, we now want to go beyond
 bisimilarity, and introduce abstract versions of substitution-closed
 bisimulation and bisimilarity.  For this, let us give the general
 definition of modules over a monoid.

\begin{defi}
  For any monoid $M$ in a monoidal category $𝒞$, let $M\Mod$ denote
  the category of algebras for the monad ${-}⊗M$.
\end{defi}
An $M$-module thus consists of an object $X$ equipped with an action
$X⊗M → X$ of $M$ satisfying straightforward coherence conditions.

\begin{exa}\label{ex:MasanMmodule}
  $M$ itself is an $M$-module, with action given by multiplication.
\end{exa}

Before going into substitution-closed bisimulation, let us record the
following useful properties of modules in a Howe context.
\begin{prop}\label{prop:modulespreserve}
  In any Howe context,  for all monoids $M$ in $\psh[ℂ₀]$,
  \begin{itemize}
  \item the forgetful functor $M\Mod → \psh[ℂ₀]$ creates all limits
    and colimits, and furthermore
  \item the category $M\Mod$ is regular and the forgetful functor
    $M\Mod → \psh[ℂ₀]$
  creates image factorisations.
\end{itemize}
\end{prop}
\begin{proof}
  For creation of limits and colimits:
  \begin{itemize}
  \item As algebras for the monad ${-}⊗M$, $M$-modules are closed
    under limits.
  \item   They are also closed under all types of colimits
    preserved by ${-}⊗M$, i.e., all of them by definition of Howe
    contexts.
\end{itemize}
Thus, $M\Mod$ is complete and cocomplete, hence regularity reduces to
showing that the pullback of any regular epi is again a regular
epi. So let us consider any pullback square
  \begin{center}
    \Diag{%
      \stdpbk %
    }{%
      A \& B \\
      C \& D
    }{%
      (m-1-1) edge[labela={u}] (m-1-2) %
      edge[labell={v}] (m-2-1) %
      (m-2-1) edge[labelb={g}] (m-2-2) %
      (m-1-2) edge[labelr={f}] (m-2-2) %
    }
  \end{center}
  in $M\Mod$, with $f$ a regular epi, and show that $v$ must also be a
  regular epi.  By creation, hence preservation, of limits and
  colimits, the given pullback square is also a pullback in $\psh[ℂ₀]$ and
  $f$ is a regular epi there too. So by regularity of the presheaf
  category $\psh[ℂ₀]$, $v$ is a regular epi in $\psh[ℂ₀]$. Equivalently, it is a
  coequaliser of its kernel pair. But by creation of limits the kernel
  pair uniquely lifts to a kernel pair in $M\Mod$, and by creation of
  colimits $v$ is a coequaliser there too. This shows that $M\Mod$ is
  regular.

  Finally, given $X,Z ∈ M\Mod$, let us consider any image
  factorisation $X \xarrow[onto]{e} Y \xarrow[into]{m} Z$ in $\psh[ℂ₀]$ of a
  morphism $f∶ X → Z$ in $M\Mod$, i.e., $e$ is a regular epi and $m$
  is a mono in $\psh[ℂ₀]$.  In this situation, $e$ is the coequaliser of its
  kernel pair in $\psh[ℂ₀]$, but, as we just saw, this kernel pair lifts to
  a kernel pair in $M\Mod$, whose coequaliser is created by the
  forgetful functor, hence $e$ is a coequaliser, hence a regular epi
  in $M\Mod$. Finally, $f$ also coequalises the kernel pair, hence the
  existence of a unique mediating morphism $Y → Z$ in $M\Mod$, which
  must be $m$ by faithfulness of the forgetful functor $M\Mod →
  \psh[ℂ₀]$. Thus, $m$ is also a morphism in $M\Mod$. Finally, its monicity
  follows again by faithfulness of the forgetful functor.
\end{proof}

Let us now introduce substitution-closed spans, first in an arbitrary
monoidal category, and then in a Howe context. This then leads us to
substitution-closed bisimulation.
\begin{defi}
  In a monoidal category $𝒞$ with binary products, given a monoid $M$
  and $M$-modules $X$ and $Y$, a \emph{substitution-closed} span is a
  span $p∶ R → X×Y$ equipped with $M$-module structure $ρ∶ R⊗M → R$ on
  $R$, such that $p∶ R → X×Y$ is an $M$-module morphism.
\end{defi}
The last condition is equivalent to commutation of the following
diagram.
\begin{center}
\diag|baseline=(m-1-1.base)|(.6,1.8){%
  R⊗M \& \& R \\
  (X×Y)⊗M \& (X⊗M)×(Y⊗M) \& X×Y %
}{%
  (m-1-1) edge[labela={ρ}] (m-1-3) %
  edge[labell={p⊗M}] (m-2-1) %
  (m-2-1) edge[labelb={⟨π₁⊗M,π₂⊗M⟩}] (m-2-2) %
  (m-2-2) edge[labelb={a_X × a_Y}] (m-2-3) %
  (m-1-3) edge[labelr={p}] (m-2-3) %
}
\end{center}
\begin{defi}
  Consider any Howe context $\source,\but ∶ ℂ₁ → ℂ₀$ and 
  \transitionmonoid $M ∈ 𝐂$. Let $X,Y ∈ 𝐂$ be equipped with $M₀$-module structure
  on $X₀$ and $Y₀$. A \emph{substitution-closed} span is a span
  $R → X×Y$ equipped with substitution-closed structure on
  $R₀ → X₀×Y₀$.
\end{defi}

\begin{defi}
  For any Howe context $\source,\but ∶ ℂ₁ → ℂ₀$ and \transitionmonoid
  $M ∈ 𝐂$, a \emph{substitution-closed simulation (resp.\
    bisimulation)} is a substitution-closed span $R → M²$ (viewing
  $M₀$ itself as an $M₀$-module by Example~\ref{ex:MasanMmodule})
  which is a simulation (resp.\ bisimulation).
  Let $𝐁𝐢𝐬𝐢𝐦^⊗(M)$ denote the full subcategory of $𝐂/M²$ spanned by
  substitution-closed bisimulations. 
\end{defi}
Let us now prove the existence of substitution-closed bisimilarity.
\begin{lem}\label{lem:bisimuniongeneral}
  Substitution-closed simulation and bisimulation relations are stable
  under unions.
\end{lem}
\begin{proof}
  By Lemma~\ref{lem:bisimunion}, the union of a family of
  substitution-closed simulation (resp.\ bisimulation) relations is
  again a simulation (resp.\ bisimulation) relation.  But by
  Proposition~\ref{prop:modulespreserve}, the union in $𝐂$ is again
  substitution-closed, which concludes the proof.
\end{proof}
\begin{prop}\label{prop:scbisim}
  For any Howe context $\source,\but ∶ ℂ₁ → ℂ₀$ and monoid
  $M ∈ 𝐂$, the category $𝐁𝐢𝐬𝐢𝐦^⊗(M)$ of
  substitution-closed bisimulations over $M$ admits a terminal object
  $∼^⊗_M$, called \emph{substitution-closed bisimilarity}.
\end{prop}
\begin{proof}
  Straightforward generalisation of the proof of
  Proposition~\ref{prop:bisim} using the lemma.
\end{proof}

\begin{nota}\label{not:simtens}
  When $M = 𝐙$, we abbreviate $∼^⊗_𝐙$ to just $∼^⊗$.
\end{nota}

We now want to state the abstract version of our main theorem, but we
need an additional hypothesis, which we now introduce.  The idea is
essentially that $Σ₁$ should preserve functional bisimulations, which
does not quite make sense, because the codomain of $Σ₁^F$ is
$\psh[ℂ₁]$, where no notion of functional bisimulation has been
defined yet.  Recalling Proposition~\ref{prop:sim:wpbk}, we work
around this as follows.
\begin{defi}
  \label{def:Sigma1-preserve-bisim}
  A dynamic signature $Σ₁ = (Σ₁^F,Σ₁^∂)$ \emph{preserves functional
    bisimulations} iff for any functional bisimulation $f∶ R → X$ in
  $Σ₀\Trans$, the following square is a pointwise weak pullback.
  \begin{equation}
    \diag(.6,1.6){%
      Σ₁^F(R) \& Σ₁^F(X) \\
      Σ₀(R₀)\source \& Σ₀(X₀)\source %
    }{%
      (m-1-1) edge[labela={Σ₁^F(f)}] (m-1-2) %
      edge[labell={π₁ ∘ (Σ₁^∂)_R}] (m-2-1) %
      (m-2-1) edge[labelb={Σ₀(f₀)\source}] (m-2-2) %
      (m-1-2) edge[labelr={π₁ ∘ (Σ₁^∂)_X}] (m-2-2) %
    }
    \label{eq:presfib}
  \end{equation}
\end{defi}

\begin{rem}
  It may not be obvious that the dynamic signature for call-by-name
  $λ$-calculus preserves functional bisimulations. We will come back
  to this in~§\ref{s:cellularity} by showing that it satisfies a
  sufficient condition, cellularity.
\end{rem}

\begin{rem}
  It may seem linguistically inappropriate to say that $Σ₁$ preserves
  functional bisimulations, since $Σ₁$ is not merely a functor, and we
  have not even defined fibrations in the codomain category $\psh[ℂ₁]$
  anyway. We will justify the terminology in
  Proposition~\ref{prop:ling-mischief}, but for now let us move on
  directly to the main result.
\end{rem}

\begin{thm}\label{thm:main}
  If $Σ₁$ preserves functional bisimulations, then substitution-closed
  bisimilarity is context-closed. More precisely, ${∼^⊗}$ is a
  \transitionmonoidof{Σ₀}, and ${∼^⊗} → 𝐙²$ is a
  \transitionmonoidof{Σ₀} morphism.
\end{thm}

\begin{proof}[{Proof sketch (see~§\ref{s:congruence} for the full proof)}]
  The proof takes inspiration from Howe's original method.
\begin{enumeratewide}
\item \label{item:wowo} We first define the \emph{Howe closure} $H₀$
  of substitution-closed bisimilarity $∼₀^⊗$ on states as the initial
  $Σ₀ᴴ$-monoid for the pointed strong endofunctor $Σ₀ᴴ$ on
  $\psh[ℂ₀]/𝐙²$ defined by $Σ₀ᴴ(R) = Σ₀(R) + (R;{∼₀^⊗})$.  We then
  show that, by construction, $H₀$ is a $Σ₀$-monoid and both
  projections are $Σ₀$-monoid morphisms. \medskip
  
\item We then define the \emph{transition Howe closure} ${H}$ of (the
  full) substitution-closed bisimilarity ${∼^{⊗}}$, as an initial algebra for
  an endofunctor $Σ₁ᴴ$ on a suitable category $𝐂ᴴ_𝐙$.  Very roughly,
  $𝐂ᴴ_𝐙$ is the category of spans $R → 𝐙²$ whose projection is
  precisely $H₀ → 𝐙²₀$, and $Σ₁ᴴ(R) = \check{Σ}₁(R) + (R;{∼^⊗})$.  We
  show:
\begin{restatable}{lem}{simintowow}
  \label{lem:simintowow}
    There exists a span morphism $𝐢ᴴ∶ {∼^{⊗}} → H$.
\end{restatable}

\item Next comes the key lemma:
\begin{restatable}{lem}{howisim}
  \label{lem:howisim}
    If $Σ₁$ preserves functional bisimulations, then the transition
    Howe closure ${H}$ is a substitution-closed simulation.
  \end{restatable}

  \begin{rem}
    Since ${H₀}$ is a $Σ₀$-monoid by construction, $H$ is
    easily seen to be substitution-closed, so the lemma really is
    about it being a simulation.
  \end{rem}
  The key lemma is proved by characterising $H$ as an initial algebra
  for a different endofunctor on a different category, whose initial
  chain involves iterated applications of $Σ₁$ (preserving simulations by
  hypothesis) to $π₁∶ ℳ(H₀) → ℳ(𝐙₀)$, which is trivially a simulation.
  \medskip
  
\item \label{item:sym} In the standard proof method, the next step is
  to prove that the transitive closure of $H₀$ is symmetric. But in
  our case $H₀$ is a general span, not a relation. In order to avoid
  some coherence issues, we introduce a suitable notion of
  \emph{relational} transitive closure for general spans, denoted by
  $-^{\overline{+}}$, which is equipped with a canonical
  map $𝐢^{\overline{+}}_R∶ R → R^{\overline{+}}$ for each span $R$. We then show:

  \begin{restatable}{lem}{howtranssym}
  \label{lem:howtranssym}
    The relational transitive closure ${H₀}{}^{\overline{+}}$ of the 
    Howe closure $H₀$ on states is symmetric.
 \end{restatable}

  As substitution-closed simulations are closed under transitive
  closure, we obtain
\begin{restatable}{cor}{howtransbisim}
  \label{lem:howtransbisim}
    ${H}{}^{\overline{+}}$ is a substitution-closed simulation which
    is symmetric on states.
\end{restatable}

  We then use the following lemma\iffull{ (proved in~\S\ref{ss:sym-tran-clos})}.
\begin{restatable}{lem}{projsimsym}
  \label{lem:projsimsym}
    For any substitution-closed simulation $R$ such that $R₀$ is
    symmetric, there exists a substitution-closed bisimulation $R'$
    and a span morphism $𝐢'_R∶ R → R'$.
\end{restatable}

  By terminality of $∼^⊗$, we thus get a unique morphism
  ${!}_{{{H}{}^{\overline{+}}}{}'}∶ {{{H}{}^{\overline{+}}}{}'} → {∼^⊗}$ over $𝐙²$.
  \medskip
  
\item From the chain \hfil
  \[{∼^⊗} \xto{𝐢ᴴ} {H} \xto{𝐢^{\overline{+}}_H}
  {{H}{}^{\overline{+}}} \xto{𝐢'_R}
  {{H}{}^{\overline{+}}}{}' \xto{{!}_{{{H}{}^{\overline{+}}}{}'}} {∼^⊗}
  \]
 we get by terminality that ${∼^⊗}$ is a retract of a
  \transitionmonoidof{Σ₀}, namely ${H}$.  The result then readily
  follows from monadicity of \transitionmonoidsof{Σ₀}
  (Proposition~\ref{prop:SigmaMon-monadic}) and the following result,
  taking $X = {H}$, $Y = {∼^⊗}$, and $Z = 𝐙²$.\qedhere
\end{enumeratewide}
\end{proof}

  \begin{lem}\label{lem:retracte}
    Consider a monad $T∶ 𝒞 → 𝒞$ on any category $𝒞$, $T$-algebras $X$
    and $Z$, and morphisms $X \xarrow[onto]{e} Y \xarrow[into]{m} Z$ in $𝒞$ such
    that the composite is a $T$-algebra morphism, $e$ is a split epi,
    and $m$ is monic. Then there is a unique $T$-algebra structure on
    $Y$ such that $e$ and $m$ both are $T$-algebra morphisms.
  \end{lem}
  \begin{proof}
    Let $s∶ Y → X$ denote any section of $e$.
    The desired structure is given by 
    \[T(Y) \xto{T(s)} T(X) → X \xarrow[onto]{e} Y\rlap{,}\]
    where the middle morphism is the given $T$-algebra structure on
    $X$, and the rest follows by monicity of $m$.
  \end{proof}
  
\section{Preservation of functional bisimulations, and cellularity}\label{s:cellularity}
Let us now consider the main hypothesis of Theorem~\ref{thm:main},
preservation of functional bisimulations.  In~§\ref{ss:rephrase}, we
rephrase the condition in a way that makes more sense linguistically,
i.e., by an actual preservation condition.  We then work towards a
characterisation in terms of cellularity.  In~§\ref{ss:fam}, we first
briefly recall \emph{familial}
functors~\cite{Diers1978Spectres,DBLP:journals/mscs/CarboniJ95,Weber:famfun},
and show that the operational semantics signature for call-by-name
$λ$-calculus gives rise to two familial functors, in a suitable sense.
In~§\ref{ss:cell}, we restrict attention to the case where both
components of the dynamic signature give rise to familial functors in
this sense, and show that preservation of functional bisimulations is
then equivalent to a \emph{cellularity}
condition~\cite{garner:hal-01246365,BHL}, which itself comes with a
useful sufficient condition.

\subsection{An alternative characterisation}\label{ss:rephrase}
Let us first give an alternative definition of dynamic
signatures.
 \begin{defi}\label{def:Covers}
  Let $\psh[ℂ₁]/Δ_{\source}$ denote the following lax limit category.
  \begin{center}
    \Diag(.6,.6){%
      \twocell[.5][.5]{m-1-2}{m-2-3}{m-1-2}{m-2-1}{}{cell=.4} %
    }{%
      \& \psh[ℂ₁]/Δ_{\source}   \\
      \psh[ℂ₀] \& \& \psh[ℂ₁]
    }{%
      (m-1-2) edge[labelal={π₂}] (m-2-1) %
      edge[labelar={π₁}] (m-2-3) %
      (m-2-1) edge[labelb={Δ_{\source}}] (m-2-3) %
    }%
    \end{center}
 \end{defi}
 Concretely, an object consists of presheaves $X₁$ and $X₀$, together
 with a morphism $X₁ → X₀\source$.
 Just as $𝐂$, $\psh[ℂ₁]/Δ_{\source}$ is in fact a presheaf category:
 \begin{prop}\label{prop:Csource}
   The category $\psh[ℂ₁]/Δ_{\source}$ is isomorphic to the presheaf
   category over the lax colimit $ℂ_\source$ of the functor
   $\source∶ ℂ₁ → ℂ₀$, as in
  \begin{center}
    \Diag(.6,.6){%
    \twocell[.55]{m-2-2}{m-1-3}{m-2-2}{m-1-1}{}{cell=0.2,labelon={s}} %
  }{%
      ℂ₁ \& \& ℂ₀ \\
      \& ℂ_\source\rlap{.} %
    }{%
      (m-1-1) edge[labela={\source}] (m-1-3) %
      edge[labelbl={in₁}] (m-2-2) %
      (m-1-3) edge[labelbr={in₀}] (m-2-2) %
    }
  \end{center}   
 \end{prop}
 \begin{proof}
   Similar to Proposition~\ref{prop:isom}.
 \end{proof}
 \begin{rem}\label{rk:Csigma}
   Concretely, $ℂ_\source$ is the coproduct of $ℂ₀$ and $ℂ₁$, augmented with
   morphisms $s_L∶ \source(L) → L$ for all $L ∈ ℂ₁$, naturally in $L$.
\end{rem}
\begin{nota}\label{not:Downarrowsource}
  In the case of call-by-name $λ$-calculus, as in $𝐂$
  (Notation~\ref{not:Downarrow}), we call $⇓$ the unique object coming
  from $ℂ₁ = 1$.
\end{nota}

\begin{defi}
  For any operational semantics signature $(Σ₀,Σ₁)$, let
  $Σ₁^\source∶ Σ₀\Trans → \psh[ℂ₁]/Δ_{\source}$ map any $X₁ → X₀\source×X₀\but$
  to the first leg $Σ₁^F(X) → Σ₀(X₀)\source$ of $(Σ₁^∂)_X$.
\end{defi}
 \begin{prop}\label{prop:Sigmacomma}
   For any operational semantics signature $(Σ₀,Σ₁)$, the functor
   \[Σ₁^\source∶ Σ₀\Trans → \psh[ℂ₁]/Δ_{\source}\] is finitary, and makes the
   following diagram commute,
   \begin{equation}
     \diag{%
       Σ₀\Trans \& \& \psh[ℂ₁]/Δ_{\source} \\
       Σ₀\Mon \& \psh[ℂ₀] \& \psh[ℂ₀] %
     }{%
       (m-1-1) edge[labela={Σ₁^\source},shorten >=.5ex] (m-1-3) %
       edge[labell={𝒟},shorten >=.5ex] (m-2-1) %
       (m-2-1)edge[labelb={𝒰₀}] (m-2-2) %
       (m-2-2) edge[labelb={Σ₀}] (m-2-3) %
       (m-1-3) edge[labelr={π₂},shorten >=.5ex] (m-2-3) %
     }
     \label{eq:dynalt}
   \end{equation}
   where $π₂$, as in Definition~\ref{def:Covers}, maps any object
   $X₁ → X₀\source$ to $X₀$.
 \end{prop}
 \begin{proof}
   Finitarity holds because it does pointwise, by assumption.  The
   diagram commutes by construction.
 \end{proof}

 In $\psh[ℂ_\source]$ (or, through the isomorphism of
 Proposition~\ref{prop:Csource}, in $\psh[ℂ₁]/Δ_{\source}$), we may
 define functional bisimulations by analogy with
 Definition~\ref{def:fib}.
 \begin{defi}\label{def:fbisim:Csource}
   A morphism $f∶ X → Y$ in $\psh[ℂ_{\source}]$ is a \emph{functional
     bisimulation}, or a \emph{fibration}, iff it enjoys the (weak)
   right lifting property w.r.t.\ $𝐲_{s_L}∶ 𝐲_{\source(L)} → 𝐲_L$, for all
   $L ∈ ℂ₁$.
 \end{defi}

 \begin{prop}[Price for our linguistic mischief]
   \label{prop:ling-mischief}
   A dynamic signature preserves functional bisimulations
   (Definition~\ref{def:Sigma1-preserve-bisim}) iff the
   induced functor $Σ₁^\source∶ Σ₀\Trans → \psh[ℂ₁]/Δ_{\source}$ does.
 \end{prop}
 \begin{proof}
   The functor $Σ₁^\source$ maps any functional bisimulation
   $f∶ R → Y$ to the square~\eqref{eq:presfib}, and just as in
   Proposition~\ref{prop:sim:wpbk} a morphism in $\psh[ℂ_\source]$ is
   a functional bisimulation iff the corresponding square in
   $\psh[ℂ₁]$ is a pointwise weak pullback.
 \end{proof}

 \subsection{Familiality}\label{ss:fam}
 In the previous sections, we have seen that functional bisimulations
 may be defined by lifting both in $Σ₀\Trans$ and $\psh[ℂ₁]/Δ_{\source}$.  We
 now want to exploit this to obtain a characterisation of preservation
 of functional bisimulations, which will then lead us to useful
 sufficient conditions.

 For this, let us briefly recall familial functors, and show that the
 functors $Σ₀$ and $Σ₁^\source∶ Σ₀\Trans → \psh[ℂ₁]/Δ_{\source}$ induced
 by the dynamic signature for call-by-name $λ$-calculus are familial.

 Familial functors are a generalisation of polynomial functors on
 sets, i.e., functors of the form $F(X) = ∑_{o ∈ O} X^{nₒ}$, where $O$
 is a set of `operations', and $nₒ ∈ ℕ$ is the `arity' of any $o ∈ O$.
 \begin{short}
    First, let us recall from MacLane and Moerdijk~\cite{MM}:
 \begin{defi}
   The \emph{category of elements} $\el(X)$ of a presheaf $X$ over any
   category $ℂ$ has pairs $(c,x)$ with $x ∈ X(c)$ as objects, and a
   morphism $f↾x'∶ (c,x) → (c',x')$ for all $f∶ c → c'$ such that
   $X(f)(x') = x$.
\end{defi}
 \end{short}
 \begin{full}
   We want to generalise this to presheaf categories.
   \begin{exa}
   Consider for
 example the `free category' monad $T$ on $𝐆𝐩𝐡$.  Analysing and
 abstracting over the definition of $T$, we will arrive at the notion
 of familial functor.  Let us first recall that graphs are presheaves
 over the category
 \begin{center}
   \diaginline(0,.6){{[0]} \& {[1].}}{%
     (m-1-1) edge[bend left=10,labela={s}] (m-1-2) %
     (m-1-1) edge[bend right=10,labelb={t}] (m-1-2) %
   }
 \end{center}
$T$ does not change the vertex set, and an edge of $T(G)$ is
 merely a path in $G$. Indexing this by the length of the path, we
 obtain
   \begin{center}
     $T(G)[0] ≅ 𝐆𝐩𝐡(𝐲_{[0]},G)$ \hfil and \hfil
     $T(G)[1] ≅ ∑ₙ 𝐆𝐩𝐡([n],G)$,
   \end{center}
   where $[n]$ denotes the filiform graph $• → • … → •$ with $n$ edges
   (which is consistent with $[0]$ and $[1]$ through the Yoneda
   embedding).  Furthermore, the source of a path $[n] → G$ 
   is obtained as the composite
   \[[0] \xto{sₙ} [n]→G\rlap{,}\]
   where the first morphism selects the first vertex of the path.
   Similarly the target is obtained by precomposition with the
   morphism, say $tₙ$, selecting the last vertex.

   From these observations, let us now explain how the whole of $T$
   may be derived from
   \begin{itemize}
   \item the graph $T(1)$, which generalises the set of operations, and
   \item a functor $\el(T(1)) → 𝐆𝐩𝐡$, morally describing the arity of
     each operation,
\end{itemize}
where we recall from MacLane and Moerdijk~\cite{MM}:
 \begin{defi}
   The \emph{category of elements} $\el(X)$ of a presheaf $X$ over any
   category $ℂ$ has pairs $(c,x)$ with $x ∈ X(c)$ as objects, and a
   morphism $f↾x'∶ (c,x) → (c',x')$ for all $f∶ c → c'$ such that
   $X(f)(x') = x$.
\end{defi}
The graph $T(1)$ has a single vertex, and as many paths as we can derive from
a single endo-edge on this vertex: $ℕ$, because there is one for each length.
The category of elements of $T(1)$ thus looks like the following,
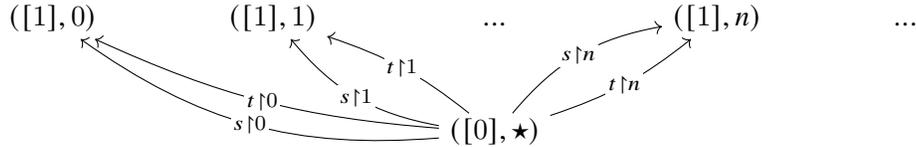
\begin{center}
  \diag(1,1.5){%
    ([1],0) \& ([1],1) \& … \& ([1],n) \& … \\
    \& \& ([0],⋆) %
  }{%
    (m-2-3) edge[bend left=20,labelon={s↾0}] (m-1-1) %
    (m-2-3) edge[bend left=10,labelon={t↾0}] (m-1-1) %
    (m-2-3) edge[bend left=20,labelon={s↾1}] (m-1-2) %
    (m-2-3) edge[bend right=10,labelon={t↾1}] (m-1-2) %
    (m-2-3) edge[bend right=10,labelon={t↾n}] (m-1-4) %
    (m-2-3) edge[bend left=20,labelon={s↾n}] (m-1-4) %
  }
\end{center}
and the assignments
\begin{mathpar}
([0],⋆) ↦ [0] \and ([1],n) ↦ [n]
\end{mathpar}
extend to a functor $E∶ \el(T(1)) → 𝐆𝐩𝐡$ by mapping each source or target
map $([0],⋆) → ([1],n)$ to the corresponding map $[0] → [n]$.
This functor may be visualised as
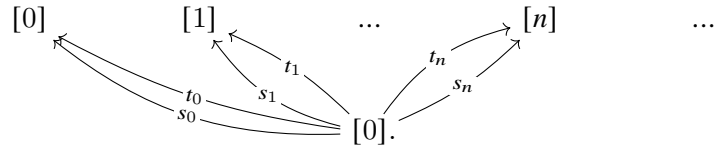
\begin{center}
  \diag(1,1.5){%
    {[0]} \& {[1]} \& … \& {[n]} \& … \\
    \& \& {[0]}\rlap{.} %
  }{%
    (m-2-3) edge[bend left=20,labelon={s₀}] (m-1-1) %
    (m-2-3) edge[bend left=10,labelon={t₀}] (m-1-1) %
    (m-2-3) edge[bend left=20,labelon={s₁}] (m-1-2) %
    (m-2-3) edge[bend right=10,labelon={t₁}] (m-1-2) %
    (m-2-3) edge[bend right=10,labelon={sₙ}] (m-1-4) %
    (m-2-3) edge[bend left=20,labelon={tₙ}] (m-1-4) %
  }
\end{center}
The promised expression of $T$ in terms of $T(1)$ and $E$ is: 
\[T(G)(c) ≅ ∑_{o ∈ T(1)(c)} 𝐆𝐩𝐡(E(c,o),G).\]
\end{exa}
\end{full}

 \begin{defi}\label{def:familial}
   A functor $F∶ 𝒜 → \psh$ to some presheaf category is \emph{familial}
   iff we have a natural isomorphism
   \[F(X)(c) ≅ ∑_{o ∈ O(c)} 𝒜(E(c,o),X)\rlap{,}\]
   for some presheaf $O ∈ \psh$ and functor $E∶ \el(O) → 𝒜$.  The
   presheaf $O$ is called the presheaf of \emph{operations}, or the
   \emph{spectrum}~\cite{Diers1978Spectres} of $F$, while $E$ is
   called the \emph{arity}, or \emph{exponent} functor.
 \end{defi}
 \begin{rem}
   This definition is a bit elliptic, so let us make functoriality
   more explicit.
   \begin{itemize}
   \item Functoriality in $X$ is by post-composition.
   \item For functoriality in $c$, for any $f∶ c → d$ in $ℂ$
     and $o ∈ O(d)$, letting $o' = o · f$, we get a morphism
     $f ↾ o∶ (c,o') → (d,o)$ in $\el(O)$, hence a morphism
     $E(f↾o)∶ E (c,o') → E(d,o)$. Precomposition by this morphisms
     yields a map \[𝒜(E(d,o),X) → 𝒜(E(c,o'),X).\] Postcomposing with
     the obvious coproduct injections, and cotupling, we get the desired map
     \[∑_{o ∈ O(d)} 𝒜(E(d,o),X) → ∑_{o' ∈ O(c)} 𝒜(E(c,o'),X).\]
   \end{itemize}
 \end{rem}
 \begin{rem}\label{remark:O}
   If $𝒜$ has a terminal object, we always have $O ≅ F(1)$.
 \end{rem}
 
 \begin{exa}
   Let us show that the endofunctor $Σ₀∶ \psh[ℂ₀] → \psh[ℂ₀]$
   from~§\ref{s:exsyntax} for $λ$-calculus is familial, where we
   recall that $ℂ₀ = \op{𝔽}$. Indeed, we then have
   \[
   \begin{array}[t]{rcll}
     Σ₀(X)(n) & = & X(n+1) + X(n)² \\
              & ≅ & \psh[ℂ₀](𝐲_{n+1},X) + \psh[ℂ₀](2·𝐲ₙ,X) \\
   \end{array}\]
    Thus, we choose:
    \begin{center}
    $\begin{array}[t]{rcl}
     O(n) & = &  \{ \mathrm{abs}, \mathrm{app}\}
   \end{array}$
   \hfil
   $\begin{array}[t]{rcl}
     E(n, \mathrm{abs}) & = & 𝐲_{n+1} \\
     E(n, \mathrm{app}) & = & 2·𝐲ₙ\rlap{.}
   \end{array}$
 \end{center}
 These definitions can be straightforwardly upgraded to functors
 $O∈\psh[ℂ₀]$ and $E ∶ \el(O) → \psh[ℂ₀]$, and we get the desired
 isomorphism.
 \end{exa}

 \begin{exa}\label{ex:cbn:familial}
   Let us now show that the functor
   $Σ₁^\source∶ Σ₀\Trans → \psh[ℂ₁]/Δ_{\source}$ for call-by-name
   $λ$-calculus is familial.  By definition, it maps any $X$ to the
   set-map $X₀(1) + Aᵦ(X) → Σ₀(X₀)(0)$ defined in
   Example~\ref{ex:Sigmaun}.  Let us transfer this across the
   isomorphism $\psh[ℂ₁]/Δ_{\source} ≅ \psh[ℂ_\source]$ (recalling
   Remark~\ref{rk:Csigma}), and show that $Σ₁^\source$ may be
   expressed as in Definition~\ref{def:familial}, first for
   $c ∈ 𝔽$ and then for $c = ⇓$.

   For $c = n ∈ 𝔽$, we almost may proceed as for $Σ₀$, except that the
   domain category has changed (from $\psh[ℂ₀]$ to $Σ₀\Trans$). But
   recalling that $ℒ∶ 𝐂 → Σ₀\Trans$ denotes the left adjoint to the
   forgetful functor $𝒰$, we have
   \[\begin{array}{rcll}
       Σ₁^\source(X)(n) & = & Σ₀((𝒰(X))₀)(n) \\
              & ≅ & 𝐂(𝐲_{n+1},𝒰(X)) + 𝐂(2·𝐲ₙ,𝒰(X)) \\
              & ≅ & Σ₀\Trans(ℒ(𝐲_{n+1}),X) + Σ₀\Trans(ℒ(2 · 𝐲ₙ),X)\rlap{,}
     \end{array}\]
     so we may (partially) define
     \begin{align}
       \label{eq:Sigma1-lc}
       \begin{aligned}
         O(n) & = \{\mathrm{abs}, \mathrm{app}\}
         &
         E(n, \mathrm{abs}) & = ℒ(𝐲_{n+1}) \\
         &&
         E(n, \mathrm{app}) & = ℒ(2·𝐲ₙ)\rlap{.}
       \end{aligned}
     \end{align}
     
     Now, for $c = ⇓$, remembering from
     Notation~\ref{not:Downarrowsource} that we call ${⇓} ∈ ℂ_\source$
     the unique object of $ℂ_\source$ coming from $ℂ₁ = 1$, on
     transitions, we have:
     \begin{align*}
       Σ₁^\source(X)(⇓) & = X₀(1) + Aᵦ(X) \\
                & ≅ 𝐂(𝐲_1,𝒰(X)) + Aᵦ(X) \\
                & ≅ Σ₀\Trans(ℒ(𝐲_1),X) + Aᵦ(X) \rlap{.}
     \end{align*}
     We thus need to find $Eᵦ$ such that $Aᵦ(X)≅Σ₀\Trans(Eᵦ,X)$, and
     then we would complete equations~\eqref{eq:Sigma1-lc} with:
     \hfil 
     \begin{math}
       \begin{array}[t]{rcl}
         O(⇓) & = & \{ λ\text{-}\mathrm{val}, β\text{-}\mathrm{red} \} 
       \end{array}
     \end{math}
     \hfil
     \begin{math}
       \begin{array}[t]{rcl}
         E(⇓, λ\text{-}\mathrm{val}) & = & ℒ(𝐲₁) \\
         E(⇓, β\text{-}\mathrm{red}) & = & Eᵦ\rlap{.}
       \end{array}
     \end{math}
     \begin{defi}\label{def:chibar}
       The morphism $\bar{χ}∶ ℒ(𝐲₀) → ℒ(𝐲₁+𝐲₀)$ is defined as follows.
       \begin{itemize}
       \item We start from $k₁⦇k₀⦈ ∈ 𝒰 (ℒ (𝐲₁+𝐲₀))(0)$, where we
         recall from the rules below Proposition~\ref{prop:freelambda}
         that the presheaf $𝒰 (ℒ (𝐲₁+𝐲₀))(0)$ has as states $λ$-terms
         over a closed constant $k₀$, and a unary constant $k₁$.
       \item We then let $\bar{χ}$ denote the mate of the morphism
         $𝐲₀ → 𝒰 (ℒ (𝐲₁+𝐲₀))$ corresponding to $k₁⦇k₀⦈$ by the Yoneda
         lemma.
       \end{itemize}
     \end{defi}
     Let now $Eᵦ$ denote the following pushout,
  \begin{center}
    \Diag(.6,1.4){%
      \pbk{m-2-1}{m-2-3}{m-1-3} %
    }{%
      ℒ(𝐲₀) \& ℒ(𝐲₁+𝐲₀) \& ℒ(𝐲_⇓+𝐲₀) \\
      ℒ(𝐲_⇓) \& \& Eᵦ %
    }{%
      (m-1-1) edge[labela={\bar{χ}}] (m-1-2) %
      edge[labell={ℒ(𝐲_{s_⇓})}] (m-2-1) %
      (m-1-2) edge[labela={ℒ(𝐲_{t_⇓}+𝐲₀)}] (m-1-3) %
      (m-2-1) edge[labelb={in₁}] (m-2-3) %
      (m-1-3) edge[labelr={in₂}] (m-2-3) %
    }%
  \end{center}
  which exists because by Proposition~\ref{prop:SigmaMon-monadic}
  $Σ₀\Trans$ is the category of algebras for a finitary monad on a
  presheaf category, hence a locally finitely presentable category
  by~\cite[Remark in~§2.78]{Adamek}, hence cocomplete.

  Let us now show that $Aᵦ(X) ≅ Σ₀\Trans(Eᵦ, X)$ for any $X$: 
  as $Σ₀\Trans(-,X)$ turns colimits into limits, we have the pullback
  \begin{center}
    \Diag(.6,2){%
      \pullbackk(5pt){m-2-1}{m-1-1}{m-1-3}{draw,-} %
    }{%
      {[Eᵦ,X]} \& \& {[ℒ(𝐲_⇓), X]} \\
      {[ℒ(𝐲_⇓+𝐲₀), X]} \& {[ℒ(𝐲₁ + 𝐲₀), X]} \& {[ℒ(𝐲₀), X]\rlap{,}} %
    }{%
      (m-2-2) edge[labelb={[\bar{χ}, X]}] (m-2-3) %
      (m-1-1) edge[labelr={}] (m-2-1) %
      (m-1-3) edge[labelr={[ℒ(𝐲_{s_⇓}),X]}] (m-2-3) %
      (m-2-1) edge[labelb={[ℒ(𝐲_{t_⇓}+𝐲₀),X]}] (m-2-2) %
      (m-1-1) edge[labelb={}] (m-1-3) %
    }%
  \end{center}
  where $\bar{χ}$ is as in Definition~\ref{def:chibar}, and we
  abbreviate $Σ₀\Trans(-₁,-₂)$ to $[-₁,-₂]$ for readability.  By
  Yoneda, this reduces to
  \begin{center}
    \Diag(.6,1.4){%
      \pullbackk(5pt){m-2-1}{m-1-1}{m-1-3}{draw,-} %
    }{%
      {[Eᵦ,X]} \& \& X(⇓) \\
      X(⇓)×X(0) \& X(1)×X(0) \& X(0)\rlap{,} %
    }{%
      (m-2-1) edge[labelb={X(t_⇓)×X(0)}] (m-2-2) %
      (m-1-1) edge[labelr={}] (m-2-1) %
      (m-1-3) edge[labelr={X({s_⇓})}] (m-2-3) %
      (m-2-2) edge[labelb={(t,u)↦t[u]}] (m-2-3) %
      (m-1-1) edge[labelb={}] (m-1-3) %
    }%
  \end{center}
  which shows that we have $Aᵦ(X) ≅ [Eᵦ,X]$ as desired.

  We have thus defined the actions of the functors
  $O∈ \psh[ℂ_\source]$ and $E∶ \el(O) → Σ₀\Trans$ on objects.  On
  morphisms, the only non-obvious point is the image of
  $s_⇓ ↾ λ\text{-}\mathrm{val}$ and $s_⇓ ↾ β\text{-}\mathrm{red}$.
  The former morphism is mapped to the identity on
  $E(0,\mathrm{abs}) = ℒ(𝐲₁) = E(⇓,λ\text{-}\mathrm{val})$.
  The latter is mapped to the composite
  \begin{equation}
    ℒ(2·𝐲₀) \xto{ℒ(𝐲_{s_⇓}+𝐲₀)} ℒ(𝐲_⇓+𝐲₀) \xto{in₂} Eᵦ\rlap{.}\label{eq:Ebetasource}
  \end{equation}
  This achieves the desired isomorphism
  $Σ₁^\source(X)(c) ≅ ∑_{o ∈ O(c)} Σ₀\Trans(E(c,o),X)$.
  \end{exa}

\subsection{Cellularity}\label{ss:cell}
We now want to exploit familiality to obtain an alternative
characterisation of preservation of functional bisimulations.  The
starting point is the observation that when a functor $F∶ 𝒜 → \psh$ is
familial, say as $F(A)(c) = ∑_{o ∈ O(c)} 𝒜(E(c,o),A)$, then any
morphism of the form $f∶ 𝐲_c → F(A)$, corresponding by Yoneda and
familiality to some pair $(o,φ)$ with $φ∶ E(c,o) → A$, factors as
\[𝐲_c \xto{(o,\id_{E(c,o)})} F(E(c,o)) \xto{F(φ)} F(A)\rlap{.}\]
Furthermore, the first component $(o,\id_{E(c,o)})$ is easily seen to
be \emph{generic}, in the following sense.

\begin{defi}
  Given any functor $F∶ 𝒜 → ℬ$, a morphism $ξ∶ B → F(A)$ is
  \emph{$F$-generic} (or \emph{generic} for short) whenever
  for all $χ$, $f$, and $g$ making the square below (solid) commute,
  \begin{center}
    \diag{%
      B \& F(C) \\
      F(A) \& F(D) %
    }{%
      (m-1-1) edge[labela={χ}] (m-1-2) %
      edge[labell={ξ}] (m-2-1) %
      (m-2-1) edge[labelb={F(f)}] (m-2-2) %
      edge[dashed,labelal={F(k)}] (m-1-2) %
      (m-1-2) edge[labelr={F(g)}] (m-2-2) %
    }
  \end{center}
  there is a unique lifting $k$ (dashed) such that $F(k) ∘ ξ = χ$ and
  $g ∘ k = f$.
\end{defi}

In fact, we have the following important alternative characterisation
of familial functors to presheaf categories.
\begin{thm}\label{thm:characterisation:familial}
  For any functor $F∶ 𝒜 → \psh$ such that $𝒜$ has a terminal
  object, $F$ is familial iff all morphisms $f∶ X → F(A)$ factor as
  \[X \xto{ξ} F(U) \xto{F(φ)} F(A)\rlap{,}\] with $ξ$ generic.
\end{thm}
\begin{proof}
  This is~\cite[Proposition~3.8]{garner:hal-01246365}, using the
  remark just before it.
\end{proof}
\begin{rem}
  The factorisation is essentially unique.
\end{rem}

The characterisation of familiality in terms of generic morphisms
allows us to characterise familial functors which preserve fibrations,
as cellular functors, which we now introduce.

\begin{defi}
  A \emph{category with generating cofibrations} is a category
  $𝒞$ equipped with a set $𝕁$ of morphisms.

  For any such $(𝒞,𝕁)$, as in~§\ref{ss:notation}, we call
  \emph{fibrations} all morphisms in $\wbotright{𝕁}$, and
  \emph{cofibrations} all morphisms in $\wbotrightleft{𝕁}$.
\end{defi}

\begin{exa}
  In $Σ₀\Trans$, with generating cofibrations consisting of all
  morphisms $ℒ(𝐲_{s_L})$, fibrations are precisely functional
  bisimulations, by Proposition~\ref{prop:bisimadj}.
\end{exa}

\begin{exa}
  In $\psh[ℂ_\source]$, with generating cofibrations consisting of all
  morphisms $𝐲_{s_L}$, fibrations are precisely functional
  bisimulations, by Definition~\ref{def:fbisim:Csource}.
\end{exa}

\begin{defi}
  For any categories with generating cofibrations $(𝒜,𝕁)$ and
  $(\psh,𝕂)$, such that $ℂ$ is small and $𝒜$ has a terminal object, a familial functor
  $F∶ 𝒜 → \psh$ is \emph{cellular} iff for all commuting squares
  \begin{equation}
    \diag{%
      C \& D \\
      F(X) \& F(Y)%
    }{%
      (m-1-1) edge[labela={k}] (m-1-2) %
      edge[labell={ξ}] (m-2-1) %
      (m-2-1) edge[labelb={F(δ)}] (m-2-2) %
      (m-1-2) edge[labelr={χ}] (m-2-2) %
    }
    \label{eq:gensquare}
  \end{equation}
  with $k ∈ 𝕂$ and $ξ$ and $χ$ generic,
  $δ$ is a cofibration (i.e., $δ ∈ \wbotrightleft{𝕁}$).
\end{defi}

Before exploiting cellularity as promised, let us briefly pause to
give an equivalent characterisation of cellularity in suitably nice
cases.

\begin{defi}
  Let us consider any familial functor $F∶ 𝒜 → \psh$, say as
  $F(A)(c) = ∑_{o ∈ O(c)} 𝒜(E(c,o),A)$, such that $ℂ$ is small and $𝒜$
  has a terminal object.  Then, for any operation $o ∈ F(1)(d)$ and
  morphism $s∶ c → d$ in $ℂ$, the \emph{$s$-arity} of $o$, or its
  \emph{boundary arity} when $s$ is clear from context, is the
  morphism $E(d,o↾s)∶ E(c,o·s) → E(d,o)$ in $𝒜$.
\end{defi}

\begin{nota}
Let $χₒ$ denote the generic morphism $𝐲_d → F(E(d,o))$
induced by any operation $o ∈ F(1)(d)$.
\end{nota}
By construction, for any $o$ and $s$ as in the definition, if $𝐲ₛ$ is
a generating cofibration (i.e., is in $𝕂$), then the diagram
  \begin{center}
    \diag(.6,2){%
      𝐲_c \& 𝐲_d \\
      F(E(c,o·s)) \& F(E(d,o))%
    }{%
      (m-1-1) edge[labela={𝐲ₛ}] (m-1-2) %
      edge[labell={χ_{o·s}}] (m-2-1) %
      (m-2-1) edge[labelb={F(E(d,o↾s))}] (m-2-2) %
      (m-1-2) edge[labelr={χₒ}] (m-2-2) %
    }
  \end{center}
  is of the form~\eqref{eq:gensquare} by construction.  We thus have
  by definition:
  \begin{prop}\label{prop:karities}
    Let us consider any cellular $F∶ 𝒜 → \psh$ between categories with
    generating cofibrations $(𝒜,𝕁)$ and $(\psh,𝕂)$ such that $ℂ$ is
    small and $𝒜$ has a terminal object.  Then, for any representable
    morphism $𝐲ₛ∶ 𝐲_c → 𝐲_d$ in $𝕂$, the $s$-arity of any operation
    $o ∈ F(1)(d)$ is a cofibration.
  \end{prop}

  When all domains and codomains of generating cofibrations $k ∈ 𝕂$
  are representable, say as $𝐲ₛ∶ 𝐲_c → 𝐲_d$, then all generic
  morphisms $χ∶ D → F(Y)$, with $D$ a codomain of some generating
  cofibration, are isomorphic to some $𝐲_d → F(E(d,o))$, for some
  operation $o ∈ F(1)(d)$, and the morphism $δ$ must then be
  isomorphic to $E(d,o↾s)∶ E(c,o·s) → E(d,o)$.  We thus get a partial
  converse to Proposition~\ref{prop:karities}:
\begin{prop}\label{prop:cellular:boundaries}
  For any categories with generating cofibrations $(𝒜,𝕁)$ and
  $(\psh,𝕂)$, such that $ℂ$ is small, $𝒜$ has a terminal object, and
  all domains and codomains of generating cofibrations $k ∈ 𝕂$ are
  representable, a familial functor $F∶ 𝒜 → \psh$ is cellular iff the
  boundary arities of all operations are cofibrations.
\end{prop}

Returning to the general case, let us now prove a first
characterisation of preservation of fibrations in terms of
cellularity.
\begin{lem}\label{lem:characfib}
  For any categories with generating cofibrations $(𝒜,𝕁)$ and
  $(\psh,𝕂)$ such that $ℂ$ is small and $𝒜$ has a terminal object, a
  familial functor $F∶ 𝒜 → \psh$ preserves fibrations iff it is
  cellular.
\end{lem}

\begin{full}
\begin{proof}
  Let us first prove the `if' direction.  We must show that for any
  $f∶ A → B$ in $𝕁^⋔$, $F(f)$ is in $𝕂^⋔$, i.e., that any commuting
  square
\begin{equation}
  \diag{%
    C \& F(A) \\
    D \& F(B) 
  }{%
    (m-1-1) edge[labela={u}] (m-1-2) %
    edge[labell={k}] (m-2-1) %
    (m-2-1) edge[labelb={v}] (m-2-2) %
    (m-1-2) edge[labelr={F(f)}] (m-2-2) %
  }
  \label{eq:myliftingpb}
\end{equation}
with $k ∈ 𝕂$ and $f ∈ 𝕁^⋔$ admits a lifting.  But taking generic
factorisations of both horizontal morphisms and using genericness, any
such square factors as the solid part of
\begin{center}
  \diag{%
    C \& F(X) \& F(A) \\
    D \& F(Y) \& F(B)\rlap{.}
  }{%
    (m-1-1) edge[labela={ξ}] (m-1-2) %
    edge[labell={k}] (m-2-1) %
    (m-2-1) edge[labelb={χ}] (m-2-2) %
    (m-1-2) edge[labelr={F(δ)}] (m-2-2) %
    (m-1-2) edge[labela={F(φ)}] (m-1-3) %
    (m-2-2) edge[labelb={F(ψ)}] (m-2-3) %
    edge[dashed,labelon={F(l)}] (m-1-3) %
    (m-1-3) edge[labelr={F(f)}] (m-2-3) %
  }
\end{center}
By cellularity, we have $δ ∈ \wbotrightleft{𝕁}$. We thus find a lifting
$l$ as shown, which makes $F(l) ∘ χ$ into a lifting for the original
square.

Conversely, let us assume that $F$ preserves fibrations, and consider
any square of the form~\eqref{eq:gensquare} with $k ∈ 𝕂$. We need to
show $δ ∈ \wbotrightleft{𝕁}$. But for any commuting square as below
left
  \begin{center}
    \diag{%
      X \& A \\
      Y \& B %
    }{%
      (m-1-1) edge[labela={φ}] (m-1-2) %
      edge[labell={δ}] (m-2-1) %
      (m-2-1) edge[labelb={ψ}] (m-2-2) %
      (m-1-2) edge[labelr={f}] (m-2-2) %
    }
    \hfil
    \diag{%
      C \& F(X) \& F(A) \\
      D \& F(Y) \& F(B)
    }{%
      (m-1-1) edge[labela={ξ}] (m-1-2) %
      edge[labell={k}] (m-2-1) %
      (m-2-1) edge[labelb={χ}] (m-2-2) %
      (m-1-2) edge[labellat={F(δ)}{.3}] (m-2-2) %
      (m-1-2) edge[labela={F(φ)}] (m-1-3) %
      (m-2-2) edge[labelb={F(ψ)}] (m-2-3) %
      edge[dashed,fore,labelon={F(l)}] (m-1-3) %
      (m-1-3) edge[labelr={F(f)}] (m-2-3) %
      (m-2-1) edge[dashed,fore,labelonat={γ}{.3}] (m-1-3)
    }
  \end{center}
  with $f ∈ 𝕁^⋔$, by pasting this square with our generic
  square~\eqref{eq:gensquare}, we obtain the solid part above right.
  Finally, because $F$ preserves fibrations, we find a lifting $γ$ as
  shown, which by genericness of $χ$ (and then $ξ$) yields the desired
  lifting $l$.
\end{proof}
\end{full}

\begin{cor}\label{cor:characfib2}
  For any categories with generating cofibrations $(𝒜,𝕁)$ and
  $(\psh,𝕂)$, such that $ℂ$ is small, $𝒜$ has a terminal object, and
  all domains and codomains of generating cofibrations $k ∈ 𝕂$ are
  representable, a familial functor $F∶ 𝒜 → \psh$ preserves fibrations
  iff the boundary arities of all operations are cofibrations.
  
\end{cor}
\begin{proof}
  By Proposition~\ref{prop:cellular:boundaries} and
  Lemma~\ref{lem:characfib}.
\end{proof}

Here is the announced characterisation of preservation of functional
bisimulations, which follows directly from Lemma~\ref{lem:characfib}
and Corollary~\ref{cor:characfib2}.
\begin{cor}\label{cor:cellularenough}
  In any Howe context, for any operational semantics signature
  $(Σ₀,Σ₁)$, let $Σ₁^\source∶ Σ₀\Trans → \psh[ℂ_\source]$ be familial
  with exponent $E∶ \el(Σ₁^\source(1)) → Σ₀\Trans$. Then the following are equivalent:
  \begin{enumerati}
  \item \label{item:pf} $Σ₁^\source$ preserves functional bisimulations;
  \item \label{item:cell} $Σ₁^\source$ is cellular;
  \item \label{item:explicit} the boundary arities of all operations
    are cofibrations, i.e., for all $L ∈ ℂ₁$ and
    $o ∈ Σ₁^\source(1)(L)$, the morphism
  \[E(s_L ↾ o)∶ E(\source(L),o·s_L) → E(L,o)\] is a cofibration.
\end{enumerati}
\end{cor}

This characterisation of preservation of functional bisimulations in
terms of cofibrations is easier to prove in practice, since
cofibrations in turn admit the following well-known characterisation.

\begin{defi}
  Consider any set $𝕁$ of maps in a given category.
  \begin{itemize}
  \item A \emph{basic relative $𝕁$-cell complex} is any morphism $f$ obtained by pushing out some morphism from $𝕁$ along any morphism\shortfull{.}{, as in
  \begin{center}
    \Diag{%
      \stdpo %
    }{%
      A \& B \\
      C \& D\rlap{.} %
    }{%
      (m-1-1) edge[labela={j ∈ 𝕁}] (m-1-2) %
      edge[labell={}] (m-2-1) %
      (m-2-1) edge[labelb={f}] (m-2-2) %
      (m-1-2) edge[labelr={}] (m-2-2) %
    }
  \end{center}
}
\item A \emph{relative $𝕁$-cell
    complex} is a (potentially transfinite) composite of basic relative $𝕁$-cell complexes.
  \end{itemize}
\end{defi}

  \begin{propC}[{\cite[Lemma~2.1.10]{Hovey}}]\label{prop:cellscof}
    For any set $𝕁$ of maps in a locally presentable category, all
    relative $𝕁$-cell complexes are cofibrations in the generated
    weak factorisation system $(\wbotrightleft{𝕁},\wbotright{𝕁})$.
  \end{propC}

\section{Applications}  \label{s:applications}
In this section, we apply our results to show that substitution-closed
bisimilarity is a congruence in concrete examples.
  
\subsection{Call-by-name}
We have already specified the syntax (§\ref{s:exsyntax}) and
transitions (§\ref{ss:opsem}) of call-by-name $λ$-calculus.  We have
also seen (Example~\ref{ex:cbn:familial}) that the induced functor
$Σ₁^\source∶ Σ₀\Trans → \psh[ℂ₁]/Δ_{\source}$ is familial.  By
Theorem~\ref{thm:main}, Corollary~\ref{cor:cellularenough}, and
Proposition~\ref{prop:cellscof}, congruence of substitution-closed
bisimilarity will follow if we prove that the boundary arities
$E(s_⇓ ↾ o)$ corresponding to both transition rules are relative cell
complexes.

The boundary arity $E(s_⇓ ↾ λ\text{-}\mathrm{val})$ is an identity,
hence trivially a relative cell complex.

For the $β$-rule, it is not entirely trivial that the
boundary arity
$E(s_⇓ ↾ β\text{-}\mathrm{red})$ of the second
transition rule, defined as the composite~\eqref{eq:Ebetasource}, is a
cofibration.  But, as (essentially) noted in~\cite[Example~5.21]{BHL},
it is a relative cell complex by construction, as both components are
pushouts of $ℒ(𝐲_{s_⇓})$, as should be clear from the following
diagram, where again $\bar{χ}$ is as in Definition~\ref{def:chibar}.
    \begin{center}
      \diagramme[diagorigins={1}{4}]{}{%
        \pullbackk(6pt){m-2-2}{m-2-3}{m-1-3}{draw,-} %
        \path[->,draw]
        (m-1-2) edge[labell={ℒ(𝐲_{t_⇓}+𝐲₀) ∘ \bar{χ}}] (m-2-2) %
        edge[labela={ℒ(𝐲_{s_⇓})}] (m-1-3) %
        (m-1-3) edge[labela={}] (m-2-3) %
      ; %
    }{%
      \&  ℒ(𝐲₀)   \& ℒ(𝐲_⇓)  \\
      ℒ(𝐲₀+𝐲₀) \& ℒ(𝐲_⇓+𝐲₀) \& Eᵦ
    }{%
      (m-2-1) edge[labelb={ℒ(𝐲_{s_⇓}+𝐲₀)}] (m-2-2) %
      (m-2-2) edge[labela={}] (m-2-3) %
    }
    \end{center}

\subsection{Call-by-value}\label{ss:cbv}
Let us now treat the call-by-value variant of untyped $λ$-calculus,
essentially as in~\cite{OngApplicative,Pitts:howe}.  In this setting,
it is important to distinguish substitution by values and by
terms. Indeed, letting $I$ denote the identity $λx.x$, the terms
$e = λx.I$ and $e' = λx.((λy.I)\ x)$ are contextually equivalent,
since during evaluation in any context, the bound variable $x$ will only
be replaced with a value. However, if one defines applicative
bisimulation naively, i.e., requiring it to be closed under arbitrary
substitution, then $e$ and $e'$ are not bisimilar, as $I$ is not
bisimilar to $(λy.I)\ Ω$ -- which diverges. We thus want to restrict
to value substitution -- which our treatment in~\cite{BHL} overlooks!
    
Here is one way of doing this.  The idea is to have two sorts, one for
values and the other for general terms.  We thus would like
$\psh[ℂ₀]$ to be equivalent to the category $[𝐒𝐞𝐭²,𝐒𝐞𝐭²]_f$
of finitary endofunctors on $𝐒𝐞𝐭²$. But
    \[[𝐒𝐞𝐭²,𝐒𝐞𝐭²]_f ≃ [ 2 · 𝐒𝐞𝐭_f², 𝐒𝐞𝐭 ] ≃ [ 2 · 𝔽², 𝐒𝐞𝐭 ],\]
    so we take $ℂ₀$ to be the opposite of $2 · 𝔽²$.  By the
    equivalence, composition of finitary endofunctors equips
    $\psh[ℂ₀]$ with monoidal structure, and we denote the two sorts by
    $𝐩$ and $𝐯$, respectively for ``programs'' and ``values''.  In the
    presheaf point of view, we denote by $(m,n)_𝐩$ and $(m,n)_𝐯$,
    respectively, objects in the first and second term of
    $2×𝔽^2 = 𝔽² + 𝔽²$, so that $X(m,n)(c)$ in the finitary endofunctor
    world corresponds to $X(m,n)_c$ in the presheaf world, for
    $c ∈ \ens{𝐯,𝐩}$.  We think of $X(m,n)_𝐯$ (resp.\ $X(m,n)_𝐩$) as a
    set of values (resp.\ programs) with $m$ potential free program
    variables, and $n$ potential free value variables.

    Since abstraction should bind a value variable, the syntax
    \[
    \begin{array}{rcll}
      e,f & ::= & e\ f ｜ v \\
      v &  ::= & x ｜ λx.e
    \end{array}\]
    is thus
    specified by
    \begin{center}
      $Σ₀(X)(m, n)(𝐩) = X(m, n)(𝐩)² + X(m, n)(𝐯)$ \hfil and \hfil
      $Σ₀(X)(m, n)(𝐯) = X(m,n+1)(𝐩)$.
    \end{center}
    A model $X$ of the syntax should thus in particular feature
    operations
    \begin{center}
      $\mathrm{abs}_{m,n}∶ X(m,n+1)(𝐩) → X(m, n)(𝐯)$ \hfil
      $\mathrm{app}_{m,n}∶ X(m, n)(𝐩)² → X(m, n)(𝐩)$ \hfil
      $\mathrm{val}_{m,n}∶ X(m, n)(𝐯) → X(m, n)(𝐩)$,      
    \end{center}
    where the last operation requires that values should embed into
    programs.

    \begin{nota}
      We implicitly view any $X ∈ [2·𝔽²,𝐒𝐞𝐭]$ as (some fixed, global
      choice of) the corresponding functor $𝐒𝐞𝐭² → 𝐒𝐞𝐭²$.  In
      particular, we write $X(K)$, for any $K ∈ 𝐒𝐞𝐭²$. Accordingly, we
      view pairs $(m,n)$ as objects of $𝐒𝐞𝐭²$. E.g., in this sense,
      $(m,n+1)$ is isomorphic to $(m, n) + 𝐲_𝐯$, where $𝐲_𝐯$ denotes
      the Yoneda embedding of $𝐯$ along $2 → 𝔽² ↪ 𝐒𝐞𝐭²$.  So the arity
      for abstraction in fact yields an operation
      $X(K + 𝐲_𝐯)(𝐩) → X(K)(𝐯)$.
  \end{nota}
    Denoting composition of finitary endofunctors by $⊗$, we define a
    pointed strength for $Σ₀$ as follows. For any $X ∈ \psh[ℂ₀]$,
    $Y ∈ I/\psh[ℂ₀]$, and $K ∈ 𝔽²$:
    \begin{itemize}
    \item at $𝐩$, we have
      \[(Σ₀(X)⊗Y)(K)(𝐩) = X(Y(K))(𝐩)² + X(Y(K))(𝐯) =
      Σ₀(X⊗Y)(K)(𝐩)\rlap{,}\]
      so we take $st_{X,Y,K,𝐩}$ to be the identity;
    \item at $𝐯$, we have 
      \begin{center}
        $(Σ₀(X)⊗Y)(K)(𝐯) = X(Y(K) + 𝐲_𝐯)(𝐩)$ \hfil and \hfil
        $Σ₀(X⊗Y)(K)(𝐯) = X(Y(K + 𝐲_𝐯))(𝐩)\rlap{,}$
      \end{center}
      so we define $st_{X,Y,K,𝐩}$ by applying
      $X(-)(𝐩)$ to the copairing of $Y(K) → Y(K+𝐲_𝐯)$ and
      \[𝐲_𝐯 → I(K+𝐲_𝐯) → Y(K+𝐲_𝐯).\]
    \end{itemize}

    We then specify transitions, which we first recall:
    \begin{mathpar}
      \inferrule{ }{λx₁.e ⇓ e} \and
      \inferrule{e₁ ⇓ e'₁ \\ e₂ ⇓ e'₂ \\ e'₁[λx₁.e'₂] ⇓ e₃}{e₁\ e₂ ⇓ e₃}~·
    \end{mathpar}
    \begin{rem}
      Here, we adopt the same convention as in the categorical
      picture, where terms are implicitly considered as indexed over
      sets, say $\ens{x₁,…,xₙ}$, of potential free variables.
      Furthermore, the evaluation relation again relates closed terms
      to terms over one potential free variable $x₁$.
  \end{rem}
    In order to specify such transitions, we take $ℂ₁ = 1$, as in the
    call-by-name case, with $\source$ and $\but$ mapping the unique
    object to $(0,0)_𝐩$ and $(0,1)_𝐩$, respectively: transitions will
    relate a closed program to a program with one value variable
    (morally the body of the obtained abstraction).  For any
    $X ∈ 𝐂 = \psh[ℂ₁]/Δ ≅ 𝐒𝐞𝐭/Δ$, We take $Σ₁^F(X)$ to be the
    coproduct \[Σ₁^F(X) = X₀(0,1) + A_{βᵥ}(X)\rlap{,}\] where
    $A_{βᵥ}(X)$ denotes the set of valid premises for the second rule,
    i.e., triples $(r₁,r₂,r₃) ∈ X₁$ such that
    \[r₃·s_⇓ = (r₁·t_⇓)[λ(r₂·t_⇓)].\]

    Familiality of the induced functor $Σ₁^\source∶ Σ₀\Trans → \psh[ℂ₁]/Δ_{\source}$ follows
    similarly to the call-by-name case, from the fact that $A_{βᵥ}(X)$ is isomorphic
    to $[E_{βᵥ},X]$, where $E_{βᵥ}$ denotes the following pushout.
    \begin{center}
      \diagramme[diagorigins={1}{4}]{}{%
        \pullbackk(6pt){m-2-1}{m-2-2}{m-1-2}{draw,-} %
        \path[->,draw]
        (m-1-1) edge[labell={ξ}] (m-2-1) %
        edge[labela={ℒ(𝐲_{s_⇓})}] (m-1-2) %
        (m-1-2) edge[labela={}] (m-2-2) %
      ; %
    }{%
      ℒ(𝐲₀)   \& ℒ(𝐲_⇓)  \\
      ℒ(𝐲_⇓+𝐲_⇓) \& E_{βᵥ}
    }{%
      (m-2-1) edge[labela={}] (m-2-2) %
    }
    \end{center}
    Here, $ξ$ corresponds by adjunction and Yoneda to the closed term
    \[(r₁·t_⇓)[λ(r₂·t_⇓)] ∈ 𝒰 (ℒ(𝐲_⇓+𝐲_⇓))(0)\rlap{,}\] where $r₁$ and $r₂$
    denote the two transition constants generating $𝐲_⇓+𝐲_⇓$.

    By Theorem~\ref{thm:main}, Corollary~\ref{cor:cellularenough}, and
    Proposition~\ref{prop:cellscof}, congruence of substitution-closed
    bisimilarity will follow if we can prove that boundary arities
    $E(s_L ↾ o)$ of both transition rules are relative
    cell complexes. This is again trivial for the first rule, while
    for the second one we obtain the following morphism, which is a
    relative cell complex by construction.
    \begin{center}
      \diagramme[diagorigins={1}{4}]{}{%
        \pullbackk(6pt){m-2-2}{m-2-3}{m-1-3}{draw,-} %
        \path[->,draw]
        (m-1-2) edge[labell={ξ}] (m-2-2) %
        edge[labela={ℒ(𝐲_{s_⇓})}] (m-1-3) %
        (m-1-3) edge[labela={}] (m-2-3) %
      ; %
    }{%
      \&  ℒ(𝐲₀)   \& ℒ(𝐲_⇓)  \\
      ℒ(𝐲₀+𝐲₀) \& ℒ(𝐲_⇓+𝐲_⇓) \& E_{βᵥ}
    }{%
      (m-2-1) edge[labelb={ℒ(𝐲_{s_⇓}+𝐲_{s_⇓})}] (m-2-2) %
      (m-2-2) edge[labela={}] (m-2-3) %
    }
    \end{center}

\subsection{Erratic non-determinism}
In this section, we consider the non-deterministic $λ$-calculus
investigated in~\cite[§7]{DavideLazy}.  Its syntax is that of pure
$λ$-calculus, augmented with a unary operation $⊎$, and its reduction
rules~\cite[on pages 125 and 142]{DavideLazy} are
\begin{mathpar}
      \inferrule{ }{{{⊎}e} ⇒ e}  \and
      \inferrule{ }{{{⊎}e} ⇒ Ω} \and
      \inferrule{ }{(λx.e)\ f ⇒ e[f]} \\
      \inferrule{e₁ ⇒ e₃}{e₁\ e₂ ⇒ e₃\ e₂} \and
      \inferrule{ }{e ⇒ e} \and
      \inferrule{e₁ ⇒ e₂ \\ e₂ ⇒ e₃}{e₁ ⇒ e₃},
    \end{mathpar}
where $Ω$ denotes any diverging term.

Let us start by giving a big-step presentation of this language.  We
consider a labelled transition relation, i.e., we have two transition
relations $⇓_λ$ (between a closed term and a term with one free
variable, as before) and $⇓_τ$ (between closed terms), inductively
generated by the following rules.
\begin{mathpar} \inferrule{ }{{{⊎}e} ⇓_τ e} \and %
  \inferrule{ }{{{⊎}e} ⇓_τ Ω} \and %
  \inferrule{ }{λx.e ⇓_λ e} \and %
  \inferrule{e₁ ⇓_λ e₃ \\ e₃[e₂] ⇓_λ e₄}{e₁\ e₂ ⇓_λ e₄} \and %
  \inferrule{e₁ ⇓_λ e₃ \\ e₃[e₂] ⇓_τ e₄}{e₁\ e₂ ⇓_τ e₄} \and
  \inferrule{e₁ ⇓_τ e₃}{e₁\ e₂ ⇓_τ e₃\ e₂} \and %
  \inferrule{ }{e ⇓_τ e} \and %
  \inferrule{e₁ ⇓_τ e₂ \\ e₂ ⇓_τ e₃}{e₁ ⇓_τ e₃} \and %
  \inferrule{e₁ ⇓_τ e₂ \\ e₂ ⇓_λ e₃}{e₁ ⇓_λ e₃}~·
\end{mathpar}

We have:
\begin{prop}
  A relation is an applicative bisimulation in Sangiorgi's sense, say
  a \emph{Sangiorgi bisimulation}, iff its open extension is a
  substitution-closed bisimulation with the new rules.
\end{prop}
\begin{proof}
This is an easy corollary of the following lemma.
\end{proof}
\begin{lem}
  For all closed $e₁$ and $e₂$, and $e₃$ with one potential free variable, we have
  \begin{itemize}
  \item $e₁ ⇒ e₂$ iff $e₁ ⇓_τ e₂$, and
  \item $e₁ ⇒ λ(e₃)$ iff $e₁ ⇓_λ e₃$.
  \end{itemize}
\end{lem}
\begin{proof}\hfill
  \begin{itemize}
  \item All reduction rules also occur as some new rule for $⇓_τ$,
    except the $β$-rule, which is easily derivable. We thus have
    ${⇒} ⊆ {⇓_τ}$.
  \item Similarly, if $e₁ ⇒ λ(e₃)$, then by the previous point
    $e₁ ⇓_τ λ(e₃)$, hence by the last rule $e₁ ⇓_λ e₃$.
  \item Finally, let us prove both converse statements in one go by
    induction on the transition proof.
    \begin{itemize}
    \item The first two axioms also are axioms in the original rules, hence easy.
    \item For the third axiom, we are given $λ(e) ⇓_λ e$, hence
      clearly $λ(e) ⇒ λ(e)$ as desired.
    \item For the first $β$ rule, $e₁ = e₄\ e₅$, and we know by
      induction hypothesis that $e₄ ⇒ λ(e₆)$ and $e₆[e₅] ⇒ λ(e₃)$. We
      thus get $e₄\ e₅ ⇒ λ(e₆)\ e₅ ⇒ e₆[e₅] ⇒ λ(e₃)$, hence the
      desired result.
      
    \item The second $β$ rule is similar, except that we get a chain
      $e₄\ e₅ ⇒ λ(e₆)\ e₅ ⇒ e₆[e₅] ⇒ e₂$.
    \item The next three rules also occur as original rules, hence are
      easily dealt with by induction hypothesis.
    \item For the last rule, we have by induction hypothesis $e₁ ⇒ e₄$
      and $e₄ ⇒ λ(e₃)$, hence $e₁ ⇒ λ(e₃)$, as desired. \qedhere
    \end{itemize}
  \end{itemize}
\end{proof}

The syntax is easily modelled by taking
$Σ₀(X)(n) = X(n+1) + X(n)² + X(n)$, and the new rules are easily seen
to fit into a $Σ₁$ such that $Σ₁^\source$ preserves functional
bisimulations by Corollary~\ref{cor:cellularenough} and
Proposition~\ref{prop:cellscof}. Finally, substitution-closed
bisimilarity in the initial model coincides with Sangiorgi's
applicative bisimilarity, so we again deduce congruence of applicative
bisimilarity by Theorem~\ref{thm:main}.

\subsection{Howe's format}\label{ss:wow:format}
In this section, we show that, up to suitable encoding, our framework
covers languages complying with a format proposed by Howe~\cite[Lemma
6.1]{DBLP:journals/iandc/Howe96} --- or rather a slight variant
thereof (see Remark~\ref{rk:discrepance-howe} below).  We first introduce
the format, and then explain how to embed it into our framework.  The
format illustrates in which sense cellularity is in particular an
acyclicity condition.  Furthermore, it would be easy, but much more
verbose, to extend the format to a simply-typed setting; we refrain
from doing so in this already quite long paper.

\subsubsection{Recalling Howe's format}
Regarding syntax, Howe's framework is parameterised by the choice
between call-by-name and call-by-value.  His signatures are then much
like standard binding signatures~\cite{BindingSignatures}.  We model
this in a setting similar to the call-by-value setting
of~§\ref{ss:cbv}, except that this time we build into the base
category the fact that values embed into terms.  Specifically, we work
with the monoidal category $[𝐒𝐞𝐭^𝟚,𝐒𝐞𝐭^𝟚]_f ≃ [𝟚×𝔽^𝟚,𝐒𝐞𝐭]$, where $𝟚$
denotes the walking arrow category $𝐯 \xrightarrow{ι} 𝐩$.  We think of
an object of the base category $𝟚×𝔽^𝟚$, which is a pair $(s,f∶ m→n)$
with $s ∈ \{𝐯,𝐩\}$, as the index of terms of sort $s$, with $m$ value
variables and $n$ program variables, each $f(i)$ being thought as the
program counterpart of $i ∈ m$.

Let us first explain how the syntactic part of Howe's format may be
understood in terms of pointed strong endofunctors.  Howe's notion of
syntactic signature encodes a choice between call-by-value and
call-by-name, through the choice of a distinguished sort, $𝐯$ or $𝐩$.
\begin{nota} \ \hfill
  \begin{itemize}
  \item 
  Given $m,n∈ 𝔽$ and $s ∈ \{𝐯,𝐩\}$, 
  we denote any object $(s,(m\xrightarrow{f}n))$ of $𝟚×𝔽^𝟚$ by
  $(m\xrightarrow{f}n)ₛ$,
  and the particular morphism $(ι, \id)∶
  (m\xrightarrow{f}n)_𝐯 →
  (m\xrightarrow{f}n)_𝐩
  $ by $ι$.
\item Given $m,n∈ 𝔽$, we denote the object $(m\xrightarrow{in₁}m+n)$
  of $𝔽^𝟚$ by $(m,n)$.  Otherwise said, we treat the corresponding
  embedding $𝔽² ↪ 𝔽^𝟚$ as an implicit coercion.
\item Accordingly, for any $m,n∈𝔽$ and $s ∈ \{𝐯,𝐩\}$, we denote
  $(m\xrightarrow{in₁}m+n)ₛ$ by $(m,n)ₛ$.
\item Given $T∶ 𝟚×𝔽^𝟚 → 𝐒𝐞𝐭$ and $s∈ \{𝐯,𝐩\}$, we denote by
  $Tₛ ∶ 𝔽^𝟚 → 𝐒𝐞𝐭$ the functor mapping $(m\xto{f}n)$ to $T(m\xto{f}n)ₛ$.
\end{itemize}
\end{nota}
\begin{rem}
  \label{rk:howe-format-mon-prod}
  Through the equivalence $[𝐒𝐞𝐭^𝟚,𝐒𝐞𝐭^𝟚]_f ≃ [𝟚×𝔽^𝟚,𝐒𝐞𝐭]$,
  composition of endofunctors becomes a monoidal product
  defined by
  \[(F⊗G)(m\xrightarrow{f}n)ₛ = \Lan_J (Fₛ) (G(m\xrightarrow{f}n)_𝐯 \xrightarrow{Gι}
  G(m\xrightarrow{f}n)_𝐩)\rlap{,}\] where $J: 𝔽^𝟚 → 𝐒𝐞𝐭^{𝟚}$ is the
  canonical embedding. The unit $I$ is defined
  by \[I(n_𝐯\xto{f}n_𝐩)ₛ = nₛ.\] for any object
  $(n_𝐯\xto{f}n_𝐩)ₛ$. Intuitively, $I$ merely returns the set
  of variables of each sort.
\end{rem}
Howe distinguishes value operations\footnote{Value operations correspond
  to canonical operators in~\cite{DBLP:journals/iandc/Howe96}.}
from program operations:
\begin{defi}
  A \emph{Howe binding signature} consists of 
  \begin{itemize}
  \item a \emph{binding sort} $sᵥ ∈ \ens{𝐯,𝐩}$,
  \item a set $O_𝐩$ of \emph{program operations}, equipped with
    a map $N^𝐩∶ O_𝐩 → ℕ$ and a family
    $d^𝐩 ∈ ∏_{o ∈ O_𝐩} ℕ^{N^{𝐩}ₒ}$, and
  \item a set $O_𝐯$ of \emph{value operations}, equipped with two maps
    $N⁺∶ O_𝐯 → ℕ$ and $N⁻∶ O_𝐯 → ℕ$, and a family
    $d⁻ ∈ ∏_{o ∈ O_𝐯} ℕ^{N⁻ₒ}$.
  \end{itemize}
\end{defi}
\begin{terminology}\label{term:howesignatures}
  For any operation $o ∈ O_𝐩$ (resp.\ $O_𝐯$), the sequence
  $(d^𝐩_{o,1},…,d^𝐩_{o,N^𝐩ₒ})$ (resp.\ 
  $(0,…,0,d⁻_{o,1},…,d⁻_{o,N⁻ₒ})$, with $N⁺ₒ$ leading $0$s) is called the
  \emph{(binding) arity} of $o$.
Typically, $λ$-abstraction and pairing are value operations.
The numbers $N⁺ₒ$ and $N⁻ₒ$ respectively
count
\begin{itemize}
\item \emph{active} arguments which should be evaluated before
  reaching a value, as both arguments of the pairing operation, and
\item \emph{passive} arguments which are not evaluated until the
  operation is destroyed, as the argument of a $λ$-abstraction.
\end{itemize}
Active arguments are not allowed to bind any variable, hence the
absence of a family $d⁺$ in signatures.
\end{terminology}

Despite their name, value operations may or may not return values,
depending on the status of their arguments: a pair $⟨e₁,e₂⟩$ is only a
value when both $e₁$ and $e₂$ are.  In order to generate the right
syntax in our strongly-sorted setting, we need to introduce two
incarnations of each value operation, one for constructing values, the
other for constructing programs.
Any Howe binding signature generates a pointed
strong endofunctor, as follows.  We first define an auxiliary
operation for adding bound variables in the right component --- as
prescribed by the binding sort $sᵥ$.
\begin{nota}\label{not:pluss}
  For any $(m\xrightarrow{f}n) ∈ 𝔽^𝟚$, $p ∈ 𝔽$, and $s∈\{𝐯, 𝐩\}$ let
  $(m\xrightarrow{f}n) +_{s} p$ denote
  \begin{itemize}
  \item $(m+p \xrightarrow{f+\id_p}n+p)$ if $s = 𝐯$, and
  \item $(m \xrightarrow{f}n ↪ n+p)$ if $s = 𝐩$.
  \end{itemize}
\end{nota}
\begin{defi}\label{def:generated}
  Given any Howe binding signature $B$, the \emph{generated}
  endofunctor $Σ₀ᴮ$ on $[𝟚×𝔽^𝟚,𝐒𝐞𝐭]$ is defined by 
\[
\begin{array}{rcl}
  Σ₀ᴮ(F)(m\xrightarrow{f}n)_𝐯 & = & ∑_{o ∈ O_𝐯} F((m\xrightarrow{f}n)_𝐯)^{N⁺ₒ} × ∏_{i ∈ N⁻ₒ} F((m\xrightarrow{f}n)
                        +_{sᵥ} d⁻_{o,i})_𝐩\rlap{,} \\
  Σ₀ᴮ(F)(m\xrightarrow{f}n)_𝐩 & = &  
                                   ∑_{o ∈ O_𝐯} F((m\xrightarrow{f}n)_𝐩)^{N⁺ₒ} × ∏_{i ∈ N⁻ₒ} F((m\xrightarrow{f}n) +_{sᵥ} d⁻_{o,i})_𝐩\\  
                              &&{} +  ∑_{o ∈ O_𝐩} ∏_{i ∈ N^𝐩ₒ} F((m\xrightarrow{f}n) +_{sᵥ} d^𝐩_{o,i})_𝐩 
\end{array}
\]
for all $F ∈ [𝟚×𝔽^𝟚,𝐒𝐞𝐭]$ and $m,n ∈ 𝔽$, with obvious action on morphisms.
\end{defi}
We now want to express $Σ₀ᴮ$ as a familial functor.
For this, we start by defining
\begin{itemize}
\item a functor $A^𝐩ₒ∶ \op{𝔽^𝟚} → [𝟚×𝔽^𝟚,𝐒𝐞𝐭]$ for
each program operation $o ∈ O_𝐩$,
\item a functor $A^𝐯ₒ∶ \op{(𝟚×𝔽^𝟚)} → [𝟚×𝔽^𝟚,𝐒𝐞𝐭]$ for
each value operation $o ∈ O_𝐯$,
\end{itemize}
in such a way that $Σ₀ᴮ(F)(m\xrightarrow{f}n)ₛ = ∑ₒ[Aˢₒ(m\xrightarrow{f}n),F]$.
\begin{defi}
  Given any Howe binding signature,
  \begin{itemize}
  \item for any program operation $o ∈ O_𝐩$, let the \emph{arity}
    of $o$ be the functor $A^𝐩ₒ∶ \op{(𝔽^𝟚)} → [𝟚×𝔽^𝟚,𝐒𝐞𝐭]$
    mapping $(m\xrightarrow{f}n)$ to the coproduct of representable presheaves
   \[A^𝐩ₒ(m\xrightarrow{f}n) = ∑_{i ∈ N^𝐩ₒ} 𝐲_{((m\ \xrightarrow{f}\ n) +_{sᵥ} d^𝐩_{o,i})_𝐩}\rlap{;}\]
 \item for any value operation $o ∈ O_𝐯$, let the \emph{arity} of $o$ be the
   functor $A^𝐯ₒ∶ \op{(𝟚×𝔽^𝟚)} → [2·𝔽,𝐒𝐞𝐭]$ mapping any object $(m\xrightarrow{f}n)ₛ$
   to the coproduct of representable presheaves
   \[A^𝐯ₒ(m\xrightarrow{f}n)ₛ = N⁺ₒ · 𝐲_{(m\xrightarrow{f}n)ₛ} + ∑_{i ∈ N⁻ₒ} 𝐲_{((m\xrightarrow{f}n) +_{sᵥ} d⁻_{o,i})_𝐩}.\]   
\end{itemize}
\end{defi}
\begin{rem}
  The arity functors are contravariant because the Yoneda embedding
  is, for covariant presheaves.
\end{rem}
By construction:
\begin{prop}\label{prop:sarity}
  The endofunctor $Σ₀ᴮ$ is familial, with
  \begin{itemize}
  \item as spectrum the functor
    \[Sᴮ ≅ Σ₀ᴮ(1) ≅ O_𝐩 · 𝐲_{(0→0)_𝐩} + O_𝐯 · 𝐲_{(0 → 0)_𝐯}\rlap{,}\]
    i.e.,
    \[\begin{array}[t]{rcl}
      Sᴮ(m\xrightarrow{f}n)_𝐩 & = & O_𝐩 + O_𝐯 \\
      Sᴮ(m\xrightarrow{f}n)_𝐯 & = & O_𝐯\rlap{,}
      \end{array}\]
      with obvious action on morphisms,
      so that \[\el(Sᴮ) = O_𝐩·\op{(𝔽^𝟚)} + O_𝐯·\op{(𝟚×𝔽^𝟚)}\rlap{,}\]
    \item and as exponent the functor $Eᴮ∶ \el(Sᴮ) → [𝟚×𝔽^𝟚,𝐒𝐞𝐭]$
      defined as the cotupling
        \[[[A^𝐩ₒ]_{o ∈ O_𝐩},[A^𝐯ₒ]_{o ∈ O_𝐯}].\]
  \end{itemize}
\end{prop}
\begin{proof}
  By mere calculation, using the Yoneda lemma.
\end{proof}

\begin{cor}
    We have, for all $F ∈ [𝟚×𝔽^𝟚,𝐒𝐞𝐭]_f$ and $(m\xrightarrow{f}n) ∈ 𝔽^𝟚$,
    \[Σ₀ᴮ(F)((m\xrightarrow{f}n)_𝐩) ≅ ∑_{o ∈ O_𝐩} [A^𝐩ₒ(m\xrightarrow{f}n), F] %
    + ∑_{o ∈ O_𝐯} [A^𝐯ₒ(m\xrightarrow{f}n)_𝐩, F]\] %
    and \[Σ₀ᴮ(F)(m\xrightarrow{f}n)_𝐯 ≅ ∑_{o ∈ O_𝐯} [ A^𝐯ₒ(m\xrightarrow{f}n)_𝐯, F].\]
\end{cor}

We now describe the syntax generated by $B$.

Let $𝒯₀ᴮ$ be the monad induced by the monadic forgetful functor
$Σ₀ᴮ\Mon → [𝟚×𝔽^𝟚,𝐒𝐞𝐭]$.

In order to define the format, we need to unfold $𝒯₀ᴮ(K)$, for $K$ a
coproduct of representable presheaves of the form $𝐲_{(n_𝐯,n_𝐩)ₛ}$.
For a single such presheaf, by Proposition~\ref{prop:freelambda},
$𝒯₀ᴮ(𝐲_{(n_𝐯,n_𝐩)ₛ})=μA.(I + Σ₀ᴮ(A) + 𝐲_{(n_𝐯,n_𝐩)ₛ}⊗A)$.  Unfolding
the definition of tensor product
(Remark~\ref{rk:howe-format-mon-prod}), noticing that the hom-set
$[s,s']$
\begin{itemize}
\item is empty iff $s = 𝐩$ and $s' = 𝐯$, and
\item is otherwise a singleton,
\end{itemize}
we obtain the following result.

\begin{prop}
  We have, for all $n_𝐯,n_𝐩,m,n ∈ ℕ$, $s,s' ∈ \ens{𝐯,𝐩}$,
  $A∶ 𝟚×𝔽^𝟚 → 𝐒𝐞𝐭$, and $f∶ m → n$:
\[
\begin{array}{rcl}
(𝐲_{(n_𝐯,n_𝐩)ₛ}⊗A)(m\xrightarrow{f}n)_{s'} & ≅ &
[s,s'] × A^{n_𝐩}(m\xrightarrow{f}n)_𝐩×A^{n_𝐯}(m\xrightarrow{f}n)_𝐯 \\
& ≅ &
\left \{
  \begin{array}{ll}
    ∅ & \mbox{if $s=𝐩$ and $s'=𝐯$} \\
    A^{n_𝐩}(m\xrightarrow{f}n)_𝐩×A^{n_𝐯}(m\xrightarrow{f}n)_𝐯
    & \mbox{otherwise.}
  \end{array}
\right .
\end{array}
\]
\end{prop}
\begin{rem}
  Intuitively, this will entail that $𝒯₀ᴮ(𝐲_{(n_𝐯,n_𝐩)ₛ})$ extends the
  syntax generated by $Σ₀ᴮ$ with a new operation taking $n_𝐩$ programs
  and $n_𝐯$ values as arguments, and returning a program or a value,
  depending on $s$.  As usual, if the output is a value, by action of
  $ι$, there is a corresponding program operation.
\end{rem}

We now describe $𝒯₀ᴮ(K)$, for $K$ any coproduct of representable presheaves
of the form $𝐲_{(n_𝐯,n_𝐩)ₛ}$.
Let us first introduce some notations.
\begin{defi}
  A \textbf{family of operation arities} is a family
  $c = (mᵢ,nᵢ,sᵢ)_{i ∈ I}$ of triples $(mᵢ,nᵢ,sᵢ)∈ℕ²×\{ 𝐯, 𝐩 \}$.  
\end{defi}
\begin{nota}
  We think of any such triple $cᵢ = (mᵢ,nᵢ,sᵢ)$ as an operation
  $cᵢ ∶ 𝐯^{mᵢ}×𝐩^{nᵢ} → sᵢ$, hence denote any such family by
  $(cᵢ ∶ 𝐯^{mᵢ}×𝐩^{nᵢ} → sᵢ)_{i ∈ I}$.
  We in fact extend this notation by 
    \begin{itemize}
    \item writing $𝐯^{m}× 𝐩^{n}×s^{q}$, 
      to denote either $𝐯^{m+q}×𝐩^{n}$, when $s = 𝐯$, or $𝐯^{m}×𝐩^{n+q}$, when
      $s = 𝐩$, 
    \item omitting $-ⁿ$ if $n = 0$ in the above product, e.g.,
      $𝐯ᵐ → s$ denotes $𝐯ᵐ×𝐩⁰ → s$, and finally
    \item writing $c∶ s$ for $c∶ 𝐯⁰×𝐩⁰ → s$.
    \end{itemize}
    We furthermore denote the disjoint union of families of operation
    arities $K$ and $L$ by $K+L$.  Given a family
    $K = ( cᵢ ∶ 𝐯^{mᵢ}×𝐩^{nᵢ} → sᵢ)ᵢ$ of operation arities, we also
    denote by $K$ the coproduct of representable presheaves
    $∑ᵢ 𝐲_{(mᵢ,nᵢ)_{sᵢ}} ∶ 𝟚×𝔽^𝟚 → 𝐒𝐞𝐭$.
\end{nota}
\begin{rem}
  By Notation~\ref{not:pluss}, for $m,n,p∈𝔽$, $(m,n)+ₛ p$ denotes
  $(m+p,n)$ if $s=𝐯$, and $(m,n+p)$ otherwise.
\end{rem}
We now give an inductive description of $𝒯₀ᴮ(K)$ on objects
of the form $(m,n)ₛ$ (Notation~\ref{not:pluss}).
Letting $m,n ⊢_K e : s$ mean that $e ∈
𝒯₀ᴮ(K)(m,n)ₛ$, for a family $K$ of operation arities,
the free $Σ₀ᴮ$-monoid over $K$ is inductively generated on objects of
the form $(m,n)ₛ$ by the following rules.
\begin{mathpar}
  \inferrule{ }{m,n ⊢_K xᵢˢ : s}~(i ∈ m) \and
  \inferrule{ }{m,n ⊢_K aᵢ : 𝐩}~(i ∈ n) \and
  \inferrule{m,n ⊢_K e₁ : 𝐯 \\ … \\ m,n ⊢_K e_{m'} : 𝐯 \\\\
  m,n ⊢_K f₁ : 𝐩 \\ … \\ m,n ⊢_K f_{n'} : 𝐩}{m,n ⊢_K k(e₁,…,e_{m'};f₁,…,f_{n'})
  : s}~((k ∶ 𝐯^{m'}×𝐩^{n'} → s) ∈ K) \\
  \inferrule{(m,n) +_{sᵥ} d^𝐩_{o,1} ⊢_K e₁ : 𝐩 \\ … \\ (m,n) +_{sᵥ} d^𝐩_{o,N^𝐩ₒ} ⊢_K e_{N^𝐩ₒ} : 𝐩}{%
    m,n ⊢_K o(e₁,…,e_{N^𝐩ₒ}) : 𝐩}~(o ∈ O_𝐩) \and
  \inferrule{
    m,n ⊢_K e₁ : s \\ … \\ m,n ⊢_K e_{N⁺ₒ} : s \\\\
    (m,n) +_{sᵥ} d⁻_{o,1} ⊢_K f₁ : 𝐩 \\ … \\ (m,n) +_{sᵥ} d⁻_{o,N⁻ₒ} ⊢_K f_{N⁻ₒ} : 𝐩 %
  }{%
    m,n ⊢_K oₛ(e₁,…,e_{N⁺ₒ};f₁,…,f_{N⁻ₒ}) : s}~(o ∈ O_𝐯)
\end{mathpar}
\begin{rem}
  Program operations $o ∈ O_𝐩$ only have one list of arguments.
  Furthermore, they always return programs, so there is no need to
  annotate them. By contrast, value operations $o' ∈ O_𝐯$ have two
  lists of arguments (active and passive arguments, see
  Terminology~\ref{term:howesignatures}), and may return values or
  programs, depending on the status of their active arguments. Thus,
  e.g., any unannotated operation application $o(e₁,…,eₙ)$ must be a
  program operation application, while any annontated operation
  application $oₛ(e₁,…,eₘ;f₁,…,fₙ)$ must be a value operation
  application.
\end{rem}

The action on morphisms is straightforward: for renaming, we rename
(value and program) variables accordingly; for $ι$, we replace each
$xᵢ^𝐯$ (resp.\ $o_𝐯$) with $xᵢ^𝐩$ (resp.\ $o_𝐩$). For morphisms
$K → L$, we proceed similarly.

\begin{terminology}
  We think of elements $k$ from $K$ as \emph{metavariables}, while terms
  of the form $aᵢ$ or $xᵢˢ$ are mere \emph{variables}.
\end{terminology}
\begin{nota}\ \hfill 
  \begin{itemize}
  \item Following~\cite{DBLP:conf/aplas/Hamana04}, for any $(k ∶ 𝐯ᵐ×𝐩ⁿ →
    s) ∈ K$, we abbreviate
    $k(x₁^𝐯,…,xₘ^𝐯;  a₁,…,aₙ )$ to just $k$ when $(m,n)$ is clear
    from context.
\item For value operations $o$, we sometimes omit the subscript
  $s$ in $oₛ(…)$, when the expected sort is clear from context.
  \item Similarly, we sometimes omit the exponent $s$ in variables $xᵢˢ$.
\item In metavariable application
  $k(e₁,…,e_{m'};f₁,…,f_{n'})$, as well as in value
  operation application $o(e₁,…,eₚ;f₁,…,f_q)$, when one sequence is
  empty we omit it altogether, writing, e.g., $k(e₁,…,e_{m'})$ or
  $o(f₁,…,f_q)$. When both lists are empty we simply write $k$, resp.\
  $o$.
\end{itemize}
Thus, e.g., $k$ may denote a nullary metavariable, or a non-nullary
metavariable with identity substitution.
\end{nota}

Let us now introduce Howe's notion of signature for evaluation rules,
restricting attention to signatures satisfying the syntactic condition
of~\cite[Lemma 6.1]{DBLP:journals/iandc/Howe96}.  Evaluation is a
binary relation between closed programs and closed values.  By
default, all value operations $o$ are considered as coming equipped
with their \emph{canonical} evaluation rule
\begin{equation}
  \inferrule{%
    e₁ ⇓ v₁ \\ … \\ e_{N⁺ₒ} ⇓ v_{N⁺ₒ}
  }{%
    o_𝐩(e₁,…,e_{N⁺ₒ};f₁,…,f_{N⁻ₒ})
    ⇓ o_𝐯(v₁,…,v_{N⁺ₒ};f₁,…,f_{N⁻ₒ})~\rlap{.}
  }
  \label{rules:canonical}
\end{equation}

Howe's signatures thus only need to specify evaluation of program
operations.  We now successively introduce notions of premises, rules,
and signatures.

We start by fixing a global choice of finite coproducts in both
$[2·ℕ²,𝐒𝐞𝐭]$ (families of operation arities) $[𝟚·𝔽^𝟚,𝐒𝐞𝐭]$
(presheaves).

\begin{defi}
  Given any families $K$ and $L$ of operation arities,
  a \emph{premise} $K → L$
  consists of  a \emph{source} program  $0,0 ⊢_K e : 𝐩$
  and a \emph{target} value $0,0 ⊢_L v : 𝐯$, where
  $L$ and $v$ take one of the following two forms:
  \begin{itemize}
  \item either $L = K+ ( c ∶ 𝐯)$, extending $K$ with one closed value
    metavariable,
    in which case $v=c$,
  \item or $L = K +
    (
    αᵢ ∶ 𝐯  
    )_{i ∈ \{1, …, N^+ₒ\}}
    + 
    (
    βᵢ ∶ sᵥ^{d⁻_{o,i}} → 𝐩  
    )_{i ∈ \{1, …, N^-ₒ\}}
    $
    for some value operation $o ∈ O_𝐯$, in
    which case the target is $o_𝐯(α₁,…,α_{N⁺ₒ}; β₁,…,β_{N^-_o})$.
\end{itemize}
\end{defi}
\begin{nota}
  We write any premise $(e,v)∶ K → L$ as $K \xto{e⇓v} L$.
\end{nota}
\begin{defi}
  A \emph{rule} consists of
  \begin{itemize}
  \item a \emph{head} program operation $o$,
  \item a composable sequence of premises
    $K₀ \xto{e₁ ⇓ v₁} K₁ →  … → K_{q-1}  \xto{e_q ⇓ v_q} K_q$,
    where $K₀ = (
    kᵢ ∶ sᵥ^{d^𝐩_{o,i}} → 𝐩)_{i ∈ \{ 1, …, N^𝐩ₒ\}} 
    $, and
  \item a \emph{tail} metavariable $(kᵥ ∶ 𝐯)$ in $K_q$.
  \end{itemize}
  A \emph{Howe dynamic signature} over a Howe binding signature is a
  family of rules.
\end{defi}
\begin{nota} Such a rule is denoted by
  \begin{mathpar}
    \inferrule{e₁ ⇓ v₁ \\ … \\ e_q ⇓ v_q}{o(k₁,…,k_{N^𝐩ₒ}) ⇓ kᵥ}\rlap{ .}
  \end{mathpar}
  The families of operation arities $K₀,…,K_q$ are left implicit.
\end{nota}
\begin{rem}\label{rk:discrepance-howe}
  There are some discrepancies w.r.t.\ Howe's original
  format~\cite[Lemma 6.1]{DBLP:journals/iandc/Howe96}.
  \begin{itemize}
  \item We restrict to rules with a finite number of premises\footnote{We
      expect that our setting can be extended to account for rules with an infinite
      number of premises by replacing the finitarity condition on $Σ₁^F$ with
      a weaker accessibility requirement.}.
  \item We have a slightly different treatment of value vs.\ program
    variables and metavariables.
  \end{itemize}
\end{rem}

\begin{defi}
  For any Howe binding signature $B$, and Howe dynamic signature
  $D$ over it, the \emph{evaluation}, denoted by $⇓^{B,D}$, or $⇓$ for short,
  is the binary relation between programs and values, obtained
  inductively by instantiating the given rules together with the
  canonical rules~\eqref{rules:canonical}.
\end{defi}
Let us now recall Howe's general definition of applicative
bisimulation, and its open extension.
\begin{defi}
  For any Howe binding signature $B$, and Howe dynamic signature $D$
  over it, an \emph{applicative simulation} is a binary relation $R$
  in $𝐒𝐞𝐭^𝟚$ over the injection $𝒯₀ᴮ(∅)(0,0)_𝐯 ↪ 𝒯₀ᴮ(∅)(0,0)_𝐩$, such
  that for any closed programs $e$ and $e'$ such that
  $e \mathrel{R_𝐩} e'$, and any transition
  $e ⇓ o_𝐯(v₁,…,v_{N⁺ₒ};e₁,…,e_{N⁻ₒ})$, there exist
  $v'₁,…,v'_{N⁺ₒ},e'₁,…,e'_{N⁻ₒ}$ and a transition
  $e' ⇓ o_𝐯(v'₁,…,v'_{N⁺ₒ};e'₁,…,e'_{N⁻ₒ})$, such that
  $vᵢ \mathrel{R_𝐯} v'ᵢ$ for all $i ∈ N⁺_O$, and
  $eⱼ[σ] \mathrel{R_𝐩} e'ⱼ[σ]$, for all $j ∈ N⁻ₒ$ and closing
  substitutions $σ$.

  An \emph{applicative bisimulation} is an applicative simulation
  whose converse relation also is one.

  Applicative bisimulations are closed under union, and we let
  $∼_{B,D}$, or $∼$ for short, denote the largest one.

  Finally, the open extension $R°$ of a relation $R$ is the largest
  substitution-closed relation contained in $R$ on closed programs and
  values.
\end{defi}

\begin{rem}
  By definition, being a subobject of the injection
  \[ 𝒯₀ᴮ(∅)(0,0)_𝐯 ↪ 𝒯₀ᴮ(∅)(0,0)_𝐩 \] in $𝐒𝐞𝐭^𝟚$, any applicative
  simulation is such that for any related values $v \mathrel{R}_𝐯 v'$,
  $v$ and $v'$ should also be related as programs, i.e.,
  $ι · v \mathrel{R}_𝐩 ι · v'$.
\end{rem}

\subsubsection{Congruence by encoding in our framework}
We now want to prove using Theorem~\ref{thm:main} that $∼°$ is a
congruence. For this, we could try to naively model evaluation rules
as a dynamic signature $Σ₁$ over $Σ₀ᴮ$. However, substitution-closed
bisimilarity in the obtained system would not coincide with
applicative bisimilarity in Howe's sense.  Furthermore, the induced
functor $Σ₁^\source∶ Σ₀ᴮ\Trans → \psh[ℂ_\source]$ would not be
cellular.  Indeed, consider for instance the usual, call-by-name rule
for application. In Howe's format, it gives the following:
\begin{mathpar}
  \inferrule{k₁ ⇓ λ(k₃) \\ k₃(k₂) ⇓ k₄}{k₁\ k₂ ⇓ k₄},
\end{mathpar}
where all $kᵢ$ are metavariables.
\begin{nota}
  For any $n ∈ ℕ$, we use the abbreviation
  $nₛ ≔ ((0,0) +_{sᵥ} n)ₛ$ --- because we will mainly need tuples
  $(n_𝐩,n_𝐯)$ with only $n_{sᵥ} ≠ 0$ from now on.
\end{nota}
Letting $E$ denote the exponent of the induced $Σ₁^\source$, and $o$
denote the element corresponding to this rule in the spectrum, the
boundary arity $E(s_⇓ ↾ o)∶ E(\source(⇓),o·s_⇓) → E(⇓,o)$ is the
bottom composite in the following diagram,
\begin{center}
  \Diag(1,2.2){%
    \pbk[3em]{m-3-2}{m-3-4}{m-1-4} %
    }{%
    \& \& ℒ(𝐲_{0_𝐩}) \& ℒ(𝐲_⇓) \\
    \& ℒ(𝐲_{0_𝐯}+𝐲_{0_𝐩}) \& ℒ(𝐲_{1_𝐩}+𝐲_{0_𝐩}) \& \\
    ℒ(𝐲_{0_𝐩}+𝐲_{0_𝐩}) \& ℒ(𝐲_⇓+𝐲_{0_𝐩}) \&  \& E(⇓,o) %
  }{%
    (m-3-1) edge[labelb={ℒ(𝐲_{s_⇓} + 𝐲_{0_𝐩})}] (m-3-2) %
    (m-2-2) edge[labela={λ(k₁)+\id_{ℒ(𝐲_{0_𝐩})}}] (m-2-3) %
    edge[labell={ℒ(𝐲_{t_⇓}+𝐲_{0_𝐩})}] (m-3-2) %
    (m-1-3) edge[labela={ℒ(𝐲_{s_⇓})}] (m-1-4) %
    edge[labell={\bar{χ}}] (m-2-3) %
    (m-3-2) edge[labelb={}] (m-3-4) %
    (m-1-4) edge[labelr={}] (m-3-4) %
  }
\end{center}
where
\begin{itemize}
\item $\bar{χ}$ is analogous to Definition~\ref{def:chibar};
\item we implicitly use transposition, the Yoneda lemma, and canonical
  isomorphisms $ℒ(A)+ℒ(B)≅ℒ(A+B)$.
\end{itemize}
This bottom composite is not a relative cellular complex, 
because $λ(k₁)∶ ℒ(𝐲_{0_𝐯}) → ℒ(𝐲_{1_𝐩})$ is not of the form
$ℒ(𝐲_{s_⇓})$.

So we have two problems: substitution-closed bisimilarity is not as
desired, and $Σ₁^\source$ is not cellular.  The solution to both
problems is to rectify this last point. Namely, we construct a new
signature from the evaluation rules, in such a way that we may observe
each argument of a value.  For this, we first need to generalise the
notions of premise, rule, and signature, as well as the evaluation
relation induced by a signature, which in turn requires us to
introduce the Howe context induced by a Howe dynamic signature over a
Howe binding signature.

We fix a Howe binding signature $B$ and a Howe dynamic signature $D$
over it for the rest of this section.
\begin{defi}
  We define the \emph{Howe context} induced by $(B,D)$, as follows.
  \begin{itemize}
  \item For state types, we take $ℂ₀ = \op{(𝟚×𝔽^𝟚)}$, with the
    monoidal structure on $\hat{ℂ}₀ = [𝟚×𝔽^𝟚,𝐒𝐞𝐭]$ specified in
    Remark~\ref{rk:howe-format-mon-prod}.
  \item For transition types, we take \[ℂ₁ = \{⇓\}
    ⊎ ∑_{o ∈ O_𝐯, i ∈ N⁺ₒ} \{ ⇓⁺_{o,i}\}
    ⊎ ∑_{o ∈ O_𝐯, j ∈ N⁻ₒ} \{ ⇓⁻_{o,j} \}.\]
  \item We define source and target as follows
    \begin{center}
      $\begin{array}{rcl}
         \source∶ ℂ₁ & → & ℂ₀ \\
         ⇓ & ↦ & 0_𝐩 \\
         ⇓⁺_{o,i} & ↦ & 0_𝐯  \\
         ⇓⁻_{o,j} & ↦ & 0_𝐯  \\
      \end{array}$
      \qquad 
      $\begin{array}{rcll}
         \but∶ ℂ₁ & → & ℂ₀ \\
         ⇓ & ↦ & 0_𝐯 \\
         ⇓⁺_{o,i} & ↦ & 0_𝐯 & \mbox{($o ∈ O_𝐯$, $i ∈ N⁺ₒ$)} \\
         ⇓⁻_{o,j} & ↦ & (d⁻_{o,j})_𝐩 &
                       \mbox{($o ∈ O_𝐯$, $j ∈ N⁻ₒ$)} \\
      \end{array}$
    \end{center}
    \end{itemize}
  \end{defi}

  We now introduce the new notions of premise, rule, and dynamic
  signature, which we deem \emph{rigid} to avoid confusion with the
  original.
\begin{defi}
  A \emph{rigid premise} consists of a family of operation arities $K$,
  called the \emph{source type},  
  a transition type $α ∈ ℂ₁$ and a \emph{source} term
  $0,0 ⊢_K e : \source(α)$.

  The \emph{target type} of a rigid premise $(K,α,e)$ is the family
  $L ≔ K + ( k ∶ sᵥⁿ → s )$, where $\but(α)=nₛ$, and its \emph{target} is the fresh
  metavariable $k ∶sᵥⁿ → s  ∈ L$.

  We denote such a premise by $K \xto{e \mathrel{α} k} L$.
\end{defi}

\begin{defi}
  A \emph{rigid rule} consists of
  \begin{itemize}
  \item a \emph{head} program operation $o$,
  \item a composable sequence of rigid premises
    $K₀ \xto{e₁ \mathrel{α₁} k₁} K₁ → … → K_{q-1} \xto{eₚ
      \mathrel{α_q} k_q} K_q$,
    where $K₀ = (kᵢ ∶ sᵥ^{d^𝐩_{o,i}} → 𝐩)_{i ∈ \{ 1, …, N^𝐩ₒ\}} $, and
  \item a \emph{tail} metavariable $(kᵥ ∶ 𝐯)$ in $K_q$.
  \end{itemize}
  A \emph{rigid dynamic signature} over a Howe binding signature is a
  family of rigid rules.
\end{defi}

\begin{defi}
  The labelled transition system induced by $B$ any the rigid dynamic
  signature $D'$ is defined by instantiating the given rigid rules,
  together with the canonical rules~\eqref{rules:canonical}, and the
  following new rules.
\begin{equation}
  \inferrule{ }{o_𝐯(v₁,…,v_{N⁺ₒ};v'₁,…,v'_{N⁻ₒ}) ⇓⁺_{o,i} vᵢ} \qquad
  \quad
    \inferrule{ }{o_𝐯(v₁,…,v_{N⁺ₒ};v'₁,…,v'_{N⁻ₒ}) ⇓⁻_{o,j} v'ⱼ}~·
\label{rules:new}
\end{equation}
\end{defi}
We may now define the rigid signature induced by $(B,D)$.

\begin{defi}
  Let the rigid dynamic signature $R(B,D)$ induced by $(B,D)$ be
  obtained as follows.  For each rule in $D$, $R(B,D)$ has a rule for
  the same program operation, whose premises are obtained by replacing
  each original premise
  $K \xto{e ⇓ o_𝐯(k₁,…,k_{N⁺ₒ};k'₁,…,k'_{N⁻ₒ})} L$ with any
  linearisation (in the straightforward, suitable sense) of the
  following tree.
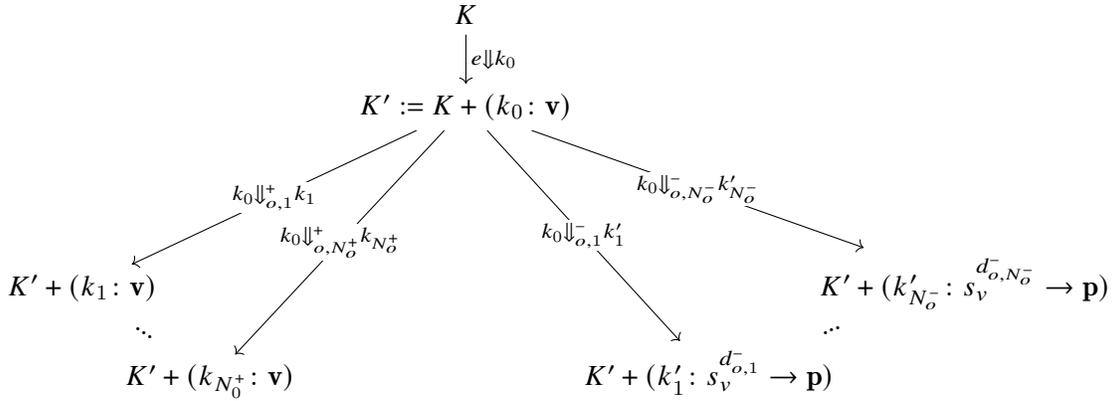
\begin{center}
  \diagramme[diagorigins={1.2}{1.7}]{}{%
    \mkdots{KHi}{KHl} %
    \mkdots{KKi}{KKl} %
  }{%
    \& \& \& |(K)| K \\
    \& \& \&  |(Ki)| K' ≔ K + ( k₀ ∶ 𝐯 ) \\
    \& \\
    |(KHi)| K' +  ( k₁ ∶ 𝐯  ) \&
    \&
    \& \&
    \&  \&
    |(KKl)| K' + (  k_{Nₒ^-}'∶ sᵥ^{d⁻_{o,Nₒ^-}} → 𝐩  ) \& \&
    \\
    \& 
    |(KHl)| K' + (  k_{N₀⁺} ∶  𝐯  ) \& \&
    \&
    |(KKi)| K' + (  k₁'∶ sᵥ^{d⁻_{o,1}} → 𝐩  ) \& \&
  }{%
    (K) edge[labelr={e ⇓ k₀}] (Ki) %
    (Ki) edge[labelon={k₀ ⇓⁺_{o,1} k₁}] (KHi) %
    edge[labelon={k₀ ⇓⁺_{o,N⁺ₒ} k_{N⁺ₒ}}] (KHl) %
    edge[labelon={k₀ ⇓⁻_{o,1} k'₁}] (KKi) %
    edge[labelon={k₀ ⇓⁻_{o,N⁻ₒ} k'_{N⁻ₒ}}] (KKl) %
  }
\end{center}
\end{defi}

Let us show that the rigid signature $R(B,D)$ is adequate.
\begin{prop}\label{prop:adequacy}
  The open extension of applicative bisimilarity in the sense of
  $(B,D)$ coincides with substitution-closed bisimilarity in the
  labelled transition system generated by $R(B,D)$.
\end{prop}
\begin{proof}
  Let $∼'$ denote substitution-closed bisimilarity in the sense of
  $R(B,D)$.

  First of all, ${∼°}$ is a substitution-closed bisimulation; indeed,
  $⇓$ is the same relation in both systems --- by a straightforward
  induction --- and by definition of applicative bisimulation, for any
  $o_𝐯(v₁,…,v_{N⁺ₒ},e₁,…,e_{N⁻ₒ}) ∼
  o_𝐯(v'₁,…,v'_{N⁺ₒ},e'₁,…,e'_{N⁻ₒ})$, we have $vᵢ ∼ v'ᵢ$ and
  $eⱼ ∼° e'ⱼ$, for all $i$ and $j$, which ensures that the new
  transitions $⇓⁺_{o,i}$ and $⇓⁻_{o,j}$ are matched.  Thus, we have
  ${∼°} ⊆ {∼'}$.

  Conversely, the new transitions ensure that ${∼'}$ is an applicative
  bisimulation, hence is contained in $∼$ on closed terms. But it is
  substitution-closed, so we get ${∼'} ⊆ {∼°}$, as desired.
\end{proof}

Finally, we show that the rigid signature $R(B,D)$ straightforwardly
gives rise to a dynamic signature $Σ₁^{B,D}$, such that the initial
vertical $Σ₁^{B,D}$-algebra is isomorphic to the generated labelled
transition system, and furthermore that the induced functor
$(Σ^{B,D}₁)^\source$ is cellular, hence that by
Proposition~\ref{prop:adequacy}, Corollary~\ref{cor:cellularenough},
and Theorem~\ref{thm:main}, $∼°$ is a congruence.

\begin{nota}
  We tend to abbreviate $Σ₁^{B,D}$ to $Σ₁$, for readability.
\end{nota}

Let us present the dynamic signature $Σ₁^{B,D}$, which will consist of
various components: one dynamic signature $Σ_{1,o}$ for each value
operation $o ∈ O_𝐯$, describing both kinds of canonical
rules~\eqref{rules:canonical} and~\eqref{rules:new},
plus one dynamic signature $Σ_{1,r}$ for
each other rule.

\begin{itemize}
\item For each value operation $o ∈ O_𝐯$, we define
  \[\begin{array}{rcl}
      Σ_{1,o}^F(X)(⇓) &=& X(⇓)^{N⁺ₒ} × ∏_{j ∈ N⁻ₒ} X((d⁻_{o,j})_𝐩) \\
      Σ_{1,o}^F(X)(⇓⁺_{o,i}) &=& X(0_𝐯)^{N⁺ₒ} × ∏_{j ∈ N⁻ₒ} X((d⁻_{o,j})_𝐩) \\
      Σ_{1,o}^F(X)(⇓⁻_{o,i}) &=& X(0_𝐯)^{N⁺ₒ} × ∏_{j ∈ N⁻ₒ} X((d⁻_{o,j})_𝐩) \\
    \end{array}\]
    with as
    $Σ_{1,o}^F(X) → Σ₀ᴮ(X)\source × X\but$:
    \begin{itemize}
    \item at $⇓$: $(r₁,…,r_{N⁺ₒ},e₁,…,e_{N⁻ₒ}) ↦
      \begin{array}[t]{l}
        (in_{o_𝐩}(r₁·s_⇓,…,r_{N⁺ₒ}·s_⇓,e₁,…,e_{N⁻ₒ}), \\
        o_𝐯(r₁·t_⇓,…,r_{N⁺ₒ}·t_⇓,e₁,…,e_{N⁻ₒ}))\rlap{,}
      \end{array}$
    \item at $⇓⁺_{o,i}$:
      $(v₁,…,v_{N⁺ₒ},e₁,…,e_{N⁻ₒ}) ↦
      (in_{o_𝐯}(v₁,…,v_{N⁺ₒ},e₁,…,e_{N⁻ₒ}), vᵢ)\rlap{,}$
    \item at $⇓⁻_{o,j}$:
      $(v₁,…,v_{N⁺ₒ},e₁,…,e_{N⁻ₒ}) ↦
      (in_{o_𝐯}(v₁,…,v_{N⁺ₒ},e₁,…,e_{N⁻ₒ}), eⱼ)\rlap{.}$
    \end{itemize}
\item For each rigid premise $p = (K \xto{e \mathrel{α} k} L)$,
  we define the cospan induced by $p$ to be
\[ℒ(K) \xto{ℓ} Eₚ \xot{r} ℒ(L)\rlap{,}\]
where the left-hand morphism $ℓ$ is defined by the pushout
\begin{equation}
  \Diag{%
    \pbk{m-2-2}{m-2-3}{m-1-3} %
  }{%
    \& ℒ(𝐲_{\source(α)}) \& ℒ(K) \\
    ℒ(𝐲_{\but(α)}) \& ℒ(𝐲_α) \& Eₚ\rlap{,}
  }{%
    (m-1-2) edge[labela={e}] (m-1-3) %
    edge[labell={ℒ (𝐲_{s_α})}] (m-2-2) %
    (m-2-2) edge[labelb={}] (m-2-3) %
    (m-1-3) edge[labelr={ℓ}] (m-2-3) %
    (m-2-1) edge[labelb={ℒ(𝐲_{t_α})}] (m-2-2) %
  }%
  \label{eq:cellpremise}
  \end{equation}
  and the right-hand morphism
  \[r∶ ℒ(L) = ℒ(K+𝐲_{\but(α)}) ≅ ℒ(K) + ℒ(𝐲_{\but(α)}) → Eₚ\]
  is obtained by copairing $ℓ$ and the bottom composite
  in~\eqref{eq:cellpremise}.
\item For each rigid rule
  \begin{mathpar}
    \inferrule{%
      (kᵢ ∶ sᵥ^{d^𝐩_{o,i}} → 𝐩)_{i ∈ \{ 1, …, N^𝐩ₒ\}} = K₀ \xto{e₁ \mathrel{α₁} k₁} K₁ → … → K_{n-1}
      \xto{eₙ \mathrel{αₙ} kₙ} Kₙ %
    }{
      o(k₁,…,k_{Nₒ^𝐩}) ⇓ k''
      }~\rlap{,}
    \end{mathpar}
    say $r$, we define
    \begin{equation}
      Σ_{1,r}^F(X)(⇓) = [E_{p₁},X] ×_{[ℒ(K₁),X]} … ×_{[ℒ(K_{p-1}),X]} [E_{pₙ},X]\rlap{,}\label{eq:Howearity}
  \end{equation}
  where for all $i ∈ n$, $ℒ(K_{i-1}) \xto{sᵢ} E_{pᵢ} \xot{tᵢ} ℒ(Kᵢ)$
  denotes the cospan induced by the $i$th premise, with
  $Σ_{1,r}^F(X)(⇓⁺_{o',i}) = Σ_{1,r}^F(X)(⇓⁻_{o',j}) = ∅$ for all
  $o' ∈ O_𝐯$, $i ∈ N⁺_{o'}$, $j ∈ N⁻_{o'}$, and again with the morphism
  $Σ_{1,r}^F(X) → Σ₀ᴮ(X)\source × X\but$ mapping any compatible tuple
  $(ρ₁,…,ρₙ) ∈ Σ_{1,r}^F(X)(⇓)$ to $(inₒ(a),b)$,
  where $a∈ ∏_{i ∈ N^𝐩ₒ} X(( d^𝐩_{o,i})_𝐩)$ corresponds by Yoneda and adjunction to the morphism
  \[
  ℒ(K₀) \xto{s₁} E₁ \xto{ρ₁} X,\] and
    $b∈ X(0ᵥ)$ corresponds to
    \[ℒ(𝐲_{0_𝐯}) \xto{ℒ(in₂)} ℒ(Kₙ) \xto{tₙ} Eₙ \xto{ρₙ} X.\]
  \item We let $(Σ₁^{B,D})^F = 
    ∑_{o ∈ O_𝐯} Σ_{1,o}^F + ∑ᵣ Σ_{1,r}^F$,
    with morphism to $Σ₀ᴮ(-)\source×(-)\but$ given by cotupling.
  \end{itemize}

  By construction, we have:
  \begin{prop}\label{prop:furtheradequacy}
    The transition $Σ₀ᴮ$-monoid generated by the dynamic signature
    $Σ₁^{B,D}$ is isomorphic to the rigid transition system $R(B,D)$.
  \end{prop}
    Furthermore, 
    we observe:
    \begin{lem}\label{lem:howecellular}
      The induced functor $(Σ₁^{B,D})^\source$ is cellular.
    \end{lem}
    \begin{proof}
      The functor is familial by construction.  By
      Corollary~\ref{cor:cellularenough}\ref{item:explicit} and
      Proposition~\ref{prop:cellscof}, it suffices to show that for
      any rule $r$ whose conclusion is a transition of type $α$, the
      boundary arity \[E(s_α ↾ r)∶ E(\source(α),r·s_α) → E(α,r)\] is a
      relative cell complex.  The arity of each canonical
      rule~\eqref{rules:canonical} is clearly a coproduct of
      generating cofibrations.  The arities of canonical
      rules~\eqref{rules:new} 
      are identities, hence trivially
      relative cell complexes. For any other rule $r$,
      by~\eqref{eq:Howearity} the arity is the left-hand leg of cospan
      obtained by composing all $K_{i-1} → Eᵢ ← Kᵢ$, which is thus by
      construction a composite of pushouts of generating cofibrations,
      hence a relative cell complex.
    \end{proof}
    We at last obtain:
    \begin{thm}
      For any Howe binding signature $B$ and dynamic signature $D$,
      the open extension of applicative bisimilarity on the generated
      transition system is a congruence.
    \end{thm}
    \begin{proof}
      The open extension of applicative bisimilarity for $(B,D)$
      coincides with substitution-closed bisimilarity for $R(B,D)$ by
      Proposition~\ref{prop:adequacy}, which further coincides with
      substitution-closed for the transition system generated by
      $Σ₁^{B,D}$ by Proposition~\ref{prop:furtheradequacy}. Finally,
      the latter is a congruence by Lemmas~\ref{lem:characfib}
      and~\ref{lem:howecellular}, and Theorem~\ref{thm:main}.
    \end{proof}

\begin{full}
\section{Congruence of substitution-closed bisimilarity}\label{s:congruence}
In this section, we elaborate on the proof sketch of
Theorem~\ref{thm:main} given in~§\ref{s:substclosedbisim}.  The
overall structure remains the same, and the final part of the proof
sketch is complete, so we mainly elaborate on
items~\ref{item:wowo}--\ref{item:sym}.

\subsection{Preliminaries on spans}
In this section, we fix a locally finitely presentable category $𝒞$,
recall some known tools about spans, and develop a few new ones,
including categorified notions of reflexivity, transitivity, symmetry,
and transitive closure.  As announced in~§\ref{ss:notation}, we freely
switch from spans $X ← S → Y$ to their pairings $S → X×Y$ in
$𝒞/X×Y$. Furthermore,
\begin{defi}
  Given any spans $R \xto{⟨p₁,p₂⟩} X×Y$ and $S \xto{⟨q₁,q₂⟩} Y×Z$,
  their \emph{span}, or \emph{sequential}, composition, denoted by
  ${R;S} → X×Z$, is the (or, rather, some global choice of) following
  span obtained by pullback.
\begin{center}
  \Diag{%
    \pbk{m-2-2}{m-1-3}{m-2-4} %
  }{%
    \& \& {R;S} \\
    \& R \& \& S \& \&  \\
    X \& \& Y \& \& Z
  }{%
    (m-1-3) edge[labelal={}] (m-2-2) %
    edge[labelar={}] (m-2-4) %
    (m-2-2) edge[labelal={p₁}] (m-3-1) %
    edge[labelar={p₂}] (m-3-3) %
    (m-2-4) edge[labelal={q₁}] (m-3-3) %
    edge[labelar={q₂}] (m-3-5) %
  }
\end{center}

\end{defi}
\begin{defi}
  A span $p∶ S → X²$ (a.k.a.\ a graph $S ⇉ X$) is \emph{reflexive} if
  there is a morphism
  \begin{center}
    \diag{%
      X \& \& S \\
      \& X² %
    }{%
      (m-1-1) edge[dashed,labela={}] (m-1-3) %
      edge[labelbl={⟨\id_X,\id_X⟩}] (m-2-2) %
      (m-1-3) edge[labelbr={p}] (m-2-2) %
    }
  \end{center}
  from the diagonal to $S$ in $𝒞/X²$. It is \emph{transitive} if there
  is a morphism $S;S → S$ in $𝒞/X²$.  Finally, it is \emph{symmetric}
  if there is a morphism $S^† → S$, where $(-)^†$ denotes the functor
  swapping projections, i.e., mapping any $R \xto{⟨p₁,p₂⟩} X²$ to
  $R \xto{⟨p₂,p₁⟩} X²$.
\end{defi}

For potentially non-reflexive spans, we will use the following
reflexive transitive closure.
\begin{defi}
  The \emph{reflexive transitive closure} $S^*$ of any span $S → X²$
  is the coproduct $∑_{n ∈ ℕ} S^{{;}n}$, or for short $∑_{n ∈ ℕ} Sⁿ$
  when the context is clear, where $S^{{;}n}$ denotes iterated span
  composition of $S$ with itself, inductively defined by
  $S^{{;}0} = X$ and $S^{{;}n+1} = S^{{;}n};S$.
\end{defi}

At some point, we will also use a more ``relational'' notion of
transitive closure, which we now introduce.
\begin{defi}
  Given a span $S → X²$, we denote by $\overline{S}$ the relation on
  $X$ induced by the image factorisation of $S → X × X$.
\end{defi}
    \begin{defi}\label{def:trans}
      The \emph{relational transitive closure} $S^{\overline{+}}$ of a span
      $S$ on $X$ is the union $⋃_{n>0}\overline{S^{{;}n}}$. 
    \end{defi}
    \begin{rem}
      Unions of relations exist by Proposition~\ref{prop:union}, as we
      have assumed the ambient category $𝒞$ to be locally finitely
      presentable.
    \end{rem}

    We immediately observe the following.
    \begin{lem}
      For all $n ∈ ℕ$, spans $S → X²$ and $R → X²$, and families
      $(Sᵢ → X²)_{i ∈ I}$ of relations spans, we have
      \begin{enumerate}
      \item \label{item:symrel} $\overline{S^†} ≅ \overline{S}^†$, 
      \item \label{item:symseq} $(R;S)^† ≅ S^†;R^†$, hence
        $(S^†)^{{;}n} ≅ (S^{{;}n})^†$, and
      \item \label{item:symU} $⋃_{i ∈ I} Sᵢ^† ≅ (⋃_{i ∈ I} Sᵢ)^†$.
      \end{enumerate}
    \end{lem}
    \begin{proof}
      Let $σ ≔ ⟨π₂,π₁⟩∶ X² → X²$.
      \begin{itemize}
      \item For~\eqref{item:symrel}, by definition each side of the
        isomorphism corresponds to one side of the exterior of the
        following commuting pentagon.
        \begin{center}
          \diag{%
            S \& \& \overline{S} \\
            \overline{S^{†}} \& \& X² \\
            \& X² %
          }{%
            (m-1-1) edge[onto,labela={}] (m-1-3) %
            edge[onto,shorten >=.5em,labell={}] (m-2-1) %
            edge[labelon={⟨π₁,π₂⟩}] (m-2-3) %
            edge[bend left=10,labelon={⟨π₂,π₁⟩}] (m-3-2) %
            (m-1-3) edge[into,labelr={}] (m-2-3) %
            (m-2-1) edge[into,labelbl={}] (m-3-2) %
            (m-2-3) edge[labelbr={σ}] (m-3-2) %
          }
        \end{center}
        But, $σ$ being an isomorphism, both sides are strong epi-mono
        factorisations of the morphism $⟨π₂,π₁⟩∶ S → X²$, hence are isomorphic.
      \item The first point of~\eqref{item:symseq} is clear.  For the
        second one, we proceed by induction.  The base case is
        trivial, and assuming $(S^†)^{{;}n} ≅ (S^{{;}n})^†$, we have
        \[(S^†)^{{;}n+1} ≅ (S^†)^{{;}n};S^† ≅ (S^{{;}n})^†;S^† ≅
        (S;S^{{;}n})^† ≅ (S^{{;}n+1})^†\rlap{,}\]
        hence the result.
      \item For~\eqref{item:symU} we have
        $⋃_{i ∈ I} Sᵢ^† = \overline{∑_{i ∈ I} Sᵢ^†}$ and
        $(⋃_{i ∈ I} Sᵢ)^† = \overline{(∑_{i ∈ I} Sᵢ)}^† ≅
        \overline{(∑_{i ∈ I} Sᵢ)^†}$ by~\eqref{item:symrel}, so it
        suffices to show $∑_{i ∈ I} Sᵢ^† ≅ (∑_{i ∈ I} Sᵢ)^†$.  But,
        letting $bᵢ∶ Sᵢ → X²$ denote each projection, the former is
        (shorthand for) the morphism $[σ∘bᵢ]ᵢ∶ ∑ᵢ Sᵢ → X²$, and the
        latter is $σ∘[bᵢ]ᵢ$, which are in fact equal.
        \qedhere
      \end{itemize}      
    \end{proof}
    \begin{cor}
      \label{cor:commute-sym-trans}
      For all spans $S → X²$, we have 
      $S^{†\overline{+}} ≅S^{\overline{+}†}$. 
    \end{cor}
    \begin{proof}
      We have $S^{†\overline{+}} = ⋃_{n>0} \overline{(S^†)^{{;}n}}
      ≅ (⋃_{n>0} \overline{S^{{;}n}})^†
      ≅ S^{\overline{+}†}$, by the lemma.
    \end{proof}
    The following lemma will later be used to exploit preservation of
    sifted colimits by $Σ₀$.
    \begin{lem}
      \label{lem:rel-trans-filtered-colim}
     If $S$ is a reflexive span on $X$, then $S^{\overline{+}}$ is the (filtered) colimit of the chain
     \[
       X → \overline{S} ≅ \overline{S;X} → \overline{S;S} ≅ \overline{S;S;X} → \overline{S;S;S} → {…}
     \]
    \end{lem}
    \begin{proof}
      Consider a colimiting cone for the given diagram, say to
      $C ∈ 𝒞$. Because the forgetful functor $𝒞/X² → 𝒞$ creates
      colimits, it is in fact a colimiting cocone in $𝒞/X²$.  Now, the
      diagram consists of monic morphisms, hence
      by~\cite[Proposition~1.62(i)]{Adamek}, so does the colimiting
      cocone. Furthermore, by~\cite[Proposition~1.62(ii)]{Adamek}, the
      mediating morphism $C → X²$ is again monic.  The cocone thus
      lifts to the category $𝐑𝐞𝐥(X)$ of (binary) relations over
      $X$. Because the forgetful functor $𝐑𝐞𝐥(X) → 𝒞/X²$ is fully
      faithful, the cocone remains colimiting in $𝐑𝐞𝐥(X)$.  Finally,
      $𝐑𝐞𝐥(X)$ is a preorder category, so the colimit of the
      considered, directed diagram is equally a colimit of the
      underlying discrete diagram, which is $S^{\overline{+}}$ by
      definition.
    \end{proof}

    The next result will be useful to show that the relational
    transitive closure of the Howe closure of substitution-closed
    bisimilarity is symmetric on states.

    \begin{lem}\label{lem:spansymtrans}
      For any reflexive span $R → X²$, $R^{\overline{+}}$ is symmetric
      if there exists a span morphism $R → R^{\overline{+}†}$.
    \end{lem}
    \begin{proof}
    Given a morphism $j∶ R → R^{\overline{+}†}$, we consider the
    composite
      \[
        R^{\overline{+}} = ⋃_{n>0}\overline{R^{;n}} →
        ⋃_{n>0}\overline{(R^{\overline{+}†})^{;n}}
        ≅
        ⋃_{n>0}\overline{((R^{\overline{+}})^{;n})^{†}}
        →
        ⋃_{n>0}\overline{R^{\overline{+}†}} ≅
        ⋃_{n>0}{R^{\overline{+}†}} ≅
        R^{\overline{+}†},
      \]
      where the first morphism is obtained from $j$, and the second one
      is obtained from morphisms $(R^{\overline{+}})^{;n} → R^{\overline{+}}$.
    \end{proof}

\subsection{Howe closure on states}
We fix an operational semantics signature $(Σ₀,Σ₁)$ on a Howe context
$\source,\but∶ ℂ₁ → ℂ₀$, and recall
\begin{itemize}
\item from Notation~\ref{not:Z0} that $𝐙₀$
denotes the initial $Σ₀$-monoid, 
\item from Theorem~\ref{thm:Z} that $𝐙$ denotes the initial vertical
  $\check{Σ}₁$-algebra (hence in particular that $𝒟(𝐙) = 𝐙₀$), and
\item from Notation~\ref{not:simtens} that $∼^⊗$ denotes
  substitution-closed bisimilarity on $𝐙$ (hence in particular that
  $𝒟(∼^⊗) = {∼^⊗₀}$ is its state part).
\end{itemize}
\begin{defi}
  Let $Σ₀ᴴ∶ \psh[ℂ₀]/𝐙₀² → \psh[ℂ₀]/𝐙₀²$ map any span $X → 𝐙₀²$ to
  \[Σ₀(X) + (X;{∼₀^{⊗}}) → Σ₀(𝐙₀)² + 𝐙₀² → 𝐙₀²\rlap{.}\]
\end{defi}

This functor $Σ₀ᴴ$ is clearly inspired from the standard Howe closure.
We now want to prove that it is pointed strong, which requires us to
equip $\psh[ℂ₀]/𝐙₀²$ with monoidal structure. But $𝐙₀²$ is a monoid, and it
is well-known~\cite[§2]{Weber:genmorph} that any slice of a monoidal
category over any monoid $M$ is again monoidal.  The tensor of $X → M$
and $Y → M$ is simply $X⊗Y → M⊗M → M$, and the unit is $I → M$.  We
may thus state the following result.
\begin{prop}\label{prop:SoHptstr}
The functor $Σ₀ᴴ$ is pointed strong.  
\end{prop}

For proving this, we first need the following.
\begin{lem}\label{lem:distrib}
  For any monoid $X$ in any monoidal category $𝒞$ with finite limits, there is a natural transformation with components
  $δ_{U,V,W}∶ (U;V)⊗W → U⊗W ; V⊗X$ in $𝒞/X²$.
\end{lem}
\begin{proof}
  Let $m∶ X⊗X → X$ denote multiplication.
  By tensoring the defining pullback of $U;V$ with $W$ we obtain the
  back face below.
  \begin{center}
  \diagramme[diagorigins={1,2}]{}{%
    \pbk{m-4-2}{m-2-2}{m-2-4} %
    \path[draw,->] (m-1-1) edge[fore,dashed,labelar={δ_{U,V,W}}] (m-2-2) ; %
    }{%
      (U;V)⊗W \& \& V⊗W \\
      \& (U⊗W);(V⊗X) \& \& V ⊗ X \\
      U⊗W \& \& X⊗W \\
      \& U⊗W \& \& X
    }{%
      (m-1-1) edge[labela={}] (m-1-3) %
      (m-1-3) edge[labelar={V⊗π₂}] (m-2-4) %
      (m-3-1) edge[labelbat={π₂⊗W}{.85}] (m-3-3) %
      edge[identity,labell={}] (m-4-2) %
      (m-4-2) edge[labelb={m∘(π₂⊗π₂)}] (m-4-4) %
      (m-3-3) edge[labelon={m∘(X⊗π₂)}] (m-4-4) %
      (m-1-1) edge[labell={}] (m-3-1) %
      (m-1-3) edge[labelrat={π₁⊗W}{.85}] (m-3-3) %
      (m-2-2) edge[fore,labell={}] (m-4-2) %
      edge[fore,labelb={}] (m-2-4) %
      (m-2-4) edge[labelr={m∘(π₁⊗X)}] (m-4-4) %
    }%
  \end{center}
    By universal property of pullback, we then get the dashed arrow
    making all faces commute, which gives our candidate $δ_{U,V,W}$.
    Naturality follows by universal property of pullback.
\end{proof}

\begin{proof}[Proof of Proposition~\ref{prop:SoHptstr}]
  Because the tensor preserves all colimits on the left, pointed
  strong endofunctors are closed under coproducts, so it suffices to
  show that both terms of the sum are pointed strong.  The first one
  inherits the pointed strength of $Σ₀$, while the pointed strength of
  ${-};{∼₀^{⊗}}$ follows from Lemma~\ref{lem:distrib} and
  substitution-closedness of ${∼₀^{⊗}}$, like so:
  $(U;{∼₀^{⊗}})⊗V → (U⊗V);({∼₀^{⊗}}⊗𝐙₀) → (U⊗V);{∼₀^{⊗}}\rlap{.}$
\end{proof}

Presheaf categories being well-known to be closed under the slice
construction, we have the following.
\begin{lem}\label{lem:CZpresheaf}
  The category $\psh[ℂ₀]/𝐙₀²$ is a presheaf category.
\end{lem}
This allows us to deduce the following.
\begin{prop}\label{prop:SoHfinitary}
  The endofunctor $Σ₀ᴴ$ is finitary.
\end{prop}
\begin{proof}
  By commutation of filtered colimits with finite limits in presheaf
  categories.
\end{proof}

By Proposition~\ref{prop:SoHptstr}, Lemma~\ref{lem:CZpresheaf}, and
Proposition~\ref{prop:SoHfinitary}, the following is legitimate.
\begin{defi}
  Let ${H₀} = 𝐙_{Σ₀ᴴ}$ denote the initial $Σ₀ᴴ$-monoid.  We denote by
  $π₁,π₂∶ H₀ → 𝐙₀$ the left and right projections.
\end{defi}
By Proposition~\ref{prop:ptstrgeneral}, we also get the following for free.
\begin{prop}
\label{prop:howe-clos-state-ini-algebra}
  The object ${H₀} → 𝐙₀²$ is an initial algebra for the
  endofunctor $\psh[ℂ₀]/𝐙₀² → \psh[ℂ₀]/𝐙₀²$ mapping any $X → 𝐙₀²$ to
  $I + Σ₀ᴴ(X) → 𝐙₀²$.
\end{prop}

\begin{prop}\label{prop:howomonoid}
  The underlying object ${H₀}$ is a $Σ₀$-monoid.
\end{prop}
\begin{proof}
  Directly follows from the $Σ₀ᴴ$-monoid structure.
\end{proof}

Next, we exhibit an alternative characterisation of ${H₀}$, which relies
on the following result.
\begin{lem}[Packing lemma]\label{lem:muFplusG}
  Consider finitary endofunctors $F$ and $G$ on a cocomplete category
  $𝐂$, and let $G^⋆(A) ≅ μS.(A + G(S))$ denote the `free $G$-algebra'
  monad~\cite[Theorem~2.1]{Reiterman}.  Then we have
  $μS.(F(S)+G(S)) ≅ μS.G^⋆(F(S))$.
\end{lem}
\begin{proof}
  Indeed, we have
  \begin{center}
    $
    \begin{array}[b]{rcl}
      μS.G^⋆(F(S)) & ≅ &  μS.μU.(F(S) + G(U)) \\
                   & ≅ &  μS.(F(S) + G(S))\rlap{,}
    \end{array}$
  \end{center}
  by the Diagonal rule~\cite[Theorem~16]{DBLP:conf/ctcs/BackhouseBGW95}.
\end{proof}

\begin{prop}\label{prop:initalt}
  The object ${H₀} → 𝐙₀²$ is an initial algebra for the
  endofunctor \[ Σ₀ᴴ{}'∶ \psh[ℂ₀]/𝐙₀² → \psh[ℂ₀]/𝐙₀² \] mapping any $X → 𝐙₀²$ to
  $I;{∼₀^{⊗*}} + Σ₀(X);{∼₀^{⊗*}} → 𝐙₀²$.
\end{prop}
\begin{proof}
  Let $F(S) = I + Σ₀(S)$ and $G(S) = S;{∼₀^⊗}$.  We observe that $G$
  preserves coproducts (because pullback along the first projection
  ${∼₀^⊗} → 𝐙₀$, as a left adjoint, preserves colimits), so that
  $G^⋆(U) ≅ ∑ₙ Gⁿ(U)$.  By commutation of coproducts with $U;{-}$, we
  thus have
  \[G^⋆(U) ≅ ∑ₙ Gⁿ(U) ≅ ∑ₙ U;(∼₀^⊗)ⁿ ≅ U; ∑ₙ(∼₀^⊗)ⁿ = U; {∼₀^{⊗*}}\rlap{.}\]
  Thus, by Lemma~\ref{lem:muFplusG}, we get
  $G^⋆(F(S)) ≅ (I + Σ₀(S));{∼₀^{⊗*}} ≅ (I;{∼₀^{⊗*}}) +
  (Σ₀(S);{∼₀^{⊗*}})$, as desired.
\end{proof}

\subsection{Double categorical notation}\label{ss:doublecat}
Our next goal is to define the Howe closure on transitions.  For this,
we appeal to Morton's \emph{double bicategories}~\cite{Morton}.  They
are a refinement of double categories, in which both the horizontal
and vertical categories are bicategories. We rely in particular on his
Theorem~4.1.3, which (when dualised) states that for any category $𝒞$
with pullbacks, there is a double bicategory $\Bispan(𝒞)$:
\begin{itemize}
\item objects are objects of $𝒞$,
\item both the vertical and horizontal bicategories are $\Span(𝒞)$,
\item cells, called \emph{double spans}, are precisely commuting
  diagrams of the following form.
\begin{equation}
  \diag(.5,.5){%
    A \& B \& C \\
    A' \& B' \& C' \\
    A'' \& B'' \& C'' %
  }{%
    (m-1-1) edge[<-,labela={}] (m-1-2) %
    edge[<-,labell={}] (m-2-1) %
    (m-2-1) edge[<-,labelb={}] (m-2-2) %
    (m-1-2) edge[<-,labelr={}] (m-2-2) %
    (m-1-2) edge[labela={}] (m-1-3) %
    (m-2-2) edge[labelb={}] (m-2-3) %
    (m-1-3) edge[<-,labelr={}] (m-2-3) %
    (m-2-1) edge[labell={}] (m-3-1) %
    (m-3-1) edge[<-,labelb={}] (m-3-2) %
    (m-2-2) edge[labelr={}] (m-3-2) %
    (m-3-2) edge[labelb={}] (m-3-3) %
    (m-2-3) edge[labelr={}] (m-3-3) %
  }
  \label{eq:cell0}
\end{equation}
\end{itemize}
We will not need the rest of the structure. All we need to know is
that cells compose horizontally and vertically just as in a weak
double category.  We will use the double bicategories $\Bispan(\psh[ℂ₀])$ and
$\Bispan(\psh[ℂ₁])$.

\begin{nota}\label{not:pasting}
  We use the following notational conventions.
  \begin{itemize}
  \item We denote cells $\Bispan(\psh[ℂ₀])$ such as~\eqref{eq:cell0} above by
\begin{center}
  \diag{%
    A \& C \\
    A'' \& C''\rlap{.}  }{%
    (m-1-1) edge[glob,labela={B}] (m-1-2) %
    edge[glob,twol={A'}] (m-2-1) %
    (m-2-1) edge[glob,labelb={B''}] (m-2-2) %
    (m-1-2) edge[glob,twor={C'}] (m-2-2) %
    (l) edge[draw=none,labelon={B'}] (r) }%
  \end{center}

\item Furthermore, cells in $\psh[ℂ₁]$ of the form below left will
  be denoted as below right.
  \begin{equation}
    \diag(.5,.5){%
      X₀\source \& S₀\source \& Y₀\source  \\
      X₁ \& S₁ \& Y₁ \\
      X'₀\but \& S'₀\but \& Y'₀\but %
    }{%
      (m-1-1) edge[<-,labela={}] (m-1-2) %
      edge[<-,labell={}] (m-2-1) %
      (m-2-1) edge[<-,labelb={}] (m-2-2) %
      (m-1-2) edge[<-,labelr={}] (m-2-2) %
      (m-1-2) edge[labela={}] (m-1-3) %
      (m-2-2) edge[labelb={}] (m-2-3) %
      (m-1-3) edge[<-,labelr={}] (m-2-3) %
      (m-2-1) edge[labell={}] (m-3-1) %
      (m-3-1) edge[<-,labelb={}] (m-3-2) %
      (m-2-2) edge[labelr={}] (m-3-2) %
      (m-3-2) edge[labelb={}] (m-3-3) %
      (m-2-3) edge[labelr={}] (m-3-3) %
    }
    \qquad \qquad  \diag{%
      X₀ \& Y₀  \\
      X'₀ \& Y'₀
    }{%
      (m-1-1) edge[glob,labela={S₀}] (m-1-2) %
      edge[pro,twol={X₁}] (m-2-1) %
      (m-2-1) edge[glob,labelb={S'₀}] (m-2-2) %
      (m-1-2) edge[pro,twor={Y₁}] (m-2-2) %
      (l) edge[draw=none,labelon={S₁}] (r)
    }%
    \label{eq:doublecell}
  \end{equation}
  Explicitly, spans of the form $X₀\source ← X₁ → X'₀\but$ are denoted by
  $X₀ \xarrow[pro,->]{X₁} X'₀$, while spans of the form $X₀ ← S₀ → Y₀$
  are still denoted by $X₀ \xarrow[glob,->]{S₀} Y₀$, but silently
  coerced by $Δ_{\source}$ or $Δ_{\but}$ depending on context.

\item For both types of cells, we collapse identity borders, as usual.
  
\item When a span is trivial on one side, we use standard arrows for
  its borders, and a double arrow for its middle arrow, all in the
  relevant direction.  E.g., the diagram below left may be depicted as
  below right.
\begin{center}
  \diag|baseline=(m-1-1.base)|(.5,.5){%
    A \& B \& C \\
    A' \& B' \& C' \\
  }{%
    (m-1-1) edge[<-,labela={}] (m-1-2) %
    edge[<-,labell={a}] (m-2-1) %
    (m-2-1) edge[<-,labelb={}] (m-2-2) %
    (m-1-2) edge[<-,labelr={b}] (m-2-2) %
    (m-1-2) edge[labela={}] (m-1-3) %
    (m-2-2) edge[labelb={}] (m-2-3) %
    (m-1-3) edge[<-,labelr={c}] (m-2-3) %
  }
  \qquad
  \diag|baseline=(m-1-1.base)|(.5,.5){%
    A \& C \\
    A' \& C' %
  }{%
    (m-1-1) edge[glob,twoa={B}] (m-1-2) %
    edge[<-,twol={a}] (m-2-1) %
    (m-2-1) edge[glob,twob={B'}] (m-2-2) %
    (m-1-2) edge[<-,twor={c}] (m-2-2) %
    (b) edge[cell=.1,labell={b}] (a) %
  }%
\end{center}
\end{itemize}
\end{nota}

Cells of the form~\eqref{eq:doublecell} live in $\Bispan(\psh[ℂ₁])$, hence
may be composed horizontally.  Relevant examples of vertical
composition will be obtained by embedding cells of the
form~\eqref{eq:cell0} along $Δ_{\source}$ (resp.\ $Δ_{\but}$), and
vertically composing with cells of the form~\eqref{eq:doublecell} in
$\Bispan(\psh[ℂ₁])$. This yields a top (resp.\ bottom) action of
$\Bispan(\psh[ℂ₀])$, which we both denote by mere pasting.

  \begin{lem}
    Given a composable pasting diagram made of cells of both types,
    any two parsings agree up to isomorphism.
  \end{lem}
  \begin{proof}
    By interchange of limits.
  \end{proof}

  Let us end this subsection by generalising simulations to
  $\Bispan(\psh[ℂ₁])$.  By Proposition~\ref{prop:sim:wpbk}, the span
  denoted by a cell~(\ref{eq:doublecell}, right) in $\Bispan(\psh[ℂ₁])$ is
  a simulation iff the top left square in the corresponding
  diagram~(\ref{eq:doublecell}, left) is a pointwise weak pullback.
  Abstracting over this:
  \begin{defi}
    A cell in any double bicategory of the form $\Bispan(\psh)$ is a
    \emph{simulation} iff its top left square is a pointwise weak
    pullback.
  \end{defi}
  \begin{prop}\label{prop:simpbk}
    Simulations are closed under horizontal and vertical composition
    in any double bicategory of the form $\Bispan(\psh[ℂ])$.
  \end{prop}
  In order to prove this, we need the following weak analogues of the
  pullback lemma.
  \begin{nota}
    We denote
      weak pullbacks (in any category) by dashed corners.
  \end{nota}
  \begin{lem}\label{lem:wpullback}
    In any category (resp.\ presheaf category),
    \begin{enumerati}
    \item
      for any commuting diagram
      \begin{center}
        \Diag{%
          \wpbk{m-2-1}{m-1-1}{m-1-2} %
          \wpbk{m-2-2}{m-1-2}{m-1-3} %
        }{%
          A \& B \& C \\
          D \& E \& F\rlap{,} %
        }{%
          (m-1-1) edge[labela={}] (m-1-2) %
          edge[labell={}] (m-2-1) %
          (m-2-1) edge[labelb={}] (m-2-2) %
          (m-1-2) edge[labelr={}] (m-2-2) %
          (m-1-2) edge[labela={}] (m-1-3) %
          (m-2-2) edge[labelb={}] (m-2-3) %
          (m-1-3) edge[labelr={}] (m-2-3) %
        }
      \end{center}
      if both squares are weak pullbacks (resp.\ pointwise weak
      pullbacks), then so is the outer rectangle; and
    \item for any commuting diagram
      \begin{center}
        \Diag{%
          \pullback{m-2-1}{m-1-1}{m-1-2}{draw,dashed,shorten >=-.8em} %
          \pbk{m-2-2}{m-1-2}{m-1-3} %
        }{%
          A \& B \& C \\
          D \& E \& F\rlap{,} %
        }{%
          (m-1-1) edge[labela={}] (m-1-2) %
          edge[labell={}] (m-2-1) %
          edge[bend left=30,labela={}] (m-1-3) %
          (m-2-1) edge[labelb={}] (m-2-2) %
          (m-1-2) edge[labelr={}] (m-2-2) %
          (m-1-2) edge[labela={}] (m-1-3) %
          (m-2-2) edge[labelb={}] (m-2-3) %
          (m-1-3) edge[labelr={}] (m-2-3) %
        }
      \end{center}
      if the right-hand square is a pullback and the outer rectangle
      is a weak pullback (resp.\ pointwise weak pullback), then the
      left-hand square is a weak pullback (resp.\ pointwise weak
      pullback).
    \end{enumerati}
  \end{lem}
  \begin{proof}
    Similar to the proof of the standard pullback lemma.
  \end{proof}
  
  \begin{proof}[Proof of Proposition~\ref{prop:simpbk}]
    Straightforward, using Lemma~\ref{lem:wpullback}.
  \end{proof}
  
  For vertical composition with cells from $\Bispan(\psh[ℂ₀])$, as in
  Notation~\ref{not:pasting}, we will also need the following.  
  \begin{prop}\label{prop:embedcells}
    Precomposition with $\source$ and $\but$ yields maps
    $\Bispan(\psh[ℂ₀]) → \Bispan(\psh[ℂ₁])$ between cell sets, which preserve
    borders and simulations.
  \end{prop}
  \begin{proof}
    Both precomposition functors straightforwardly preserve pointwise
    weak pullbacks. 
  \end{proof}

\subsection{Howe closure on transitions}
Let us now define the Howe closure on transitions.  First, we
delineate an ambient category $𝐂ᴴ_𝐙$.  The idea is that objects of
this category should be transition systems $S → 𝐙²$ over $𝐙²$ whose
image under the projection $𝐂/𝐙² → \psh[ℂ₀]/𝐙₀²$ is precisely $H₀ → 𝐙₀²$.
Thus, an object of $𝐂ᴴ_𝐙$ consists of an object $S₁ ∈ \psh[ℂ₁]$, equipped with
a dashed cone to the outer part of the diagram below.
\begin{equation}
  \diag(.5,.5){%
    𝐙₀\source \& H₀\source \& 𝐙₀\source  \\
    𝐙₁ \& S₁ \& 𝐙₁ \\
    𝐙₀\but \& H₀\but \& 𝐙₀\but %
  }{%
    (m-1-1) edge[<-,labela={}] (m-1-2) %
    edge[<-,labell={}] (m-2-1) %
    (m-2-1) edge[<-,dashed,labelb={}] (m-2-2) %
    (m-1-2) edge[<-,dashed,labelr={}] (m-2-2) %
    (m-1-2) edge[labela={}] (m-1-3) %
    (m-2-2) edge[dashed,labelb={}] (m-2-3) %
    (m-1-3) edge[<-,labelr={}] (m-2-3) %
    (m-2-1) edge[labell={}] (m-3-1) %
    (m-3-1) edge[<-,labelb={}] (m-3-2) %
    (m-2-2) edge[dashed,labelr={}] (m-3-2) %
    (m-3-2) edge[labelb={}] (m-3-3) %
    (m-2-3) edge[labelr={}] (m-3-3) %
  }
\label{eq:doublebicat}
\end{equation}
Equivalently, they are morphisms over the limit, so that we may define
$𝐂ᴴ_𝐙$ as a slice category by merely stating the
following.
\begin{defi}
  Let $R^{∂H}$ denote  the limit of the outer part of~\eqref{eq:doublebicat}.
\end{defi}
\begin{defi}
  $𝐂ᴴ_𝐙$ is the category of cones over the
  outer part of~\eqref{eq:doublebicat}, or equivalently,
  it is the slice category $\psh[ℂ₁]/R^{∂H}$.

  Furthermore, we denote by $𝒰ᴴ_𝐙∶ 𝐂ᴴ_𝐙 → 𝐂/𝐙²$ the forgetful
  functor.
\end{defi}
\begin{prop}\label{prop:initialwowoZ}
  The initial object in $𝐂ᴴ_𝐙$ is the span $𝐙 ← {H₀} → 𝐙$, i.e., the
  one with $S₁ = 0$.
\end{prop}

\begin{defi}
  Let $Σ₁ᴴ∶ 𝐂ᴴ_𝐙 → 𝐂ᴴ_𝐙$ map any object $𝐙 ← S → 𝐙$ to the coproduct of
  the following two pastings.
      \begin{equation}
\Diag(.6,1){%
  \draw[->,pro,labell={𝐙₁},rounded corners]
  (m-1-1) -- ($(m-1-1.west) + (180:5ex)$) --  node(ll) {} ($(m-3-1.west) + (180:5ex)$) -- (m-3-1) %
  ; %
  \path[->,draw] %
    (m-2-1) edge[celllr={0}{0}] (ll) %
    ; %
  \draw[->,pro,labelr={𝐙₁},rounded corners]
  (m-1-2) -- ($(m-1-2.east) + (0:5ex)$) --  node(rr) {} ($(m-3-2.east) + (0:5ex)$) -- (m-3-2) %
  ; %
  \path[->,draw] %
    (m-2-2) edge[celllr={0}{0}] (rr) %
    ; %
}{%
  𝐙₀ \& 𝐙₀ \\
  Σ₀(𝐙₀) \&  Σ₀(𝐙₀) \\
  𝐙₀ \& 𝐙₀  %
}{%
  (m-1-1) edge[glob,twoa={H₀}] (m-1-2) %
  edge[<-] node (l) {} (m-2-1) %
  (m-2-1) edge[glob,twob={},labela={Σ₀(H₀)}] (m-2-2) %
  (m-1-2) edge[<-,labelr={}]  node (r) {} (m-2-2) %
  (m-2-1) edge[pro,labell={Σ₁^F(𝐙)}] node (dl) {} (m-3-1) %
  (m-3-1) edge[glob,labelb={H₀}] (m-3-2) %
  (m-2-2) edge[pro,labelr={Σ₁^F(𝐙)}] node (dr) {} (m-3-2) %
  (dl) edge[draw=none,labelon={Σ₁^F(S)}] (dr) %
  (b) edge[celllr={1.5}{0}] (a)
}%
      \qquad
      \Diag(.6,1){%
        \draw[->,glob,labela={H₀},rounded corners]
        (m-1-1) -- ($(m-1-1.north) + (90:5ex)$) --  node(a) {} ($(m-1-3.north) + (90:5ex)$) -- (m-1-3) %
        ; %
        \path[->,draw] %
        (m-1-2) edge[celllr={1}{1}] (a) %
        ; %
        \draw[->,glob,labelb={H₀},rounded corners]
        (m-2-1) -- ($(m-2-1.south) + (-90:5ex)$) --  node(b) {} ($(m-2-3.south) + (-90:5ex)$) -- (m-2-3) %
        ; %
        \path[->,draw] %
        (m-2-2) edge[celllr={1}{1}] (b) %
    ; %
      }{%
        𝐙₀ \& 𝐙₀ \& 𝐙₀ \\
        𝐙₀ \& 𝐙₀ \& 𝐙₀ %
      }{%
        (m-1-1) edge[glob,labela={H₀}] (m-1-2) %
        edge[pro,labell={𝐙₁}] node (ll) {} (m-2-1) %
        (m-2-1) edge[glob,labelb={H₀}] (m-2-2) %
        (m-1-2) edge[pro,labelr={𝐙₁}] node (mm) {} (m-2-2) %
        (m-1-2) edge[glob,labela={{∼₀^{⊗}}}] (m-1-3) %
        (m-2-2) edge[glob,labelb={{∼₀^{⊗}}}] (m-2-3) %
        (m-1-3) edge[pro,labelr={𝐙₁}] node (rr) {} (m-2-3) %
        (ll) edge[draw=none,labelon={S₁}] (mm)
        (mm) edge[draw=none,labelon={{∼₁^{⊗}}}] (rr)
      }%
      \label{eq:SigmabulletRinit}
    \end{equation}
\end{defi}

\begin{prop}\label{prop:Sigbulletfinitary}
  The functor $Σ₁ᴴ∶ 𝐂ᴴ_𝐙 → 𝐂ᴴ_𝐙$ is finitary.
\end{prop}
    \begin{proof}
      The forgetful functor $𝐂ᴴ_𝐙≅ \psh[ℂ₁]/R^{∂H} → \psh[ℂ₁]$ creates colimits,
      so it suffices to show that the composite
      $𝐂ᴴ_𝐙 \xto{Σ₁ᴴ} 𝐂ᴴ_𝐙 → \psh[ℂ₁]$ is finitary.  This functor maps any
      $S$ to $Σ₁^F(S) + {S₁;{∼₁}^{⊗}}$, hence is finitary because
      $Σ₁^F$ is and ${-};∼₁^{⊗}$ is cocontinuous.
      \end{proof}

The last result legitimates the following definition.
\begin{defi}
  Let $H_𝐙$ denote the initial $Σ₁ᴴ$-algebra.  We call
  $H ≔ 𝒰ᴴ_𝐙(H_𝐙) ∈ 𝐂/𝐙²$ the \emph{Howe closure} of
  substitution-closed bisimilarity.
\end{defi}

We readily can prove the following.  \simintowow* 
\begin{proof}
  By construction, the underlying object of $H$ is in particular a
  $\check{Σ}₁$-algebra, so by initiality we obtain a unique span
  morphism $𝐙 → {H}$ --- in other words $H$ is reflexive.
  Furthermore, again by construction, ${H}$ is an algebra for the
  endofunctor ${-};{∼^{⊗}}$ on $𝐂/𝐙²$.  We thus may form the composite
  ${∼^{⊗}} ≅ 𝐙;{∼^{⊗}} → {H};{∼^{⊗}} → {H}$.
\end{proof}

\subsection{Alternative characterisations of the Howe closure}
In this section, we exhibit a few alternative characterisations of the
Howe closure on transitions.  The definition in the previous section
is convenient for proving that the transitive closure is symmetric,
while our final alternative characterisation will enable a conceptual
proof of the simulation property.

First of all, we have:
\begin{lem}
  \label{lem:pack-Sigma1}
  The object ${H_𝐙} ∈ 𝐂ᴴ_𝐙$ is (isomorphic to) the initial algebra of
  the endofunctor $Σ_{1,\pack}ᴴ∶ 𝐂ᴴ_𝐙 → 𝐂ᴴ_𝐙$ mapping any $𝐙 ← S → 𝐙$
  to the following pasting.
      \begin{equation}
      \Diag{%
        \sqpath{m-1-1}{180}{7}{l}{m-3-1}{pro,labell={𝐙₁}}
        \sqpath{m-3-1}{-90}{5}{b}{m-3-4}{glob,labelb={H₀}} %
        \sqpath{m-1-1}{90}{5}{aa}{m-1-4}{glob,labela={H₀}} %
        \path (m-3-1) -- node (bm) {} (m-3-4) ; %
        \path[draw] %
        (m-2-1) edge[celllr={0}{0}] (l) %
        (m-2-2) edge[celllr={1}{1},labela={}] (ddrrr) %
        ($(aa)+(-90:4ex)$) edge[celllr={0}{0}] (aa) %
        ($(a)+(-90:3ex)$) edge[celllr={0}{0}] (a) %
        ($(b)+(90:4ex)$) edge[celllr={0}{0}] ($(b)+(90:1ex)$) %
        ; %
      }{%
        𝐙₀ \& 𝐙₀ \& 𝐙₀ \& 𝐙₀ \\
        Σ₀(𝐙₀) \& Σ₀(𝐙₀) \&  \&  \\
        𝐙₀ \& 𝐙₀ \& 𝐙₀ \& 𝐙₀ %
      }{%
        (m-1-1) edge[glob,twoa={H₀}] (m-1-2) %
        (m-1-2) edge[identity] (m-1-3) %
        (m-1-3) edge[glob,labela={{∼₀^{⊗*}}}] (m-1-4)
        edge[pro,labelr={𝐙₁}] node (ddrrr) {} (m-3-3) %
        (m-2-1) edge[labell={}] (m-1-1) %
        (m-2-1) edge[glob,labela={Σ₀(H₀)}] (m-2-2) %
        edge[pro,labell={Σ₁^F(𝐙)}] node (ddr) {} (m-3-1) %
        (m-3-1) edge[glob,labelb={H₀}] (m-3-2) %
        (m-2-2) edge[pro,labellat={\scriptscriptstyle Σ₁^F(𝐙)}{.1}] node (ddrr) {} (m-3-2) %
        edge[labelr={}] (m-1-2) %
        (m-3-2) edge[identity,labelb={}] (m-3-3) %
        (m-3-3) edge[glob,labelb={{∼₀^{⊗*}}}] (m-3-4) %
        (m-1-4) edge[pro,labelr={𝐙₁}] node (ddrrrr) {} (m-3-4) %
        (ddr) edge[draw=none,labelon={Σ₁^F(S)}] (ddrr) %
        (ddrrr) edge[draw=none,labelon={{∼₁^{⊗*}}}] (ddrrrr) %
      }
      \label{eq:SigmabulletRpack}
    \end{equation}  
\end{lem}
For the proof, we need an intermediate result,
Corollary~\ref{cor:commute-free-mon-adj} below, which relies on the
following lemma.
\begin{lem}\label{lem:commute-free-monads}
  Consider any diagram
  \begin{center}
    \diag{%
      𝒜 \& 𝒜 \\
      ℬ \& ℬ \\
    }{%
      (m-1-1) edge[labela={F}] (m-1-2) %
      (m-2-1) edge[twol={K}] (m-1-1) %
      (m-2-1) edge[labelb={G}] (m-2-2) %
      (m-2-2) edge[twor={K}] (m-1-2) %
      (l) edge[draw=none,labelon={⇓α}] (r)
    }
  \end{center}
  of functors and natural transformations, such that $F$ and $G$ are
  finitary, and $𝒜$ and $ℬ$ are cocomplete.  Let $M$ be the
  induced endofunctor on the comma category $𝒜/K$, mapping
  $a \xrightarrow{f} Kb$ to
  $Fa \xrightarrow{Ff} FKb \xrightarrow{α_b} KGb$.  Then, given an
  object $f:a → K(b)$ of $𝒜/K$, there is a unique morphism
  $f^*:F^⋆ a → KG^⋆ b$ such that the following diagram commutes,
  \begin{center}
    \diag{%
      F (F^*(a)) \& \& F^*(a) \& a \\
      F (K (G^*(b)) \& K (G (G^* (b))) \& K (G^* (b)) \& K(b) %
    }{%
      (m-1-1) edge[labela={ν^Fₐ}] (m-1-3) %
      edge[labell={F(f^*)}] (m-2-1) %
      (m-2-1) edge[labelb={α_{G^*(b)}}] (m-2-2) %
      (m-2-2) edge[labelb={K (νᴳ_b)}] (m-2-3) %
      (m-1-3) edge[labelr={f^*}] (m-2-3) %
      (m-1-4) edge[labela={η^Fₐ}] (m-1-3) %
      edge[labelr={f}] (m-2-4) %
      (m-2-4) edge[labelb={K(ηᴳ_b)}] (m-2-3) %
    }
  \end{center}
  where $νₐ∶ F(F^*(a)) → F^*(a)$ denotes the canonical $F$-algebra
  structure on $F^*(a)$, and similarly for $G$.
  Furthermore, the right-hand square above, viewed as a morphism
  $f → f^*$ exhibits $f^*$ as a free $M$-algebra on $f$.
  \end{lem}
  \begin{proof}
    Let us first observe that $K$ lifts to a functor
    $\overline{K}:G\alg → F\alg$, which maps $Gx → x$ to
    $FKx \xrightarrow{αₓ}KGx → Kx$.  This in particular equips
    $K(G^*(b))$ with $F$-algebra structure.  Let us thus define $f^*$,
    by universal property of $F^*(a)$, to be the unique $F$-algebra
    morphism $F^*(a) → K (G^* (b))$ whose restriction to $a$ is
    $K(ηᴳ_b)∘f$.  This ensures in particular that the required diagram
    commutes.

    Let us finally show that $f^*$ is an initial $M$-algebra.
    For this, we observe that $M$-algebra structure on $f∶ a → K(b)$
    means morphisms $u$ and $v$ making the following diagram commute.
    \begin{center}
      \diag{%
        F(a) \& \& a \\
        F (K (b)) \& K (G (b)) \& K(b) %
      }{%
        (m-1-1) edge[labela={u}] (m-1-3) %
        edge[labell={F(f)}] (m-2-1) %
        (m-2-1) edge[labelb={α_b}] (m-2-2) %
        (m-2-2) edge[labelb={K(v)}] (m-2-3) %
        (m-1-3) edge[labelr={f}] (m-2-3) %
      }
    \end{center}
    Thus, $M$-algebra structure $(u,v)$ on $f$ is exactly the same as
    $F$-algebra structure $u$ on $a$, $G$-algebra structure $v$ on
    $b$, and an $F$-algebra morphism $a → \overline{K}(b)$.  This
    shows that $M\alg$ is isomorphic to the comma category
    $F\alg/\overline{K}$.  But for any morphism, say $(h,k)$ from $f$
    to any $f'∶ a' → \overline{K}(b')$ in $F\alg/\overline{K}$, 
    by universal property of $F^*(a)$ and $G^*(b)$, we get
    maps $\tilde{h}$ and $\tilde{k}$, respectively in $F\alg$ and
    $G\alg$, making both triangles commute in the following diagram.
    \begin{center}
      \diag{%
        F^*(a) \& \& a \\
        \& a' \\
        K(G^*(b)) \& \& K(b) \\
        \& K(b') %
      }{%
        (m-1-3) edge[labela={η^Fₐ}] (m-1-1) %
        edge[labelr={f}] (m-3-3) %
        edge[labelbr={h}] (m-2-2) %
        (m-3-3) edge[labelaat={K(η^G_b)}{.2}] (m-3-1) %
        edge[labelbr={k}] (m-4-2) %
        (m-2-2) edge[fore,labelrat={f'}{.2}] (m-4-2) %
        (m-1-1) edge[labell={f^*}] (m-3-1) %
        edge[dashed,labelbl={\tilde{h}}] (m-2-2) %
        (m-3-1) edge[dashed,labelbl={K(\tilde{k})}] (m-4-2) %
      }
    \end{center}
    By functoriality of $\overline{K}$ and uniqueness in the universal
    property of $F^*(a)$, the left-hand square also commutes, so
    $(\tilde{h},\tilde{k})$ is a morphism in $F\alg/\overline{K}$, as
    desired. Finally, uniqueness follows again by universal property
    of $F^*(a)$ and $G^*(b)$.
  \end{proof}
  \begin{cor}
    \label{cor:commute-free-mon-adj}
  Consider any diagram
  \begin{center}
    \diag{%
      𝒜 \& 𝒜 \\
      ℬ \& ℬ \\
      𝒞 \& 𝒞 \\
    }{%
      (m-1-1) edge[labela={U}] (m-1-2) %
      edge[twol={J}] (m-2-1) %
      (m-2-1) edge[labelb={V}] (m-2-2) %
      (m-1-2) edge[twor={J}] (m-2-2) %
      (l) edge[draw=none,labelon={⇓α}] (r)
      (m-3-1) edge[labelb={W}] (m-3-2) %
      (m-3-1) edge[twol={K}] (m-2-1) %
      (m-3-2) edge[twor={K}] (m-2-2) %
      (l) edge[draw=none,labelon={⇓β}] (r)
    }
  \end{center}
  of functors and natural transformations such that $J$ has a right
  adjoint, $𝒜$ and $𝒞$ are cocomplete, and $U$ and $W$ are finitary.
  
  Furthermore, let $M$ be the induced endofunctor on the comma
  category $J/K$, mapping $Ja \xrightarrow{f} Kc$ to
  $(JUa \xrightarrow{αₐ} VJa \xrightarrow{Vf} VKc \xrightarrow{β_c}
  KWc)$.

  Then, given an object $f:Ja → Kb$ of $J/K$, there is a unique
  morphism $f^*:JU^⋆ a → KW^⋆ b$ making the following diagram commute.
  \begin{center}
    \diag{%
      JUU^*a \& JU^*a \& Ja \\
      VJU^*a \& \\
      VKW^*b \\
      KWW^*b \& KW^*b \& Kb %
    }{%
      (m-1-1) edge[labela={Jνᵁₐ}] (m-1-2) %
      edge[labell={α_{U^*a}}] (m-2-1) %
      (m-2-1) edge[labell={Vf^*}] (m-3-1) %
      (m-3-1) edge[labell={β_{W^*b}}] (m-4-1) %
      (m-4-1) edge[labelb={K (νᵂ_b)}] (m-4-2) %
      (m-1-3) edge[labela={J (ηᵁₐ)}] (m-1-2) %
      edge[labelr={f}] (m-4-3) %
      (m-4-3) edge[labelb={K (ηᵂ_b)}] (m-4-2) %
      (m-1-2) edge[labelon={f^*}] (m-4-2) %
    }
  \end{center}
  Furthermore, the right-hand square above exhibits $f^*$ as
  a free $M$-algebra on $f$.
  \end{cor}
  \begin{proof}
    This follows from 
    Lemma~\ref{lem:commute-free-monads} with $F = U$, $G = W$, $K$ as $RK$ where
    $R$ is the right adjoint of $J$,
    by considering the mate
    $α' : UR → RV$ of $α:JU → VJ$, 
    defined as
    $UR \xrightarrow{ηUR} RJUR \xrightarrow{RαR}  RVJR \xrightarrow{RVε} RV$,
      (where $η$ and $ε$ denote the unit and the counit of the adjunction $J ⊣ R$)
    and composing it with $β$ to get a natural
    transformation $URK → RKW$.
  \end{proof}

\begin{proof}[Proof of Lemma~\ref{lem:pack-Sigma1}]
  Let us denote by $M$ the endofunctor on $𝐂ᴴ_𝐙$ mapping an object
  $𝐙 ← S → 𝐙$ to the the right pasting of Diagram~\ref{eq:SigmabulletRinit}:
  \[
    M
    \left(
      \Diag(.6,1){%
      }{%
        𝐙₀ \& 𝐙₀ \\
        𝐙₀ \& 𝐙₀ %
      }{%
        (m-1-1) edge[glob,labela={H₀}] (m-1-2) %
        edge[pro,labell={𝐙₁}] node (ll) {} (m-2-1) %
        (m-2-1) edge[glob,labelb={H₀}] (m-2-2) %
        (m-1-2) edge[pro,labelr={𝐙₁}] node (mm) {} (m-2-2) %
        (ll) edge[draw=none,labelon={S₁}] (mm)
      }%
    \right)
    =
          \Diag(.6,1){%
        \draw[->,glob,labela={H₀},rounded corners]
        (m-1-1) -- ($(m-1-1.north) + (90:5ex)$) --  node(a) {} ($(m-1-3.north) + (90:5ex)$) -- (m-1-3) %
        ; %
        \path[->,draw] %
        (m-1-2) edge[celllr={1}{1}] (a) %
        ; %
        \draw[->,glob,labelb={H₀},rounded corners]
        (m-2-1) -- ($(m-2-1.south) + (-90:5ex)$) --  node(b) {} ($(m-2-3.south) + (-90:5ex)$) -- (m-2-3) %
        ; %
        \path[->,draw] %
        (m-2-2) edge[celllr={1}{1}] (b) %
    ; %
      }{%
        𝐙₀ \& 𝐙₀ \& 𝐙₀ \\
        𝐙₀ \& 𝐙₀ \& 𝐙₀ %
      }{%
        (m-1-1) edge[glob,labela={H₀}] (m-1-2) %
        edge[pro,labell={𝐙₁}] node (ll) {} (m-2-1) %
        (m-2-1) edge[glob,labelb={H₀}] (m-2-2) %
        (m-1-2) edge[pro,labelr={𝐙₁}] node (mm) {} (m-2-2) %
        (m-1-2) edge[glob,labela={{∼₀^{⊗}}}] (m-1-3) %
        (m-2-2) edge[glob,labelb={{∼₀^{⊗}}}] (m-2-3) %
        (m-1-3) edge[pro,labelr={𝐙₁}] node (rr) {} (m-2-3) %
        (ll) edge[draw=none,labelon={S₁}] (mm)
        (mm) edge[draw=none,labelon={{∼₁^{⊗}}}] (rr)
      }%
  \]
  Now, by the packing lemma, it is enough to show that 
  \[
    M^⋆
    \left(
      \Diag(.6,1){%
      }{%
        𝐙₀ \& 𝐙₀ \\
        𝐙₀ \& 𝐙₀ %
      }{%
        (m-1-1) edge[glob,labela={H₀}] (m-1-2) %
        edge[pro,labell={𝐙₁}] node (ll) {} (m-2-1) %
        (m-2-1) edge[glob,labelb={H₀}] (m-2-2) %
        (m-1-2) edge[pro,labelr={𝐙₁}] node (mm) {} (m-2-2) %
        (ll) edge[draw=none,labelon={S₁}] (mm)
      }%
    \right)
    =
          \Diag(.6,1){%
        \draw[->,glob,labela={H₀},rounded corners]
        (m-1-1) -- ($(m-1-1.north) + (90:5ex)$) --  node(a) {} ($(m-1-3.north) + (90:5ex)$) -- (m-1-3) %
        ; %
        \path[->,draw] %
        (m-1-2) edge[celllr={1}{1}] (a) %
        ; %
        \draw[->,glob,labelb={H₀},rounded corners]
        (m-2-1) -- ($(m-2-1.south) + (-90:5ex)$) --  node(b) {} ($(m-2-3.south) + (-90:5ex)$) -- (m-2-3) %
        ; %
        \path[->,draw] %
        (m-2-2) edge[celllr={1}{1}] (b) %
    ; %
      }{%
        𝐙₀ \& 𝐙₀ \& 𝐙₀ \\
        𝐙₀ \& 𝐙₀ \& 𝐙₀ %
      }{%
        (m-1-1) edge[glob,labela={H₀}] (m-1-2) %
        edge[pro,labell={𝐙₁}] node (ll) {} (m-2-1) %
        (m-2-1) edge[glob,labelb={H₀}] (m-2-2) %
        (m-1-2) edge[pro,labelr={𝐙₁}] node (mm) {} (m-2-2) %
        (m-1-2) edge[glob,labela={{∼₀^{⊗*}}}] (m-1-3) %
        (m-2-2) edge[glob,labelb={{∼₀^{⊗*}}}] (m-2-3) %
        (m-1-3) edge[pro,labelr={𝐙₁}] node (rr) {} (m-2-3) %
        (ll) edge[draw=none,labelon={S₁}] (mm)
        (mm) edge[draw=none,labelon={{∼₁^{⊗*}}}] (rr)
      }%
  \]
  To this end, we are going to organise $𝐂_𝐙ᴴ$ as a comma category on which $M$
  acts, so as to apply Corollary~\ref{cor:commute-free-mon-adj}.
  Let $𝒜$ denote the category of objects $S₁$ equipped with a span $𝐙₁ ← S₁ → 𝐙₁
  $ and $ℬ$ denote the product category $ℬ_\source×ℬ_\but$,
  where $ℬ_σ$ is the category of
  objects $S₁$ equipped with a span $𝐙₀σ ← S₁ → 𝐙₀σ$.

  Let $J : 𝒜 → ℬ$ denote the functor mapping $𝐙₁ ← S₁ → 𝐙₁$ to
  $(𝐙₀\source ← S₁ → 𝐙₀\source, 𝐙₀\but ← S₁ → 𝐙₀\but)$, by postcomposing with the relevant morphisms.
  As the forgetful functor from a category of spans creates colimits, $J$ is
  cocontinuous and thus has a right adjoint by~\cite[Theorem~1.66]{Adamek}, since its domain is a
  locally presentable category.
  Let $K : 1 → ℬ$ be the functor selecting the pair $(𝐙₀\source ← H₀\source → 𝐙₀\source, 𝐙₀\but ← H₀\but → 𝐙₀\but)$.
  Now, it is straightforward to check that $𝐂_𝐙ᴴ$ is isomorphic to the comma category $J/K$.

  Next, we reconstruct $M$ as acting on $J/K$ through this isomorphism in order to fit the setting
  of Corollary~\ref{cor:commute-free-mon-adj}.

  Let
  \begin{itemize}
  \item 
    $U:𝒜 → 𝒜$ denote the functor mapping $𝐙₁ ← S₁ → 𝐙₁$ to $𝐙₁ ← S₁ ← S₁ ; ∼^⊗₁ → ∼^⊗₁ → 𝐙₁$;
  \item $V:ℬ → ℬ$ denote the functor $V_\source× V_\but$, where $V_σ:ℬ_σ → ℬ_σ$
    maps $𝐙₀σ ← S₁ → 𝐙₀σ$
    to $𝐙₀σ ← S₁;∼^{⊗}₀σ → 𝐙₀σ$;
    \item $W:1 → 1$ denote the identity endofunctor.
  \end{itemize}
  Now we apply Corollary~\ref{cor:commute-free-mon-adj} with suitable
  $α:JU → VJ$ and $β: WK → KW$ so that $M$ corresponds to our endofunctor
  through the isomorphism $𝐂_𝐙ᴴ≅J/K$.
  Since $U^⋆(𝐙₁ ← S₁ → 𝐙₁) = (𝐙₁ ← {S₁;{∼^{⊗*}₁}} → 𝐙₁)$, the only thing to check
  is that the proposed definition for $M^⋆(S₁)$ indeed defines a $M$-algebra,
  and that $S₁ → U^⋆(S₁)$ induces a morphism $S₁ → M^⋆(S₁)$, which is straightforward.
\end{proof}

Let us now turn to our final characterisation of $H$, which relies on
the following category, which is a relaxation of $𝐂ᴴ_𝐙$, in which the
left-hand object in~\eqref{eq:doublebicat} is only forced to coincide
with $𝐙$ on $\psh[ℂ₀]$.
\begin{defi}
  Let $𝐂ᴴ_{\lax}$ denote the category whose objects consist of objects $X₁$
  and $S₁$ in $\psh[ℂ₁]$, equipped with dashed arrows making the
  following diagram commute.
    \begin{equation}
      \diag(.5,.5){%
        𝐙₀\source \& H₀\source \& 𝐙₀\source  \\
        X₁ \& S₁ \& 𝐙₁ \\
        𝐙₀\but \& H₀\but \& 𝐙₀\but %
      }{%
        (m-1-1) edge[<-,labela={}] (m-1-2) %
        edge[<-,labell={},dashed] (m-2-1) %
        (m-2-1) edge[<-,dashed,labelb={}] (m-2-2) %
        (m-1-2) edge[<-,dashed,labelr={}] (m-2-2) %
        (m-1-2) edge[labela={}] (m-1-3) %
        (m-2-2) edge[dashed,labelb={}] (m-2-3) %
        (m-1-3) edge[<-,labelr={}] (m-2-3) %
        (m-2-1) edge[labell={},dashed] (m-3-1) %
        (m-3-1) edge[<-,labelb={}] (m-3-2) %
        (m-2-2) edge[dashed,labelr={}] (m-3-2) %
        (m-3-2) edge[labelb={}] (m-3-3) %
        (m-2-3) edge[labelr={}] (m-3-3) %
      }
      \label{eq:objectCH}
    \end{equation}
\end{defi}
\begin{rem}
 Using the notation of~§\ref{ss:doublecat}, an object of $𝐂ᴴ_{\lax}$ is a cell of the form
\begin{center}
  \diag{%
      𝐙₀ \& 𝐙₀  \\
      𝐙₀ \& 𝐙₀\rlap{.}
    }{%
      (m-1-1) edge[glob,labela={H₀}] (m-1-2) %
      edge[pro,twol={X₁}] (m-2-1) %
      (m-2-1) edge[glob,labelb={H₀}] (m-2-2) %
      (m-1-2) edge[pro,twor={𝐙₁}] (m-2-2) %
      (l) edge[draw=none,labelon={S₁}] (r)
    }%
\end{center}
Such an object may also be viewed as a span of the form $X ← S → 𝐙$ in
$𝐂$, which projects down to $𝐙₀ ← H₀ → 𝐙₀$ in $\psh[ℂ₀]$.
\end{rem}

Let us briefly relate $𝐂ᴴ_{\lax}$ to other useful categories.
  \begin{defi}
    Let $𝟐$ denote the free category on the graph $0 → 1$, and $𝐂^𝟐/𝐙$
    denote the comma category
    \begin{center}
      \Diag{%
        \laxpbk{m-2-1}{m-1-1}{m-1-2} %
      }{%
        𝐂^𝟐/𝐙  \& 1 \\
        𝐂^𝟐 \& 𝐂\rlap{,} %
      }{%
        (m-1-1) edge[labela={}] (m-1-2) %
        edge[labell={}] (m-2-1) %
        (m-2-1) edge[labelb={\dom}] (m-2-2) %
        (m-1-2) edge[labelr={𝐙}] (m-2-2) %
      }%
    \end{center}
    whose objects are spans of the form $X ← S → 𝐙$.
  \end{defi}

  \begin{defi}\label{def:functors}
  We  define a commuting diagram
  \begin{equation}
    \diag{%
      𝐂ᴴ_𝐙 \& 𝐂ᴴ_{\lax} \\
      𝐂/𝐙² \& 𝐂^𝟐/𝐙  \\
      \& \& 𝐂^𝟐 %
    }{%
      (m-1-1) edge[into,labela={ℐᴴ}] (m-1-2) %
      edge[labell={𝒰ᴴ_𝐙}] (m-2-1) %
      (m-2-1) edge[into,labelb={𝒥ᴴ}] (m-2-2) %
      edge[bend right=30,labelbl={𝒲}] (m-3-3) %
      (m-1-2) edge[labelr={𝒰ᴴ_{\lax}}] (m-2-2) %
      edge[bend left=30,labelar={𝒱}] (m-3-3) %
      (m-2-2) edge[labelon={\dom}] (m-3-3) %
    }
    \label{eq:functors}
  \end{equation}
  of functors:
  \begin{itemize}
  \item $𝒰ᴴ_𝐙$ and $𝒰ᴴ_{\lax}$ are the obvious forgetful functors,
  \item $ℐᴴ$ and $𝒥ᴴ$ are the obvious embeddings, and 
  \item $𝒱$ and $𝒲$ are defined by composition with the
    domain functor $\dom$.
  \end{itemize}
\end{defi}

\begin{prop}\label{prop:initialwowo}
  The initial object in $𝐂ᴴ_{\lax}$ is the span $𝐙₀ ← {H₀} → 𝐙$, i.e., the
  one with $X₁ = S₁ = 0$.
\end{prop}

Let us now introduce the endofunctor of which our characterisation of
$H$ will be an initia algebra.
\begin{defi}
  Let $Σ_{1,\lax}ᴴ∶ 𝐂ᴴ_{\lax} → 𝐂ᴴ_{\lax}$ map any object $X ← S → 𝐙$ to the following
  pasting.
      \begin{equation}
      \Diag{%
        \sqpath{m-1-1}{180}{7}{l}{m-3-1}{pro,labell={\check{Σ}₁(X)₁}}
        \sqpath{m-3-1}{-90}{5}{b}{m-3-4}{glob,labelb={H₀}} %
        \sqpath{m-1-1}{90}{5}{aa}{m-1-4}{glob,labela={H₀}} %
        \path (m-3-1) -- node (bm) {} (m-3-4) ; %
        \path[draw] %
        (m-2-1) edge[identity] (l) %
        (m-2-2) edge[celllr={1}{1},labela={}] (ddrrr) %
        ($(aa)+(-90:4ex)$) edge[celllr={0}{0}] (aa) %
        ($(b)+(90:4ex)$) edge[celllr={0}{0}] ($(b)+(90:1ex)$) %
        ($(a)+(-90:3ex)$) edge[celllr={0}{0}] (a) %
        ; %
      }{%
        𝐙₀ \& 𝐙₀ \& 𝐙₀ \& 𝐙₀ \\
        Σ₀(𝐙₀) \& Σ₀(𝐙₀) \&  \&  \\
        𝐙₀ \& 𝐙₀ \& 𝐙₀ \& 𝐙₀ %
      }{%
        (m-1-1) edge[glob,twoa={H₀}] (m-1-2) %
        (m-1-2) edge[identity] (m-1-3) %
        (m-1-3) edge[glob,labela={{∼₀^{⊗*}}}] (m-1-4)
        edge[pro,labelr={𝐙₁}] node (ddrrr) {} (m-3-3) %
        (m-2-1) edge[labell={}] (m-1-1) %
        (m-2-1) edge[glob,labela={Σ₀(H₀)}] (m-2-2) %
        edge[pro,labell={Σ₁^F(X)}] node (ddr) {} (m-3-1) %
        (m-3-1) edge[glob,labelb={H₀}] (m-3-2) %
        (m-2-2) edge[pro,labellat={\scriptscriptstyle Σ₁^F(𝐙)}{.1}] node (ddrr) {} (m-3-2) %
        edge[labelr={}] (m-1-2) %
        (m-3-2) edge[identity,labelb={}] (m-3-3) %
        (m-3-3) edge[glob,labelb={{∼₀^{⊗*}}}] (m-3-4) %
        (m-1-4) edge[pro,labelr={𝐙₁}] node (ddrrrr) {} (m-3-4) %
        (ddr) edge[draw=none,labelon={Σ₁^F(S)}] (ddrr) %
        (ddrrr) edge[draw=none,labelon={{∼₁^{⊗*}}}] (ddrrrr) %
      }
      \label{eq:SigmabulletRlax}
    \end{equation}
\end{defi}
\begin{rem}
  The difference with $Σ_{1,\pack}ᴴ$ is that, $X$ being different from
  $𝐙$ in general, we cannot use any $\check{Σ}₁$-algebra structure on
  the left.
\end{rem}

\begin{prop}
  The functor $Σ_{1,\lax}ᴴ∶ 𝐂ᴴ_{\lax} → 𝐂ᴴ_{\lax}$ is finitary.
\end{prop}
\begin{proof}
  Just as Proposition~\ref{prop:Sigbulletfinitary}.
\end{proof}

Let us now work towards our final characterisation of $H$.
\begin{defi}
  Let $δ$ denote the natural transformation
  \begin{center}
    \Diag{%
    }{%
      𝐂ᴴ_𝐙 \& 𝐂ᴴ_{\lax} \\
      𝐂ᴴ_𝐙 \& 𝐂ᴴ_{\lax} %
    }{%
      (m-1-1) edge[twoa={ℐᴴ}] (m-1-2) %
      edge[labell={Σ_{1,\pack}ᴴ}] (m-2-1) %
      (m-2-1) edge[twob={ℐᴴ}] (m-2-2) %
      (m-1-2) edge[labelr={Σ_{1,\lax}ᴴ}] (m-2-2) %
      (a) edge[labelr={δ},cell=.1] (b) %
    }
  \end{center}
  whose component
  at any  $𝐙 ← S → 𝐙$ in $𝐂ᴴ_𝐙$ is the following morphism
  in $𝐂ᴴ_{\lax}$.  

  \begin{center}
    \diag{%
      \check{Σ}₁(𝐙) \& \check{Σ}₁(S);{∼^{⊗*}} \& 𝐙 \\
      𝐙 \& \check{Σ}₁(S);{∼^{⊗*}} \& 𝐙 %
    }{%
      (m-1-1) edge[<-,labela={}] (m-1-2) %
      edge[labell={}] (m-2-1) %
      (m-1-2) edge[labela={}] (m-1-3) %
      edge[identity,labell={}] (m-2-2) %
      (m-1-3) edge[identity,labelr={}] (m-2-3) %
      (m-2-1) edge[<-,labela={}] (m-2-2) %
      (m-2-2) edge[labela={}] (m-2-3) %
    }
  \end{center}

  Furthermore, let $𝒦ᴴ∶ Σ_{1,\pack}ᴴ\alg → Σ_{1,\lax}ᴴ\alg$ denote the
  induced lifting of $ℐᴴ$, as in
  \begin{center}
    \diag{%
      Σ_{1,\pack}ᴴ\alg \& Σ_{1,\lax}ᴴ\alg \\
      𝐂ᴴ_𝐙 \& 𝐂ᴴ_{\lax}\rlap{,} %
    }{%
      (m-1-1) edge[labela={𝒦ᴴ}] (m-1-2) %
      edge[labell={}] (m-2-1) %
      (m-2-1) edge[labelb={ℐᴴ}] (m-2-2) %
      (m-1-2) edge[labelr={}] (m-2-2) %
    }
  \end{center}
  where both vertical arrows denote forgetful functors.
\end{defi}

\begin{lem}\label{lem:Hlax}
  The $Σ_{1,\lax}ᴴ$-algebra $𝒦ᴴ(H_𝐙)$ is initial.  In other words,
  letting $H_{\lax}$ denote any initial $Σ_{1,\lax}ᴴ$-algebra, we have
  $𝒦ᴴ(H_𝐙) ≅ H_{\lax}$.
\end{lem}

\begin{proof}
In order to apply Corollary~\ref{cor:commute-free-mon-adj},
we organise $𝐂ᴴ_{\lax}$ as a comma category
  $J/𝐂ᴴₗ$, decomposing its left and middle/right parts, with $J∶ 𝐂ᴴᵣ → 𝐂ᴴₗ$ defined as follows:
  \begin{itemize}
  \item $𝐂ᴴₗ=\psh[ℂ₁]/Δ(𝐙₀)$;
  \item $𝐂ᴴᵣ$ is the comma category
    \begin{center}
      \Diag{%
        \laxpbk{m-2-1}{m-1-1}{m-1-2} %
      }{%
        𝐂ᴴᵣ \& 1 \\
        \psh[ℂ₁]/Δ(H₀) \& \psh[ℂ₁]/Δ(𝐙₀)\rlap{;}
      }{%
        (m-1-1) edge[labela={}] (m-1-2) %
        edge[labell={prᵣ}] (m-2-1) %
        (m-2-1) edge[labelb={\psh[ℂ₁]/Δ(π₂)}] (m-2-2) %
        (m-1-2) edge[labelr={𝐙}] (m-2-2) %
      }%
      \end{center}
    concretely, objects consist of a presheaf $S₁ ∈ \psh[ℂ₁]$,
    together with dashed maps making the following diagram commute;
    \begin{equation}
      \diag(.5,1){%
        H₀\source \& 𝐙₀\source  \\
        S₁ \& 𝐙₁ \\
        H₀\but \& 𝐙₀\but %
      }{%
        (m-1-1) edge[<-,dashed,labelr={}] (m-2-1) %
        (m-1-1) edge[labela={π₂\source}] (m-1-2) %
        (m-2-1) edge[dashed,labelb={}] (m-2-2) %
        (m-1-2) edge[<-,labelr={}] (m-2-2) %
        (m-2-1) edge[dashed,labelr={}] (m-3-1) %
        (m-3-1) edge[labelb={π₂\but}] (m-3-2) %
        (m-2-2) edge[labelr={}] (m-3-2) %
      }
      \label{eq:CHl}
    \end{equation}
  \item $J$ is the composite
    \[𝐂ᴴᵣ \xto{prᵣ} \psh[ℂ₁]/Δ(H₀) \xto{\psh[ℂ₁]Δ(π₁)}
    \psh[ℂ₁]/Δ(𝐙₀)\rlap{;}\] concretely, $J$ maps any
    object~\eqref{eq:CHl} to the following diagram.
    \begin{center}
      \diag(.5,1){%
        𝐙₀\source \& H₀\source  \\
         \&  S₁ \\
        𝐙₀\but \& H₀\but %
      }{%
        (m-1-1) edge[<-, labela={π₁\source}] (m-1-2) %
        (m-1-2) edge[<-,dashed,labelr={}] (m-2-2) %
        (m-3-1) edge[<-,labelb={π₁\but}] (m-3-2) %
        (m-2-2) edge[labelr={},dashed] (m-3-2) %
      }
    \end{center}
  \end{itemize}
  We note that $J$ is cocontinuous since colimits are computed
  pointwise in $𝐂ᴴₗ$ and $𝐂ᴴᵣ$, and thus has a right adjoint $R$,
  since its domain is locally presentable,
  by~\cite[Theorem~1.66]{Adamek}.

  Furthermore, $𝐂ᴴ_𝐙$ is isomorphic
  to the comma category $J/𝐙$\footnote{Note that this decomposition
    as a comma category differs from the one in the proof of
    Lemma~\ref{lem:pack-Sigma1}.}.

  We then define functors $Σᴴ_{1,r}$ and $\restr{(\check{Σ}₁)}{𝐙₀}$
  and a natural transformation $h$ as in
    \begin{center}
    \diag{%
      𝐂ᴴᵣ \& 𝐂ᴴᵣ \\
      𝐂ᴴₗ \& 𝐂ᴴₗ\rlap{,}
    }{%
      (m-1-1) edge[labela={Σᴴ_{1,r}}] (m-1-2) %
      edge[twol={J}] (m-2-1) %
      (m-2-1) edge[labelb={\restr{(\check{Σ}₁)}{𝐙₀}}] (m-2-2) %
      (m-1-2) edge[twor={J}] (m-2-2) %
      (l) edge[draw=none,labelon={⇓h}] (r)
    }
  \end{center}
  as follows:
  \begin{itemize}
  \item $\restr{(\check{Σ}₁)}{𝐙₀}$ is $\check{Σ}₁$ restricted to the
    fibre of $\psh[ℂ₁]/Δ → \psh[ℂ₀]$ over $𝐙₀$, as in the left part
    of~\eqref{eq:SigmabulletRlax};
  \item $Σᴴ_{1,r}$ acts as the right part
    of~\eqref{eq:SigmabulletRlax};
  \item $h$ connects both parts using the left projection.
  \end{itemize}
  Now, we apply Corollary~\ref{cor:commute-free-mon-adj} twice with
  $U = Σᴴ_{1,r}$, $V = \restr{(\check{Σ}₁)}{𝐙₀}$ and $α = h$:
  \begin{enumerate}
  \item for $W=V$, $K$ the identity endofunctor, $β$ the identity
    natural transformation, and $f = \id_∅∶ J(∅) → ∅$, the induced
    endofunctor $M$ is precisely $Σ_{1,\lax}ᴴ$, and
    Corollary~\ref{cor:commute-free-mon-adj} yields a morphism, say
    $f^*$, making the diagram below commute;
  \begin{center}
    \Diag{%
      \justify{m-1-1}{A}{m-3-2} %
    }{%
      J (Σᴴ_{1,r} ((Σᴴ_{1,r})^* (∅))) \& J ((Σᴴ_{1,r})^* (∅)) \& J (∅) \\
      \restr{(\check{Σ}₁)}{𝐙₀} (J ((Σᴴ_{1,r})^* (∅))) \\
      \restr{(\check{Σ}₁)}{𝐙₀} (\restr{(\check{Σ}₁)}{𝐙₀}^* (∅)) \& \restr{(\check{Σ}₁)}{𝐙₀}^* (∅) \& ∅ %
    }{%
      (m-1-1) edge[labela={J (ν^{Σᴴ_{1,r}}_∅)}] (m-1-2) %
      edge[labell={h_{(Σᴴ_{1,r})^*(∅)}}] (m-2-1) %
      (m-2-1) edge[shorten >=.5em,labell={\restr{(\check{Σ}₁)}{𝐙₀}(f^*)}] (m-3-1) %
      (m-3-1) edge[labelb={ν^{\restr{(\check{Σ}₁)}{𝐙₀}}_∅}] (m-3-2) %
      (m-1-2) edge[shorten >=.5em,labelr={f^*}] (m-3-2) %
      (m-1-2) edge[<-,labela={J(!)}] (m-1-3) %
      (m-3-2) edge[<-,labelb={!}] (m-3-3) %
      (m-1-3) edge[labelr={!}] (m-3-3) %
    }
  \end{center}
\item for $K:1 → 𝐂ᴴₗ$ selecting $𝐙₀\source ← 𝐙₁ → 𝐙₀\but$, $β$ induced
  by the $V$-algebra structure on it, and $g = {!}_𝐙∶ J(∅) = ∅ → 𝐙$,
  the induced endofunctor $N$ is precisely $Σ_{1,\pack}ᴴ$, and
  Corollary~\ref{cor:commute-free-mon-adj} yields a unique morphism,
  say $g^*$, making the diagram below commute.
  \begin{center}
    \Diag{%
      \justify{m-1-1}{B}{m-4-2} %
      }{%
      J (Σᴴ_{1,r} ((Σᴴ_{1,r})^* (∅))) \& J ((Σᴴ_{1,r})^* (∅)) \& J (∅) \\
      \restr{(\check{Σ}₁)}{𝐙₀} (J ((Σᴴ_{1,r})^* (∅))) \\
      \restr{(\check{Σ}₁)}{𝐙₀} (𝐙) \& \\
      𝐙 \& 𝐙 \& 𝐙  \\ %
    }{%
      (m-1-1) edge[labela={J (ν^{Σᴴ_{1,r}}_∅)}] (m-1-2) %
      edge[labell={h_{(Σᴴ_{1,r})^*(∅)}}] (m-2-1) %
      (m-2-1) edge[labell={\restr{(\check{Σ}₁)}{𝐙₀}(g^*)}] (m-3-1) %
      (m-3-1) edge[labell={ν^{\restr{(\check{Σ}₁)}{𝐙₀}}_∅}] (m-4-1) %
      (m-1-2) edge[labelr={g^*}] (m-4-2) %
      (m-1-2) edge[<-,labela={J(!)}] (m-1-3) %
      (m-4-1) edge[identity,labelb={}] (m-4-2) %
      (m-4-2) edge[identity,labelb={}] (m-4-3) %
      (m-1-3) edge[labelr={!}] (m-4-3) %
    }
  \end{center}
  \end{enumerate}
  Both right-hand squares are trivial, and both left-hand hand squares
  are in fact the same, so that $f^* = g^*$.  Finally, in terms of
  $J/𝐂ᴴₗ$, the statement claims that the $M$-algebra structure
  of $f^*$ decomposes as
  \begin{center}
    \diag{%
      \restr{(\check{Σ}₁)}{𝐙₀} (\restr{(\check{Σ}₁)}{𝐙₀}^* (∅))\& \restr{(\check{Σ}₁)}{𝐙₀} (J ((Σᴴ_{1,r})^* (∅)))
      \& J (Σᴴ_{1,r} ((Σᴴ_{1,r})^* (∅))) \\
      \restr{(\check{Σ}₁)}{𝐙₀}^* (∅) \& \& J (Σᴴ_{1,r} ((Σᴴ_{1,r})^* (∅))) \\
      \restr{(\check{Σ}₁)}{𝐙₀}^* (∅) \& \& J ((Σᴴ_{1,r})^* (∅))\rlap{,} %
    }{%
      (m-1-1) edge[<-,labela={\restr{(\check{Σ}₁)}{𝐙₀}(f^*)}] (m-1-2) %
      (m-1-2) edge[<-,labela={h_{(Σᴴ_{1,r})^*(∅)}}] (m-1-3) %
      (m-2-1) edge[<-,labela={\restr{(\check{Σ}₁)}{𝐙₀}(f^*)}] (m-2-3) %
      (m-3-1) edge[<-,labela={f^*}] (m-3-3) %
      (m-1-1) edge[labell={ν^{\restr{(\check{Σ}₁)}{𝐙₀}}_∅}] (m-2-1) %
      (m-1-3) edge[identity] (m-2-3) %
      (m-2-1) edge[identity] (m-3-1) %
      (m-2-3) edge[labelr={J (ν^{Σᴴ_{1,r}}_∅)}] (m-3-3) %
    }
  \end{center}
  which in terms of spans yields exactly the desired form.
  \end{proof}

  \begin{cor}\label{cor:HHlax}
    We have $𝒥ᴴ(H) ≅ 𝒰ᴴ_{\lax}(H_{\lax})$ in $𝐂^𝟐/𝐙$.
\end{cor}
\begin{proof}
  We have
  $𝒥ᴴ(H) = 𝒥ᴴ(𝒰ᴴ_𝐙(H_𝐙)) = 𝒰ᴴ_{\lax}(ℐᴴ(H_𝐙)) ≅ 𝒰ᴴ_{\lax}(H_{\lax})$
  in $𝐂^𝟐/𝐙$.
\end{proof}
  
    \subsection{Simulation property}
    Our next goal is to prove the following.

\howisim*

    The rest of this subsection is devoted to the proof.
    
    For substitution-closedness, ${∼₀^{⊗}}$ is reflexive and by
    Lemma~\ref{lem:simintowow} we have a span morphism
    ${∼₀^{⊗}} → {H₀}$, so we may form the composite
    \[{H₀}⊗𝐙 → {H₀}⊗{∼₀^{⊗}} → {H₀}⊗{H₀} → {H₀}\rlap{,}\]
    where the last morphism is the monoid multiplication of
    ${H₀}$, as established in Proposition~\ref{prop:howomonoid}.
    
    For the simulation property, we will use the characterisation of
    ${H}$ as $𝐙_{Σ_{1,\lax}ᴴ}$.  For this, we need to lift the notion
    of simulation from $𝐂$ to $𝐂ᴴ_{\lax}$.  Recalling the functor
    $𝒱∶ 𝐂ᴴ_{\lax} → 𝐂^𝟐$ from Definition~\ref{def:functors}, which
    maps each span $X ← S → 𝐙$ in $𝐂ᴴ_{\lax}$ to its left-hand leg
    $X ← S$, we have:
    \begin{defi}\label{def:CHsim}
      An object $X ← S → 𝐙$ of $𝐂ᴴ_{\lax}$ is a \emph{simulation} iff its
      image by $𝒱$, i.e., $X ← S$, is a functional bisimulation.
    \end{defi}
    Next, we want to show that the computation of $H$ as an initial
    chain in $𝐂ᴴ_{\lax}$ is preserved by $𝒱$. We intend to use this to apply
    Lemma~\ref{lem:colimitsim}, which will reduce our goal to proving
    that each object of the chain is a functional bisimulation.
    \begin{lem}\label{lem:Vcocont}
      All of the following functors are cocontinuous:
      \begin{enumerati}
      \item \label{item:cst} any functor from the terminal category;
      \item \label{item:K} both functors
        $\psh[ℂ₁]/Δ(πᵢ)∶ \psh[ℂ₁]/Δ(H₀) → \psh[ℂ₁]/Δ(𝐙₀)$, for
        $i=1,2$, which map any span $H₀\source ← S₁ → H₀\but$ to
        \[𝐙₀\source \xot{πᵢ} H₀\source ← S₁ → H₀\but \xto{πᵢ} 𝐙₀t;\]
      \item \label{item:CHl:CH} the projection functor
        $prᵣ∶ 𝐂ᴴᵣ → \psh[ℂ₁]/Δ(H₀)$, mapping any object~\eqref{eq:CHl}
        to (the pairing of) its left-hand border
        $H₀\source ← S₁ → H₀\but$;
      \item \label{item:CZCD} the embedding
        $\psh[ℂ₁]/Δ(𝐙₀) ↪ \psh[ℂ₁]/Δ = 𝐂$;
      \item \label{item:CHCD} the embedding
        $\psh[ℂ₁]/Δ(H₀) ↪ \psh[ℂ₁]/Δ = 𝐂$; and
      \item \label{item:V} $𝒱∶ 𝐂ᴴ_{\lax} → 𝐂^𝟐$.
    \end{enumerati}
    \end{lem}
    \begin{proof}\hfill
      \begin{enumerati}
      \item Trivial.
      \item As post-composition functors, these have right adjoints.
      \item $𝐂ᴴᵣ$ is by definition the comma category
        \begin{center}
          \Diag{%
            \laxpbk{m-2-1}{m-1-1}{m-1-2} %
          }{%
            𝐂ᴴᵣ \& 1 \\
            \psh[ℂ₁]/Δ(H₀) \& \psh[ℂ₁]/Δ(𝐙₀)\rlap{,} %
          }{%
            (m-1-1) edge[labela={}] (m-1-2) %
            edge[labell={prᵣ}] (m-2-1) %
            (m-2-1) edge[labelb={\psh[ℂ₁]/Δ(π₂)}] (m-2-2) %
            (m-1-2) edge[labelr={𝐙}] (m-2-2) %
          }%
        \end{center}
        so by~\ref{item:cst}, \ref{item:K}, and
        Proposition~\ref{prop:projcreates}, the projection functor
        \[⟨prᵣ,{!}⟩∶ 𝐂ᴴₗ → \psh[ℂ₁]/Δ(H₀)×1 ≅ \psh[ℂ₁]/Δ(H₀)\]
        to the product preserves colimits.
      \item In the commuting square
        \begin{center}
          \diag{%
            \psh[ℂ₁]/Δ(𝐙₀) \& \psh[ℂ₁]/Δ \\
            \psh[ℂ₁] × 1 \& \psh[ℂ₁] × \psh[ℂ₀]\rlap{,} %
          }{%
            (m-1-1) edge[into,labela={}] (m-1-2) %
            edge[labell={⟨pr₁,{!}⟩}] (m-2-1) %
            (m-2-1) edge[labelb={\psh[ℂ₁]× Δ(𝐙₀)}] (m-2-2) %
            (m-1-2) edge[labelr={⟨pr₁,pr₀⟩}] (m-2-2) %
          }
        \end{center}
        both vertical functors (which are the canonical projection
        functors) create colimits. Furthermore,
        the bottom functor preserves them, because
        \begin{itemize}
        \item colimits are pointwise in a product of cocomplete
          categories, and
        \item any functor from the terminal category is cocontinuous.
        \end{itemize}
        Thus, the top functor is cocontinuous, as desired.
      \item Same, with the square 
        \begin{center}
          \diag{%
            \psh[ℂ₁]/Δ(H₀) \& \psh[ℂ₁]/Δ \\
            \psh[ℂ₁] × 1 \& \psh[ℂ₁] × \psh[ℂ₀]\rlap{.} %
          }{%
            (m-1-1) edge[into,labela={}] (m-1-2) %
            edge[labell={}] (m-2-1) %
            (m-2-1) edge[labelb={\psh[ℂ₁]× Δ(H₀)}] (m-2-2) %
            (m-1-2) edge[labelr={}] (m-2-2) %
          }
        \end{center}
      \item The functor $𝒱∶ 𝐂ᴴ_{\lax} → 𝐂^𝟐$ is induced by universal properties
        of lax limits as in\begin{center}
          \Diag{%
            \laxpbk[1.7em]{m-1-2}{m-2-1}{m-3-2} %
            \laxpbk[1.7em]{m-1-4}{m-2-3}{m-3-4} %
            \laxpbk[1.7em]{m-1-6}{m-2-5}{m-3-6} %
            \path[draw,->]
            (m-2-1) edge[fore,labela={}] (m-2-3) %
            (m-2-3) edge[fore,labela={}] (m-2-5) %
            ; %
            \twocell[.3]{m-1-4}{m-1-5}{m-1-4}{m-3-4}{}{%
              cell=0.1,bend left,labelbr={α}}
          }{%
            \& 𝐂ᴴᵣ \& \& \psh[ℂ₁]/Δ(H₀) \& \& \psh[ℂ₁]/Δ \\
            𝐂ᴴ_{\lax} \& \& 𝐂ᴴₘ \& \& 𝐂^𝟐 \\
            \& 𝐂ᴴₗ \& \& \psh[ℂ₁]/Δ(𝐙₀) \& \& \psh[ℂ₁]/Δ\rlap{,}
          }{
            (m-1-2) edge[labellat={J}{.7}] (m-3-2) %
            (m-1-4) edge[labelrat={\psh[ℂ₁]/Δ(π₁)}{.8}] (m-3-4) %
            (m-1-6) edge[identity] (m-3-6) %
            (m-1-2) edge[labela={prᵣ}] (m-1-4) %
            (m-3-2) edge[identity,labela={}] (m-3-4) %
            (m-1-4) edge[into,labela={}] (m-1-6) %
            (m-3-4) edge[into,labela={}] (m-3-6) %
            (m-2-1) edge[labela={}] (m-1-2) %
            edge[labelb={}] (m-3-2) %
            (m-2-3) edge[labela={}] (m-1-4) %
            edge[shorten >=1em,labelb={}] (m-3-4) %
            (m-2-5) edge[labela={}] (m-1-6) %
            edge[labelb={}] (m-3-6) %
          }%
          \end{center}
          where $𝐂ᴴₘ$ is defined as the lax limit of $\psh[ℂ₁]/Δ(π₁)$,
          and $α$, which is induced by universal property
          of $\psh[ℂ₁]/Δ$,
          has as component at any $p∶ X → Δ(H₀)$ the morphism
          \begin{center}
            \diag{%
              X \& X \\
              Δ(H₀) \& Δ(𝐙₀) %
            }{%
              (m-1-1) edge[identity,labela={}] (m-1-2) %
              edge[labell={p}] (m-2-1) %
              (m-2-1) edge[labelb={Δ(π₁)}] (m-2-2) %
              (m-1-2) edge[labelr={Δ(π₁)∘p}] (m-2-2) %
            }
          \end{center}
          in $\psh[ℂ₁]/Δ$.  For each lax limit, the given functor is
          cocontinuous, so the projection functor to the product
          creates colimits.  We thus get a diagram
          \begin{center}
            \diag(1,2){%
              𝐂ᴴ_{\lax} \& 𝐂ᴴₘ \& 𝐂^𝟐 \\
              𝐂ᴴᵣ×𝐂ᴴₗ \& \psh[ℂ₁]/Δ(H₀)×\psh[ℂ₁]/Δ(𝐙₀) \& (\psh[ℂ₁]/Δ)² %
            }{%
              (m-1-1) edge[labela={}] (m-1-2) %
              edge[labell={}] (m-2-1) %
              edge[bend left=10,labela={𝒱}] (m-1-3) %
              (m-2-1) edge[labelb={prᵣ×\id}] (m-2-2) %
              (m-1-2) edge[labelr={}] (m-2-2) %
              (m-1-2) edge[labela={}] (m-1-3) %
              (m-2-2) edge[into,labelb={}] (m-2-3) %
              (m-1-3) edge[labelr={}] (m-2-3) %
            }
          \end{center}
          in which all vertical functors create colimits, and both
          bottom functors are cocontinuous by the previous points.
          Thus, the top functor $𝒱$ is cocontinuous as desired. \qedhere
    \end{enumerati}
    \end{proof}

    Our next step, in order to apply Lemma~\ref{lem:colimitsim}, is to
    prove that each object of the initial chain for $H_𝐙$ is a
    simulation in $𝐂ᴴ_{\lax}$. This will follow from the next result.
\begin{lem}\label{lem:Sigmawowsim}
  If $Σ₁$ preserves functional bisimulations and $S ∈ 𝐂ᴴ_{\lax}$ is a
  simulation, then so is $Σ_{1,\lax}ᴴ(S)$.
\end{lem}
\begin{proof}
  The pasting~\eqref{eq:SigmabulletRlax} is isomorphic to the following.
  \begin{equation}
      \Diag{%
        \sqpath{m-1-1}{180}{7}{l}{m-3-1}{pro,labell={\check{Σ}₁(X)₁}}
        \sqpath{m-3-1}{-90}{5}{b}{m-3-4}{glob,labelb={H₀}} %
        \path (m-2-1) -- node (am) {} (m-2-4) ; %
        \path (m-3-1) -- node (bm) {} (m-3-4) ; %
        \path[draw] %
        (m-2-1) edge[identity] (l) %
        ($(a)+(-90:4ex)$) edge[celllr={0}{0}] (a) %
        ($(b)+(90:4ex)$) edge[celllr={0}{0}] ($(b)+(90:1ex)$) %
        ; %
      }{%
        𝐙₀ \& \& \& 𝐙₀ \\
        Σ₀(𝐙₀) \& Σ₀(𝐙₀) \& 𝐙₀ \& 𝐙₀ \\
        𝐙₀ \& 𝐙₀ \& 𝐙₀ \& 𝐙₀ %
      }{%
        (m-1-1) edge[glob,twoa={H₀}] (m-1-4) %
        (m-2-1) edge[labell={}] (m-1-1) %
        (m-1-4) edge[identity,labelr={}] (m-2-4) %
        (m-2-1) edge[glob,labela={Σ₀(H₀)}] (m-2-2) %
        edge[pro,labell={Σ₁^F(X)}] node (ddr) {} (m-3-1) %
        (m-3-1) edge[glob,labelb={H₀}] (m-3-2) %
        (m-2-2) edge[pro,labellat={\scriptscriptstyle Σ₁^F(𝐙)}{.1}] node (ddrr) {} (m-3-2) %
        (m-2-2) edge[labela={}] (m-2-3) %
        (m-3-2) edge[identity,labelb={}] (m-3-3) %
        (m-2-3) edge[pro,labelr={𝐙₁}] node (ddrrr) {} (m-3-3) %
        (m-2-3) edge[glob,labela={∼₀^{⊗*}}] (m-2-4) %
        (m-3-3) edge[glob,labelb={∼₀^{⊗*}}] (m-3-4) %
        (m-2-4) edge[pro,labelr={𝐙₁}] node (ddrrrr) {} (m-3-4) %
        (ddr) edge[draw=none,labelon={Σ₁^F(S)}] (ddrr) %
        (ddrr) edge[celllr={4}{4}] (ddrrr) %
        (ddrrr) edge[draw=none,labelon={{∼₁^{⊗*}}}] (ddrrrr) %
      }
      \label{eq:SigmabulletRprimedeux}
    \end{equation}
    By Propositions~\ref{prop:simpbk} and~\ref{prop:embedcells}, it
    suffices to show that all non-identity cells
    in~\eqref{eq:SigmabulletRprimedeux} are simulations if $S$ is.
    Let us run through them, left-to-right, top-to-bottom:
      \begin{itemize}
      \item The top cell is a simulation because
        \begin{itemize}
        \item $𝐙₀$ is the initial $(I+Σ₀)$-algebra, so
          $I → 𝐙₀ ← Σ₀(𝐙₀)$ is a coproduct diagram; 
        \item similarly, by Proposition~\ref{prop:initalt},
          ${H₀}$ is the initial algebra for the endofunctor
          $X ↦ (I;{∼₀^{⊗*}} + Σ₀(X);{∼₀^{⊗*}})$, so
          $I;{∼₀^{⊗*}} → {H₀} ← Σ₀({H₀});{∼₀^{⊗*}}$ is a
          coproduct diagram;
        \item so, by extensivity of $\psh[ℂ₀]$, both squares in
          \begin{center}
            \Diag{%
              \stdpbk
              \pbk{m-2-3}{m-1-3}{m-1-2} %
            }{%
              I;{∼₀^{⊗*}} \& {H₀} \& Σ₀({H₀});{∼₀^{⊗*}} \\
              I \& 𝐙₀ \& Σ₀(𝐙₀) %
            }{%
              (m-1-1) edge[labela={}] (m-1-2) %
              edge[labell={}] (m-2-1) %
              (m-2-1) edge[labelb={}] (m-2-2) %
              (m-1-2) edge[labelr={}] (m-2-2) %
              (m-1-2) edge[<-,labela={}] (m-1-3) %
              (m-2-2) edge[<-,labelb={}] (m-2-3) %
              (m-1-3) edge[labelr={}] (m-2-3) %
            }
          \end{center}
          are pullbacks. The right-hand one is mapped by $Δ_\source$
          (which, as a right adjoint, preserves pullbacks) precisely to
          the top left square of the top cell.
      \end{itemize}
        \item The first little cell $Σ₁^F(S)$ is a simulation by hypothesis.
        \item The middle little cell is trivially a simulation.
        \item The third little cell $∼₁^{⊗*}$ is a simulation because
          simulations are closed under transitive closure.
        \item Finally, the bottom cell is trivially a simulation. \qedhere
      \end{itemize}
\end{proof}

Finally:
\begin{proof}[Proof of Lemma~\ref{lem:howisim}]
  Let us start by observing that, because the domains and codomains of
  all $s_L$ are representable, hence finitely presentable, functional
  bisimulations are closed under filtered colimits in $𝐂^𝟐$, by
  Lemma~\ref{lem:colimitsim}.

  Now, the first projection $𝐙₀ ← H₀$ is trivially a functional
  bisimulation, since there are no transitions in $𝐙₀$, so by
  Proposition~\ref{prop:initialwowo} the initial object
  $∅_{𝐂ᴴ_{\lax}}$ of $𝐂ᴴ_{\lax}$ is a simulation in the sense of
  Definition~\ref{def:CHsim}.  Hence, by induction, using
  Lemma~\ref{lem:Sigmawowsim}, so are all objects of the initial chain
  of $Σᴴ_{1,\lax}$.  Thus, by Definition~\ref{def:CHsim}, $𝒱$ maps
  this initial chain to a chain of functional bisimulations in $𝐂^𝟐$
  and so
  \[
  \begin{array}{rcll}
    𝒱(H_{\lax}) & ≔ & 𝒱(\colimₙ (Σ_{1,\lax}ᴴ)ⁿ(∅_{𝐂ᴴ_{\lax}})) \\
    & ≅ & \colimₙ 𝒱((Σ_{1,\lax}ᴴ)ⁿ(∅_{𝐂ᴴ_{\lax}})) & \mbox{(by Lemma~\ref{lem:Vcocont})}
  \end{array}\] is a functional bisimulation by closedness of functional bisimulations
  under filtered colimits.

  Finally, this entails that the left-hand leg $𝒲(H)$ of the Howe
  closure is a functional bisimulation, because we have
  \begin{center}
    \hfill
  $\begin{array}[b]{rcll}
    𝒲(H) & = & \dom (𝒥ᴴ(H)) & \mbox{(by definition of $𝒲$)}\\
         & = & \dom (𝒰ᴴ_{\lax} (H_{\lax})) &
                   \mbox{(by Corollary~\ref{cor:HHlax})} \\
         & = & 𝒱(H_{\lax})  & \mbox{(by definition of $𝒱$).} %
  \end{array}$ \qedhere
\end{center}
\end{proof}

\subsection{Symmetry of transitive closure}
\label{ss:sym-tran-clos}
In this section, we prove the remaining Lemmas~\ref{lem:howtranssym}
and~\ref{lem:projsimsym}.
Let us first recall the former:
\howtranssym*

By Lemma~\ref{lem:spansymtrans}, Lemma~\ref{lem:howtranssym} will
follow if we construct a span morphism ${H₀} → {H₀^{\overline{+}†}}$.
As $H₀$ is an initial algebra for $I + Σ₀ᴴ$
(Proposition~\ref{prop:howe-clos-state-ini-algebra}),
it suffices to prove the following lemma.
\begin{lem}\label{lem:wowomonoid}
  The span ${H₀^{\overline{+}†}}$ has an algebra structure for $(I+Σ₀ᴴ)$.
\end{lem}
This relies on the following lemmas, used in particular with $F= Σ₀$.
The first one is well known:
\begin{lem}
  \label{lem:F-alg-creation}
  Given an endofunctor $F$ on some category $𝒞$, the forgetful functor
  $F\alg → 𝒞$ creates limits, and all colimits that $F$ preserves.
\end{lem}
\begin{lem}
  \label{lem:F-alg-image-facto}
  Given an endofunctor $F$ on a regular category $𝒞$, if $F$ preserves
  reflexive coequalisers, then the forgetful functor from the category
  of $F$-algebras creates image factorisations.
\end{lem}
\begin{proof}
  Suppose given an algebra morphism $A \xto{f} B$.  The image
  factorisation is obtained as the (reflexive) coequaliser of the
  kernel pair $A×_f A \rightrightarrows A$. The diagram of this
  reflexive coequaliser lifts to $F\alg$, hence so does the
  coequaliser, by the previous lemma.
\end{proof}
\begin{proof}[Proof of Lemma~\ref{lem:wowomonoid}]
  We need to find algebra structures on $H₀^{\overline{+}†}$ for $I$, $Σ₀$,
  and $- ; ∼₀^{⊗}$.  For $I$, we have the morphism $I → 𝐙₀ →
  H₀^† → H₀^{\overline{+}†}$.

  For $Σ₀$, note that by
  Lemmas~\ref{lem:rel-trans-filtered-colim} and Corollary~\ref{cor:commute-sym-trans},
  $H₀^{\overline{+}†} $ is the colimit of the chain
  \[
    𝐙₀ → \overline{H₀^†} ≅ \overline{H₀^†;𝐙₀} → \overline{H₀^†;H₀^†} ≅ \overline{H₀^†;H₀^†;𝐙₀} → \overline{H₀^†;H₀^†;H₀^†} → {…}
  \]
  As it is filtered and thus sifted, and $Σ₀$ preserves sifted
  colimits by hypothesis, by Lemma~\ref{lem:F-alg-creation}, it is
  enough to show that each $\overline{H₀^†;…;H₀^†}$ has a structure of
  $Σ₀$-algebra (morphisms in the above chain are then automatically
  algebra morphisms because the involved spans are relations).  But,
  $Σ₀$ also preserves reflexive coequalisers (which are sifted
  colimits), so, by Lemma~\ref{lem:F-alg-image-facto}, the forgetful
  functor from $Σ₀$-algebras creates image factorisations. It is thus
  enough to equip $H₀^†;…;H₀^†$ with $Σ₀$-algebra structure, which is
  straightforward because $H₀$ is already an algebra and algebras are
  stable under pullbacks (Lemma~\ref{lem:F-alg-creation}).

  It remains to find a suitable morphism
  ${H₀^{\overline{+}†}};{∼₀^⊗} → {H₀^{\overline{+}†}}$, or equivalently, by applying
  ${-}^{†}$, a morphism ${∼₀^{⊗†}};H₀^{\overline{+}} → H₀^{\overline{+}}$.  But by symmetry of
  $∼₀^⊗$,
  we have the composite
  \begin{center}\hfill
    ${∼₀^{⊗†}};H₀^{\overline{+}} → {∼₀^⊗};H₀^{\overline{+}} →  
    {H₀};H₀^{\overline{+}} →  {H₀^{\overline{+}}}\rlap{.}$ \qedhere
  \end{center}
\end{proof}

Finally, it remains to prove:
\projsimsym*
\begin{proof}
  First, consider the relation $\overline{R}$ induced by $R$ by the image
  factorisation $R ↠ \overline{R} ↪ 𝐙×𝐙$.
  $\overline{R}$ is still a substitution-closed simulation and 
  $\overline{R}₀$ is symmetric.
  Now, we define $R'$ as follows:
  \begin{itemize}
  \item $R'₀ = \overline{R}₀$
  \item $R'₁ $ is the limit of the following diagram:
    \[
        \diag(.5,.5){%
          𝐙₀\source \& R'₀\source \& 𝐙₀\source  \\
          𝐙₁ \& R'₁ \& 𝐙₁ \\
          𝐙₀\but \& R'₀\but \& 𝐙₀\but        \rlap{.} %
        }{%
          (m-1-1) edge[<-,labela={}] (m-1-2) %
          edge[<-,labell={}] (m-2-1) %
          (m-2-1) edge[<-,dashed,labelb={}] (m-2-2) %
          (m-1-2) edge[<-,dashed,labelr={}] (m-2-2) %
          (m-1-2) edge[labela={}] (m-1-3) %
          (m-2-2) edge[dashed,labelb={}] (m-2-3) %
          (m-1-3) edge[<-,labelr={}] (m-2-3) %
          (m-2-1) edge[labell={}] (m-3-1) %
          (m-3-1) edge[<-,labelb={}] (m-3-2) %
          (m-2-2) edge[dashed,labelr={}] (m-3-2) %
          (m-3-2) edge[labelb={}] (m-3-3) %
          (m-2-3) edge[labelr={}] (m-3-3) %
        }
    \]
  \end{itemize}
  More concretely, an element of $R'₁(c₁)$ is a pair of transitions
   at $c₁$ with related sources and targets.
   The morphism $R → R'$ is obtained by the composite $R → \overline{R} → R'$,
   where the last morphism exploits the definition of $R'₁$ as a limit. 
    It is straightforward to check that $R'$ is a substitution-closed
    simulation. Moreover, it is symmetric (even at the level of transitions), so
    it is a bisimulation.
\end{proof}
\end{full}

\section{Conclusion}\label{s:conclu}
We have introduced the notion of Howe context, in which we have defined
\transitionmonoids, an abstract notion of labelled transition system
whose states feature some sort of substitution.  For them, we have
introduced an abstract variant of applicative bisimilarity called
substitution-closed bisimilarity.

Furthermore, we have introduced operational semantics signatures as a
device for specifying syntax with variable binding and operational
semantics.  We have finally shown that if the dynamic part of an
operational semantics signature preserves functional bisimulations,
then subs\-ti\-tu\-tion-closed bisimilarity on the generated
\transitionmonoid is a congruence.

This all follows the pattern of our previous work~\cite{BHL}, but
simplifying the framework and relaxing some hypotheses, as explained
in the introduction.

We hope these simplifications pave the way for more abstract results
in the same vein.  To start with, we would like to generalise our
approach. Indeed, it fails to directly account for some important
applications of Howe's method, notably to
PCF~\cite{DBLP:journals/tcs/Gordon99}, algebraic
effects~\cite{DBLP:conf/lics/LagoGL17}, and higher-order process
calculi~\cite{DBLP:conf/concur/LengletS15}.  Furthermore, methods
similar to Howe's have been used for purposes other than congruence of
applicative
bisimilarity~\cite{Pitts:howe,DBLP:conf/esop/LagoG19,DBLP:journals/mscs/Goubault-LarrecqLN08}:
it might be useful to design abstract versions of such results using
our methods.

\bibliographystyle{alphaurl}
\bibliography{bib}

\end{document}
